\documentclass[journal]{IEEEtran}
\usepackage{xcolor}
\usepackage{multicol}
\usepackage{cuted}
\usepackage{cite}
\usepackage[english]{babel}
\usepackage{algorithm}
\usepackage[utf8]{inputenc}
\usepackage{bbm}
\usepackage[noend]{algpseudocode}
\usepackage{amsmath,amssymb,amsfonts}
\usepackage{verbatim}
\usepackage{amsmath}
\usepackage{amsmath,amssymb,amsthm}
\usepackage{array,booktabs}
\usepackage{multirow}
\usepackage{graphicx}
\usepackage{textcomp}
\usepackage{stfloats}
\usepackage{multicol}
\usepackage{amsmath,bm}
\usepackage{soul,xcolor}
\usepackage{comment}
\usepackage{multirow}
\usepackage{color, colortbl}
\usepackage{hyperref}
\usepackage{balance}
\usepackage{nomencl}
\usepackage{soul}
\usepackage{nomencl}
\usepackage{bbm}
\usepackage{bm}
\addto\captionsenglish{\renewcommand{\figurename}{Fig.}}
\usepackage{caption}
\usepackage{graphicx}
\usepackage{caption}
\usepackage{subcaption}
\usepackage{mwe}
\usepackage{amsmath}

\usepackage{tikz}

\newtheorem{definition}{{Definition}}
\newtheorem{proposition}{\bf{Proposition}}
\newtheorem{theorem}{\bf{Theorem}}
\newtheorem{assumption}{Assumption}

\newtheorem{remark}{Remark}

\usepackage{environ} 
 
\NewEnviron{SMALLM}{%
    \scalebox{0.5}{$\BODY$} 
} 

\usepackage{multicol}
\usepackage{tikz}
\usetikzlibrary{calc}
\usetikzlibrary{arrows}
\usepackage{float}
\usepackage{pgfplots}
\pgfplotsset{compat=1.16}
\usetikzlibrary{external}
\usepackage[export]{adjustbox}
\usepackage{textcomp}

\usepackage{pgfplotstable}

\pgfplotsset{compat=1.8, 
xticklabel style={/pgf/number format/fixed}}

\usepackage{helvet}
\usepackage[eulergreek]{sansmath}
\pgfplotsset{
  tick label style = {font=\Large},
  every axis label = {font=\sansmath\sffamily},
  legend style = {font=\Large},
  label style = {font=\sansmath\sffamily}
}

\begin{document}

\title{Residual-Based Detection of Attacks in Cyber-Physical Inverter-Based Microgrids}

\author{Andres Intriago,~\IEEEmembership{Student~Member,~IEEE}, Francesco Liberati,  {Nikos D. Hatziargyriou},~\IEEEmembership{Life Fellow, IEEE}, Charalambos Konstantinou,~\IEEEmembership{Senior~Member,~IEEE}
\thanks{Andres Intriago and Charalambos Konstantinou are with the Computer, Electrical and Mathematical Sciences and Engineering (CEMSE) Division, King Abdullah University of Science and Technology (KAUST), Thuwal 23955-6900, Saudi Arabia (e-mail: firstname.lastname@kaust.edu.sa).}
\thanks{Francesco Liberati is with the Department of Computer, Control and Management Engineering “Antonio Ruberti” (DIAG), University of Rome “La Sapienza”, Via Ariosto, 25, 00185 Rome, Italy.}
\thanks{Nikos D. Hatziargyriou is with the School of Electrical and Computer Engineering, National Technical University of Athens, 15773 Athens, Greece.}
}

\IEEEaftertitletext{\vspace{-2.5\baselineskip}}

\maketitle
\begin{abstract}

This paper discusses the challenges faced by cyber-physical microgrids (MGs) due to the inclusion of information and communication technologies in their already complex, multi-layered systems. The work identifies a research gap in modeling and analyzing stealthy intermittent integrity attacks in MGs, which are designed to maximize damage and cancel secondary control objectives. To address this, the paper proposes a nonlinear residual-based observer approach to detect and mitigate such attacks. In order to ensure a stable operation of the MG, the formulation then incorporates stability constraints along with the detection observer. The proposed design is validated through case studies on a MG benchmark with four distributed generators, demonstrating its effectiveness in detecting attacks while satisfying network and stability constraints. 
\end{abstract}


\IEEEpeerreviewmaketitle

\vspace{-2mm}
\section{Introduction}
\subsection{Background and Motivation}
Electric power systems are increasingly using more distributed energy resources (DERs)  \cite{zografopoulos2021}. These DERs are connected to the main grid through inverters, known as inverter-based resources (IBRs). 
IBRs can operate in grid-forming or grid-following mode. In grid-following mode, they cannot maintain system frequency and voltage when disconnected from the main grid. In contrast, in grid-forming mode, IBRs can maintain frequency and voltage within narrow limits, enabling an autonomous operation of the microgrids (MGs). Grid-forming converters can act as an ideal AC voltage source at the point of common coupling (PCC). They can minimize load interruptions during adverse events and provide black-start services during wide-area blackouts \cite{Jain_Blackstart2}.

In MG hierarchical control \cite{bidram2013, bidram2014}, primary control utilizes droop-based control and virtual impedance control to maintain MG frequency and voltage within acceptable limits, relying on millisecond cycle times. Secondary control compensates for residual deviations in frequency and voltage and operates with response times of milliseconds to seconds. Tertiary control coordinates the MG with the main grid, setting power output setpoints for each distributed generator (DG) based on electricity markets and tariffs. It operates with cycle times of minutes to hours. Automating primary control with local DG data is possible, but secondary and tertiary control require communication infrastructure, adding cost and complexity to the cyber-physical system (CPS). Cyber-attacks can target MG communication controls, making robust and resilient control policies necessary to mitigate disruption of the MG operating state \cite{Vasilakis2020,KURUVILA2021107150}.

\vspace{-2mm}
\subsection{Review of Previous Work}

Over the past ten years, a significant number of research articles have been written on detecting MG attacks, and various innovative initiatives have been carried out to showcase the advantages of IBR MGs in trial runs. After analyzing the landscape in the literature, we have identified three pillars that should be employed when reviewing this prior research. 

The first criterion involves the model utilized for the attack modeling due to the attack nature and based on which previous work can be broadly classified into two categories:

\textit{1.1)} Work focusing on denial-of-service (DoS) \cite{1430652}. Interference with data transmitted by sensors, actuators, and control systems can lead to DoS. Attackers may flood communication channels with random data or jamming them, preventing devices from transmitting and/or receiving information. This has been documented in several studies \cite{9206282, 8964591, 8673781}. 

\textit{1.2)} Work focusing on false data injection (deception/integrity) \cite{7001709}. Malicious adversaries can perform deception attacks by accessing and replacing real measurements in a system with fake information to cause harm. Such attacks are challenging to identify as attackers can quickly access data, but data veracity is hard to confirm \cite{Farraj2017, 10108042}. Intermittent injection of false data is an overlooked problem that can significantly affect the stealthiness of an attack and the energy consumed by the attacker. A self-generated approach in \cite{zhang2020false} generates specific false data to achieve perfect stealthiness. Notably, recent studies have shown that deception attacks are more challenging to identify than DoS attacks \cite{zografopoulos2021}. 

DoS and intermittent attacks can compromise information integrity and result in data loss. Intermittent attacks involve piece-wise signals and optimal energy management achieved through scheduling DoS attack time instants \cite{zhang2015optimal, ZhangK2022, Zhang2023}. Zero-dynamics, replay, jamming, and covert attack models are used to simulate stealthy and intermittent integrity attacks, becoming difficult to detect by anomaly detectors \cite{Keliris2013, 8662680,  7524930, 7172466}. In replay attacks, adversaries record data from compromised sensors during a steady state and replay it while concealing the attack effects \cite{6587520}.

The second criterion involves the approach adopted for attack detection. Based on this criterion, previous work can be broadly classified into two categories, consisting of model-based and data-driven techniques: 

\textit{2.1)} Work focusing on model-based approaches \cite{8579571, chen2017optimal, chen2017secure, 8680640, bernard2019observer, Zemouche2005, pajic2014robustness, moghadam2018resilient, lu2019observer, Rana2018, Aluko2021, Luo2018, Sahoo_1, Rajamani_2016, Rajamani_1998_tac, 9763858, Sahoo_2, corradini2017sliding}, which rely on a mathematical model of a system. Estimation methods are often employed to observe the state of the system and create analytical redundancy for detecting attacks. Common techniques include Kalman filtering methods, such as the Kalman filter (KF) \cite{8579571, chen2017optimal, chen2017secure} and the unscented KF (UKF)\cite{8680640}, as well as observer-based methods \cite{bernard2019observer, Zemouche2005, pajic2014robustness, moghadam2018resilient,  lu2019observer, Rana2018, Aluko2021, Luo2018, Sahoo_1, Rajamani_2016, Rajamani_1998_tac, 9763858} and sliding mode observation \cite{Sahoo_2, corradini2017sliding}.

In terms of KF methods, in \cite{chen2017optimal}, an optimal attack strategy is developed for a linear CPS using linear quadratic Gaussian (LQG) and KF to control and estimate the state. A distributed KF approach is introduced in \cite{chen2017secure} to address CPSs with bandwidth constraints. A resilient state estimator using $L_0$-norm-based state observer is presented in \cite{pajic2014robustness}. Distributed observer-based approach is proposed in \cite{moghadam2018resilient} for the mitigation of cyber-attacks on CPS sensors and actuators. An event-triggered and observer-based control frame for detecting DoS attacks in CPSs is developed in \cite{lu2019observer}. Finally, \cite{corradini2017sliding} introduces an attack isolation method based on sliding-mode technique for linear systems with unknown inputs.

Regarding observer-based detection techniques for MGs, Rana \textit{et al.} \cite{Rana2018} presented a KF-based observer to recover states 
corrupted by attacks or random noise in a MG. In \cite{Aluko2021}, an unknown input observer within the secondary frequency controller is developed to maintain MG stability under cyber-attacks. An observer-based algorithm that uses real-time phasor measurements to detect and mitigate attacks is presented in \cite{Luo2018}. Sahoo \textit{et al.} \cite{Sahoo_1} considered an adaptive and distributed nonlinear observer to detect, reconstruct, and mitigate false data injection attacks in MGs. 
Other studies focused on designing observers for Lipschitz nonlinear systems \cite{Rajamani_2016, Rajamani_1998_tac, 9763858}, or used sliding mode nonlinear observers \cite{Sahoo_2}.

\textit{2.2)} Work focusing on model-free, data-driven detection approaches \cite{al2020ensemble, huda2017defending, nguyen2011efficient, li2019online, hussain2020deep, Zografopoulos2022, forti2017distributed, waghmare2017data, wang2020multi}.
Data-driven methods use historical system data to detect attacks when a mathematical model is unavailable or contains parameter uncertainties. They commonly employ intelligent methods such as neural networks (NNs) \cite{al2020ensemble}, support vector machines (SVMs) \cite{huda2017defending}, and naive Bayesian classifiers \cite{nguyen2011efficient}. However, the accuracy of these methods depends on the quality of historical information, data noise, and various operating and attack conditions. In the following, we survey representative data-based works. 

In the context of intelligent techniques, Li \textit{et al.} proposed an NN-based approach that combines a physical model and a generative adversarial network to detect deviations in measurements \cite{li2019online}. In \cite{hussain2020deep}, a deep learning method using convolutional NNs is suggested to detect DoS in cellular networks caused by flooding, signaling, and silent calls. In \cite{Zografopoulos2022}, a method is proposed for constructing an attack detector in MGs by identifying the stable kernel representation in the absence of attacks. In \cite{forti2017distributed}, the authors introduced a Bayesian-based approach to attack detection that uses the hybrid Bernoulli random set method to jointly estimate states and detect attacks. In \cite{waghmare2017data}, a two-stage approach is presented for false data detection that reduces data dimensionality and detects malicious activities via an SVM classifier. In \cite{wang2020multi}, the authors developed a multi-agent supervised attack detection that employs SVM with a decision tree for each agent, with final decisions made through consensus among all agents.

The third criterion involves the system considerations of the modeling approach and subsequently the realistic operating conditions. Based on this criterion, previous work can be broadly classified into {two} categories:

\textit{3.1)} Work focusing on considering MG network models and corresponding operation constraints. Specifically, optimal power flow (OPF) methods have been widely applied for MG operation.
In \cite{levron2013optimal}, a grid-connected MG OPF method is presented that optimizes energy storage while treating renewable energy sources (RES) and loads as fixed power injections. In \cite{olivares2014centralized}, the authors propose a two-stage optimization method for energy management of islanded MGs involving a mixed integer programming unit commitment problem followed by an unbalanced three-phase OPF problem. Building on this work, \cite{lara2018robust} extends the first-stage unit commitment to a robust optimization problem to address uncertainties in RES. For distributed MGs, \cite{shi2014distributed} introduces a method that relaxes power flow constraints to reform the non-convex OPF model. The authors then transform the centralized optimization into a distributed one using predictor-corrector proximal multipliers. 

\textit{3.2)} Work focusing on stability-constrained network modeling in MGs. Stability is a significant concern in MGs, in addition to steady-state security constraints. This is mainly due to the intermittent nature and fluctuating output of RES, and the reduction of system inertia by DGs, which are interfaced by inverters, leading to lower resistance to disturbances in IBR-interfaced MGs. However, few OPF models in previous works have integrated stability constraints. A single machine equivalent method is used to obtain a bus equivalent rotor angular trajectory as a stability constraint in a trajectory-based transient stability-constrained OPF \cite{5357533}. Conventional power systems have studied voltage stability-constrained OPFs \cite{8279490, 8439024}. In MGs, a nonlinear optimization-based voltage stability-constrained OPF is created in \cite{8960513} 
to enhance multi-MG voltage stability. Furthermore, a small-signal stability-constrained OPF problem has been studied extensively in traditional power systems \cite{5593195}, and recently in MGs \cite{pullaguram2021small}.

\vspace{-3mm}
\subsection{Paper Motivation and Contributions}

The challenges towards the security of cyber-physical MGs have increased with the inclusion of more resources and services to the already complex, multi-layered MGs with information and communication technologies. Based on the review of previous work, we have identified a significant research gap that drives the motivation behind this paper: no previous work has modelled and analyzed stealthy intermittent integrity attacks, adopting an observer-based approach for detection and mitigation, and employing within MG modeling both network and stability operation constraints. In the context of filling this knowledge gap, the following contributions are achieved:

\textit{{(1)}} We analyze and formulate the state space dynamics of a generic DG in a MG  system and consider integrity attacks performed in an intermittent mode, i.e., intermittent integrity attacks. In contrast to the approach used for conventional continuous integrity attacks, intermittent integrity attacks follow a two-step process that involves creating a covert attack model and scheduling the activation and pause times for the attack. The attacker goal is to maximize potential damage in order to cancel secondary control objectives in removing deviations in both MG global frequency and local voltage.

\textit{{(2)}} We design a nonlinear observer to detect such stealthy attacks, and to mitigate their impact by addressing the corrupted signals with their corresponding estimation from the observer. Rather than designing the observer based on nonlinear change of coordinates and the use of ``canonical forms" (or ``normal forms", see, e.g., \cite{bernard2019observer}), the proposed design is based on simpler Luenberger observers, properly extended to nonlinear systems. Finding a globally defined change of coordinates to put the system into a normal form is a particularly challenging task for the system at study since it is a multiple-output system, with more outputs than inputs, and with a high number of state variables (as explained, e.g., in \cite[Chapter 4]{bernard2019observer}). Our work instead is a contribution in the lines of the papers presented in \cite{Rajamani_2016, Rajamani_1998_tac}, and more recently, in \cite{9763858}, on Luenberger-like observers for nonlinear systems. Differently from the above works, in the present paper, a nonlinear gain is considered in the design of the observer, which is a novelty that allows for greater flexibility in the proof of the convergence of the estimate, achieved via Lyapunov arguments. Also, a nonlinear gain allows us to better tailor the design of the observer to the specific dynamics of MG, reducing the number of conservative boundings in the proof (Theorem \ref{theorem_observer}).

\textit{{(3)}} We present the OPF problem for the considered MG and then identify the worst-case cyber-attack in the formulation of a bi-level optimization problem. The adversary adopts intermittent integrity attacks while satisfying the OPF of the system operation. Operators, to ensure a stable operation of the MG, react by incorporating both stability constraints as well as the design of the residual-based detection observer. These conditions ensure not only the stable and optimal operation of the MG but also guarantee the detection of cyber-attacks.

\textit{{(4)}} Case studies on a MG benchmark serve a twofold purpose: \textit{a)} validating that the proposed design successfully encapsulates the stealthy nature of attack design in terms of impact assessment, and \textit{b)}  demonstrating that the proposed observer-based method achieves reasonable detection accuracy with respect to other observer-based approaches while satisfying the network and stability constraints of the system.

\vspace{-3mm}
\subsection{Paper Organization}
The remainder of this paper is structured as follows. Section \ref{s:structure} presents the modeling of an islanded MG. Section \ref{s:attackformulation} presents the attack model formulation, with a discussion on the scheduling interval and stealthiness of the intermittent integrity attack. Section \ref{s:methodology} presents the formulation of the proposed observer-based mitigation method as well as the attack detectability analysis. Section \ref{s:bilevel} presents the stability-constrained mitigation formulation. Section \ref{s:result} presents simulation results, and Section \ref{s:conclusion} concludes the paper.

\vspace{-3mm}
\subsection{Notation}
$I_n$ and $0_n$ denote identity and zero matrices of dimension $n\times n$, respectively. $n_{\boldsymbol{v}}$ denotes the dimension of the generic vector $\boldsymbol{v}$. $[A,B]$ denotes the horizontal concatenation of matrices $A$ and $B$. $\operatorname{blockdiag}(A_1,...,A_n)$ denotes the block diagonal matrix having as diagonal blocks the matrices $A_1,...,A_n$. $A_{i,j}$ denotes the entry at row $i$ and column $j$ of matrix $A$. For a linear map $A: \mathcal{X} \longrightarrow \mathcal{Y}$, we define $\operatorname{ker}(A) \triangleq\{x \in \mathcal{X} \mid A x=0\}$ and $\operatorname{Im}(A) \triangleq\{A x \mid x \in \mathcal{X}\}$. Vector $e_i$ denotes the $i_{th}$ vector of the canonical base. $\operatorname{convh}(x, y)$ denotes the convex hull of the two vectors $\boldsymbol{x}$ and $\boldsymbol{y}$. $\|\boldsymbol{v}\|$ denotes the Euclidean norm ($l^2$) of vector $\boldsymbol{v}$, i.e., $\| \boldsymbol{v} \|  \triangleq  \sqrt{\sum_{i}v_i^2}$. $|a|$ denotes the absolute value of number $a$, or the cardinality of a set, depending on the context. 
\vspace{-2mm}
\section{Cyber-Physical Structure of\\ Islanded Microgrid}\label{s:structure}
\vspace{-1mm}
In this part, the mathematical model of the DG-based MG system, also instrumental for the design of the proposed nonlinear observer, is presented. The single-line diagram in Fig. \ref{fig1a} shows the islanded MG considered in this paper.

\begin{figure}[t]
    \centering
     \includegraphics[width=0.95\linewidth]{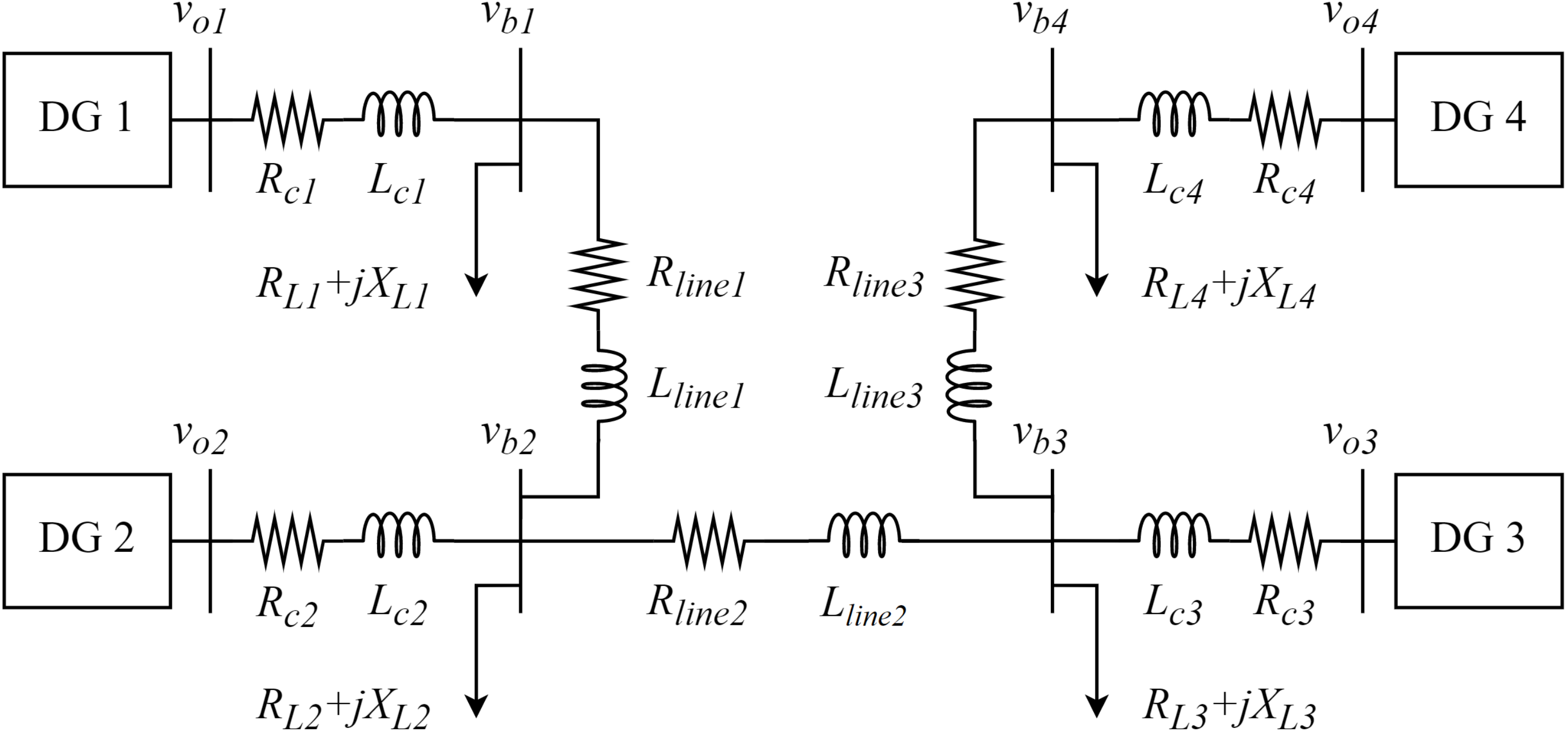}
     \caption{MG test system considered in this study \cite{bidram2013}.}
     \vspace{-3mm}
     \label{fig1a}
\end{figure}

\vspace{-2mm}
\subsection{Cyber Model of Islanded Microgrid}
An islanded or autonomous MG can be considered as a CPS with DGs, loads, and a communication network for monitoring and control. 
A directed graph (digraph) is used to describe the communication network of a multi-agent cooperative system. The digraph is modelled as 
$\mathcal{G}_{r}=(\mathrm{V}_{\mathcal{G}}, \mathrm{E}_{\mathcal{G}}, \mathrm{A}_{\mathcal{G}})$, where $\mathrm{V}_{\mathcal{G}}=\{\mathrm{v}_{1}, \mathrm{v}_{2}, \dots, \mathrm{v}_{N}\}$ is a nonempty finite collection of $N$ nodes, $\mathrm{E}_{\mathcal{G}} \subseteq \{\mathrm{V}_{\mathcal{G}}\times \mathrm{V}_{\mathcal{G}}\}$ is a set of edges, and $\mathrm{A}_{\mathcal{G}}=[\mathrm{a}_{ij}] \in \mathbb{R}^{N\times N}$  is the associated adjacency matrix. In the adopted model, the DGs are the nodes of the digraph, whereas the communication links are the edges of the communication network. For this study, it is assumed that the DGs communicate with each other through the digraph presented in the case study of \cite{bidram2013}, and the term $ \mathrm{A}_{\mathcal{G}}$ remains constant since the digraph is assumed to be invariant in time.

\vspace{-2mm}
\subsection{Physical Model of Islanded Microgrid}
{The AC microgrid has $\mathcal{N}=\{1,2, \ldots, n\}$ set of buses on the power distribution network, $\mathcal{L}=\{1,2, \ldots, l\} \subseteq \mathcal{N} \times \mathcal{N}$ set of distribution lines, and $\mathcal{G}=\left\{1,2, \ldots, n^{\mathrm{g}}\right\}$ set of inverters at every DG.} Each DG is considered to incorporate a primary and a secondary controller. {The droop control for the active ($P_i$) and the reactive power ($Q_i$) are described in \cite{bidram2013}.}

The synchronization goals for the secondary voltage and frequency controllers are to set the control input $V_{ni}$ such that the $d$-axis component of the voltage after the LC filter, $v_{odi}$, reaches the reference value for the secondary controller, $v_{ref}$, and to set the control input $\omega_{ni}$ to ensure that the frequency at the $i_{th}$ inverter, $\omega_{i}$, reaches  the reference frequency for the secondary controller,  $\omega_{ref}$. The voltage magnitude of the $i_{th}$ DG is $\|v_{o}\|= \sqrt{v_{odi}^2 + v_{oqi}^2}$. {The voltage-controlled voltage source inverter (VCVSI) DG is shown in Fig. \ref{fig-vcvsi}. The functionalities of different blocks are explained briefly below.

\begin{figure}[t]
    \centering
     \includegraphics[width=0.95\linewidth]{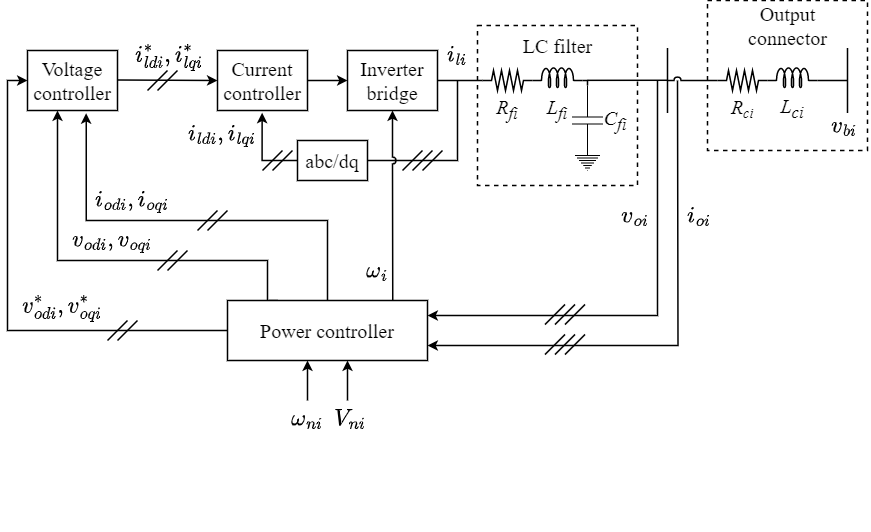}
     \vspace{-9mm}
     \caption{Voltage-controlled voltage source inverter (VCVSI) DG.}
     \vspace{-3mm}
     \label{fig-vcvsi}
\end{figure}

\noindent \underline{\textit{$\delta_{i}$ calculator:}}
{The angle of the $i_{th}$ inverter-based DGs with respect to a common reference, typically the inverter at bus 1 (i.e., $\omega_{com}=\omega_1$ ) is $\delta_{i}$, and satisfies $\dot{\delta}_i=\omega_i-\omega_{com}$. {The purpose of this angle is to transform the periodic and balanced three-phase signal $abc$ to a $dq$ reference frame, while ensuring that $v_{oqi}$ reaches zero.}} 

\noindent \underline{\textit{Power controller:}}
This block includes the droop controller and a low-pass filter to extract the fundamental components of $P_i$ and $Q_i$. The outputs of this block are $v^*_{odi}$, $v^*_{oqi}$, and $\omega_i$ \cite{bidram2013}. 

\noindent \underline{\textit{Voltage and current controller:}}
The dynamics of the voltage and current controllers can be found in \cite{bidram2013}. In both controllers, auxiliary state variables, $\phi_{di}$, $\phi_{qi}$ and $\gamma_{di}$, $\gamma_{qi}$, respectively, are introduced to formulate the complete state-space model, where $\Dot{\phi}_{di}= v^*_{odi} - v_{odi}$, $\Dot{\phi}_{qi}= v^*_{oqi} - v_{oqi}$, $\Dot{\gamma}_{di}= i^*_{ldi} - i_{ldi}$, and $\Dot{\gamma}_{qi}= i^*_{lqi} - i_{lqi}$.

\noindent \underline{\textit{LC filter and output connector:}}
The $d$-axis and $q$-axis equations for the LC filter and output connector are provided in \cite{bidram2013}, and are considered in the state-space model. 

{As for the distributed frequency and voltage controller design, the details can be found in \cite{bidram2014}. Overall, the objective of frequency control for VCVSIs is to achieve synchronization 
with the nominal frequency. In addition, it should distribute the output $P_i$ of VCVSIs based on their $P_i$ ratings. By taking the derivative of the frequency droop characteristic ($\omega_i=\omega_{n i}-m_{P i} P_i$), the rate of change of the primary control frequency reference $\omega_{ni}$ is expressed as $\dot{\omega}_{n i}=\dot{\omega}_i+m_{P i} \dot{P}_i=v_{f i}$, where $v_{fi}$ represents the auxiliary control input defined in (30) of \cite{bidram2014}. $\dot{\omega}_{n i}$ is utilized to calculate the control input $\omega_{ni}$ from $v_{fi}$, which involves integration of $v_{fi}$ over time. Similarly, the distributed voltage control mechanism selects appropriate control inputs $V_{ni}$ in the reactive droop ($|v_{o}|^*=V_{n i}-n_{Q i} Q_i$) to ensure synchronization of the voltage magnitudes $|v_{o}|$ of VCVSIs with the reference voltage $v_{{ref}}$. Additionally, it distributes the $Q_i$ of VCVSIs based on their $Q_i$ ratings. Synchronization of the voltage magnitudes $|v_{o}|$ is equivalent to synchronization of the direct term of output voltages $v_{odi}$. By differentiating the voltage droop characteristic, the rate of change of $V_{ni}$ is expressed as $\dot{V}_{n i}=\dot{v}_{o d i}+n_{Q i} \dot{Q}_i=v_{v i}$, where $v_{vi}$ denotes the auxiliary control input defined in (48) of \cite{bidram2014}. $\dot{V}_{n i}$ is utilized to calculate the control input $V_{ni}$ from $v_{vi}$ by integration over time. In our work, $\omega_{ni}$ and $V_{ni}$ will be considered as the outputs of the secondary controller, to follow the same structure as previous works, with $m_{P i}$  and $n_{Q i}$ being the active and reactive power droop coefficients that are selected based on
the $P_i$ and $Q_i$ ratings of each VCVSI.}

\vspace{-3mm}
\subsection{Full State Space Model}\label{ss:fullstatemodel}
The state space dynamics of the generic $i_{th}$ DG 
{over time $t$} are obtained by gathering and manipulating {Eqs. (1) - (20) of \cite{bidram2013}.}  The complete model can be written in matrix form as:
{
\begin{equation}\label{complete}
 \ W : \left\{
\begin{array}{ll}
       \dot{\boldsymbol{x}} = A\boldsymbol{x} + \boldsymbol{f(x)} + B\boldsymbol{u}\\
        \boldsymbol{y} = C\boldsymbol{x}
\end{array} 
\right. \
\end{equation}
}{where $\boldsymbol{x} \in \mathbb{R}^{n_x}$ is the state vector, $\boldsymbol{u} \in \mathbb{R}^{n_u}$ is the input vector generated by the power controller. The vector function $\boldsymbol{f(x)}$ $\in$ $\mathbb{R}^{n_x}$ captures the nonlinearity of the system 
and $\boldsymbol{y}$ $ \in \mathbb{R}^{n_y}$ is the vector of sensor measurements at the DG.}
The state vector {$\boldsymbol{x}$ with $n_x = 15$  in \eqref{complete}  is defined as}:
\begin{equation}\label{state_variables}
\begin{aligned}
\boldsymbol{x} = &[ P_{i}, Q_{i}, \phi_{di}, \phi_{qi}, \gamma_{di}, \gamma_{qi}, i_{ldi}, i_{lqi}, v_{odi}, v_{oqi},  \\
& i_{odi}, i_{oqi}, \delta_i , \omega_{ni} , V_{ni}]^{\top}.
\end{aligned}
\end{equation}

The input vector 
{$\boldsymbol{u}$, with $n_u=9$, in \eqref{complete} where $u_2$ and $u_3$ describes the voltage at the $i_{th}$ bus DG measured after the output connector (Fig. \ref{fig-vcvsi}) in terms of $d$ and $q$ axes components, respectively. It also includes in $u_5$ and $u_8$ the setpoints for the secondary controller $\omega_{ref}$ and $v_{ref}$, while the remaining elements $u_4$, $u_6$, $u_7$, and $u_9$ represent the data collected from the set of neighbouring/adjacent DGs in the MG $j \neq i$, i.e., $\omega_j$, $P_j$, $v_{odj}$, and $Q_j$, respectively, where $c_{fi}$ and $c_{vi}$ are control gains, and $g_i$ the pinning gain. It is defined as}:
\begin{equation}\label{u_vector}
\begin{aligned}
\boldsymbol{u} = &[ \omega_{com}, v_{bdi}, v_{bqi}, 
c_{fi}\sum_{j\in N_{j \neq i}}\mathrm{a}_{ij}\omega_j, 
c_{fi}g_i\omega_{ref}, \\
& c_{fi}\sum_{j\in N_{j \neq i}}\mathrm{a}_{ij}m_{P_j}P_j, 
c_{vi}\sum_{j \in N_{j \neq i}}\mathrm{a}_{ij}v_{odj}, \\
& c_{vi}g_iv_{ref}, c_{vi}g_iv_{ref}, 
c_{vi}\sum_{j \in N_{j \neq i}}\mathrm{a}_{ij}n_{Q_j}Q_j
]^{\top}.
\end{aligned}
\end{equation}
%

%
\vspace{-2mm}
{
\begin{remark}[Measured Variables]\label{RemarkMeasures}The DG measured variables are assumed to be the generated active and reactive powers, i.e., the state variables $x_1$ and $x_2$ (as mentioned, e.g., in \cite[see after Fig. 1]{bidram2013}), the inverter bridge current $i_l$, the output voltage $v_{oi}$, and the output current $i_{oi}$, and their direct and quadrature components, i.e., respectively, the state variables $x_{7}$, $x_{8}$, $x_{9}$, $x_{10}$, $x_{11}$, and $x_{12}$ (see, e.g., \cite{bidram2013} and \cite{Sahoo_1}); notice in fact from Fig. \ref{fig-vcvsi} that the measurement of those 
variables are available and fed-back to the current, 
voltage, and 
power controllers. The secondary control inputs $\omega_{ni}$ and $V_{ni}$ are known, since they are the inputs to the power controller (they are computed according to, respectively, (20) and (46) in \cite{bidram2013}). Finally, the DG operating frequency, $\omega_{i}$ is also measured (since it has to be communicated to the neighbour DGs for distributed secondary control purpose - see, e.g., \cite[eq. (30)]{bidram2014}. In conclusion, the state variables assumed measurable are: $x_{1}$, $x_{2}$, $x_{7}$, $x_{8}$, $x_{9}$, $x_{10}$, $x_{11}$, $x_{12}$, $x_{14}$, and $x_{15}$.
\end{remark}
}

The fifteen differential equations in \eqref{complete}, and the resulting matrices/vectors $A$, $B$, $C$, and $\boldsymbol{f(x)}$ are detailed in the following subsection \ref{ss:detailed-ss-model}. The model of \eqref{complete} is used for the design of the observer in Section \eqref{s:methodology}.

\vspace{-3mm}
\subsection{Detailed DG State Space Model} \label{ss:detailed-ss-model}
The state vector is defined in \eqref{state_variables}. The complete state space model of the generic $i_{th}$ DG is given in \eqref{complete} in matrix form. In scalar form, the fifteen differential state equations can be obtained by manipulating {(1) -- (20) of \cite{bidram2013}} as follows:
\begin{equation}\label{dotx1}
\small
\begin{aligned}
    \dot{x}_{1} = -\omega_{ci}x_1 + \omega_{ci}x_9x_{11} + \omega_{ci}x_{10}x_{12}
\end{aligned}
\end{equation}
\begin{equation}
\small
\begin{aligned}
    \dot{x}_{2} = -\omega_{ci}x_2 - \omega_{ci}x_9x_{12} + \omega_{ci}x_{10}x_{11}
\end{aligned}
\end{equation}
\begin{equation}
\small
\begin{aligned}
    \dot{x}_{3} = - n_{Qi}x_2 - x_9 + x_{15}
\end{aligned}
\end{equation}
\begin{equation}
\small
\begin{aligned}
    \dot{x}_{4} = -x_{10}
\end{aligned}
\end{equation}
\begin{equation}
\small
\begin{aligned}
    \dot{x}_{5} =& - K_{PVi}n_{Qi}x_2 + K_{IVi}x_3 - x_7 - K_{PVi}x_9 + \\ 
    & - \omega_bC_{fi}x_{10} + F_ix_{11} + K_{PVi}x_{15}
\end{aligned}
\end{equation}
\begin{equation}
\small
\begin{aligned}
    \dot{x}_{6} = K_{IVi}x_4 - x_8 + \omega_bC_{fi}x_{9} -  K_{PVi}x_{10} + F_ix_{12} 
\end{aligned}
\end{equation}
\begin{equation}
\small
\begin{aligned}
    \dot{x}_{7} = &\!-\frac{K_{PCi}K_{PVi}}{L_{fi}}n_{Qi}x_2 \!+\! \frac{K_{PCi}K_{IVi}}{L_{fi}}x_3 \!+\!\\
    & \!+\!\frac{K_{ICi}}{L_{fi}}x_5 \!-\! \Big( \frac{R_{fi}}{L_{fi}} \!+\! \frac{K_{PCi}}{L_{fi}} \Big)x_{7}  \!-\! \omega_bx_8 +\\
    &\!-\! \Big(\frac{1}{L_{fi}} \!+\! \frac{K_{PCi}K_{PVi}}{L_{fi}}\Big)x_{9} \! -\! \omega_b\frac{K_{PCi}}{L_{fi}}C_{fi}x_{10}\!+\\
    &+\! \frac{K_{PC_i}}{L_{fi}}F_ix_{11} \!+\! \frac{K_{PCi}K_{PVi}}{L_{fi}}x_{15} +\\
    & - m_{P_i}x_1x_8 + x_8x_{14} 
\end{aligned}
\end{equation}
\begin{equation}
\small
\begin{aligned}
    \dot{x}_{8} = &\frac{K_{PCi}}{L_{fi}}K_{IVi}x_4 + \frac{K_{ICi}}{L_{fi}}x_6 +  \omega_bx_7  +\\
    &- \Big( \frac{R_{fi}}{L_{fi}} + \frac{K_{PCi}}{L_{fi}} \Big)x_{8} + \omega_b\frac{K_{PCi}}{L_{fi}}C_{fi}x_9 +\\
    &- \Big(\frac{K_{PCi}K_{PVi}}{L_{fi}} + \frac{1}{L_{fi}}\Big)x_{10} + \frac{K_{PCi}}{L_{fi}}F_ix_{12} + \\
    & + m_{P_i}x_1x_7 - x_7x_{14} 
\end{aligned}
\end{equation}
\begin{equation}
\small
\begin{aligned}
    \dot{x}_{9} =  \frac{1}{C_{fi}}x_{7} - \frac{1}{C_{fi}}x_{11} - m_{P_i}x_1x_{10} + x_{10}x_{14} 
\end{aligned}
\end{equation}
\begin{equation}
\small
\begin{aligned}
    \dot{x}_{10} = \frac{1}{C_{fi}}x_{8} - \frac{1}{C_{fi}}x_{12} + m_{P_i}x_1x_{9} - x_{9}x_{14} 
\end{aligned}
\end{equation}
\begin{equation}
\small
\begin{aligned}
    \dot{x}_{11} = \frac{1}{L_{ci}}x_9 - \frac{R_{ci}}{L_{ci}}x_{11} - \frac{1}{L_{ci}}v_{bdi} - m_{P_i}x_1x_{12} + x_{12}x_{14} 
\end{aligned}
\end{equation}
\begin{equation}
\small
\begin{aligned}
    \dot{x}_{12} = \frac{1}{L_{ci}}x_{10} - \frac{R_{ci}}{L_{ci}}x_{12} - \frac{1}{L_{ci}}v_{bqi} + m_{P_i}x_1x_{11} - x_{11}x_{14} 
\end{aligned}
\end{equation}
\begin{equation}\label{dotx_13}
\small
\begin{aligned}
    \dot{x}_{13} = -m_{P_i}x_1 + x_{14} -\omega_{com}
\end{aligned}
\end{equation}
\begin{equation}
\small
\begin{aligned}
    \dot{x}_{14} = & \Big[m_{P_i}c_{fi}\Big(\sum_{j\in N_i}\mathrm{a}_{ij} + g_i\Big) -c_{fi}\Big(\sum_{j\in N_i}\mathrm{a}_{ij}m_{P_i}\Big)\Big]x_1 +\\
    &- c_{fi}\Big(\sum_{j\in N_i}\mathrm{a}_{ij} + g_i\Big)x_{14} +\\
    & + c_{fi}\sum_{j\in N_i}\mathrm{a}_{ij}\omega_j + c_{fi}g_i\omega_{ref} +\\
    & + c_{fi}\sum_{j\in N_i}\mathrm{a}_{ij}m_{P_j}P_j
\end{aligned}
\end{equation}
\begin{equation}\label{dotx15}
\small
\begin{aligned}
    \dot{x}_{15} =  & - c_{vi}n_{Q_i}\Big(\sum_{j\in N_i}\mathrm{a}_{ij}\Big)x_2 - c_{vi}\Big(\sum_j\mathrm{a}_{ij} + g_i\Big)x_{9} +\\
    & + c_{vi}\sum_j\mathrm{a}_{ij}v_{odj} + c_{vi}g_iv_{ref} +\\
    & + c_{vi}\sum_j\mathrm{a}_{ij}n_{Q_j}Q_j
\end{aligned}
\end{equation}

\normalsize

In the above equations, other parameters involve the nominal angular frequency of the MG, $\omega_b$, the cut-off frequency of the low-pass filter in the power controller, $\omega_{ci}$, gains of voltage and frequency control, $c_{vi}$ and $c_{fi}$, the current compensator $F_i$, and impedance characteristics of the LC filter and the output connector. The equations are compactly written in \eqref{complete} in matrix form, in which $A$, $\boldsymbol{f(x)}$, $B$, and $C$ can be inferred from inspection of the equations and are as follows.

The elements of matrix $A$ are: 
$A_{1,1} = A_{2,2} = -w_{ci}$, 
$A_{13,1} = -m_{P_i}$,
$A_{3,2}=-n_{Qi}$,
$A_{5,2} = - K_{PVi}n_{Qi}$,
$A_{5,3} = A_{6,4} = K_{IVi}$,
$A_{7,5} = A_{8,6} = \frac{K_{ICi}}{L_{fi}}$,
$A_{5,7} = A_{6,8} = A_{3,9} = A_{4,10} = -A_{13,14} = -A_{3,15} = -1$,
$A_{8,7} = -A_{7,8} = \omega_b$,
$A_{9,7} = A_{10,8} = -A_{9,11} = -A_{10,12} = \frac{1}{C_{fi}}$,
$A_{5,9} = A_{6,10} = -A_{5,15} = - K_{PVi}$,
$A_{6,9} = \omega_bC_{fi}$,
$A_{11,9} = A_{12,10} = \frac{1}{L_{ci}}$,
$A_{5,11} = A_{6,12} = F_i$,
$A_{7,11} = A_{8,12} = \frac{K_{PC_i}}{L_{fi}}F_i$,
$A_{11,11} = A_{12,12} = - \frac{R_{ci}}{L_{ci}}$,
$A_{14,1} = \Big[m_{P_i}c_{fi}\Big(\sum_{j\in N_i}\mathrm{a}_{ij} + g_i\Big) -c_{fi}\Big(\sum_{j\in N_i}\mathrm{a}_{ij}m_{P_i}\Big)\Big]$, 
$A_{7,7} = A_{8,8} = -( \frac{R_{fi}}{L_{fi}} \!+\! \frac{K_{PCi}}{L_{fi}} )$,
$A_{15,2} = - c_{vi}n_{Q_i}\Big(\sum_{j\in N_i}\mathrm{a}_{ij}\Big)$,
$A_{7,9} = A_{8,10} = - (\frac{1}{L_{fi}} + \frac{K_{PCi}K_{PVi}}{L_{fi}})$,
$A_{15,9} = - c_{vi}\Big(\sum_j\mathrm{a}_{ij} + g_i\Big)$,
$A_{7,2} = \frac{K_{PCi}K_{PVi}}{L_{fi}}n_{Qi}$,
$A_{7,3} = \frac{K_{PCi}K_{IVi}}{L_{fi}}$,
$A_{14,14} = - c_{fi}\Big(\sum_{j\in N_i}\mathrm{a}_{ij} + g_i\Big)$,
$A_{7,15} = \frac{K_{PCi}K_{PVi}}{L_{fi}}$,
$A_{7,10} = A_{8,9} = -\omega_b\frac{K_{PCi}}{L_{fi}}C_{fi}$, and all the other elements are equal to zero.

The elements of matrix $B$ are: $B_{11,2} = B_{12,3} = - \frac{1}{L_{ci}}$, $B_{13,1}= -1$, $B_{14,4}=B_{14,5}=B_{14,6}=B_{15,7}=B_{15,8}=B_{15,9}=1 $, and all the other elements are equal to zero.

{
In view of Remark \ref{RemarkMeasures}, matrix $C$ is as follows:\begin{align}\label{C}
    C~=~\operatorname{blockdiag}(I_2,[0_{6,4},I_6],[0_{2,1},I_2])
\end{align}  
}
\vspace{-1mm}
The elements of $\boldsymbol{f(x)}$ are: $f_{1} = \omega_{ci}x_9x_{11} + \omega_{ci}x_{10}x_{12}$, 
$f_2 = - \omega_{ci}x_9x_{12} + \omega_{ci}x_{10}x_{11}$, 
$f_7 = - m_{P_i}x_1x_8 + x_8x_{14}$ 
$f_8 = m_{P_i}x_1x_7 - x_7x_{14}$,
$f_9 = - m_{P_i}x_1x_{10} + x_{10}x_{14}$,
$f_{10} = m_{P_i}x_1x_{9} - x_{9}x_{14}$,
$f_{11} = - m_{P_i}x_1x_{12} + x_{12}x_{14}$,
$f_{12} = m_{P_i}x_1x_{11} - x_{11}x_{14}$, where the matrix $\boldsymbol{M}^{\mathrm{p}}$ is defined in  Appendix, and where each $m_{P_i}$ indicates the droop constant of the $i_{th}$ inverter, and all the other elements are zero.

\vspace{-2mm}
\section{Attack Model Formulation}\label{s:attackformulation}

The communication network connecting the participating DGs forming the MG can potentially expose the cyber-physical MG system to cyber-attacks \cite{Zografopoulos2022}. In literature, several procedures rely on the adversary's system knowledge and the skills to corrupt measurements from remote terminal units (RTUs). For example, in \cite{Farraj2017}, the consequences of false data on RTU measurements are examined at the distributed transient stability control schemes of a MG. Cyber-physical MGs include also smart inverters for renewable-based DERs grid-connection. Recent work has demonstrated that such firmware electronic devices can be tampered leading to system disruption and instabilities \cite{KURUVILA2021107150}.

\begin{figure}[t]
\centering
    \subfloat[]{
        \includegraphics[width=0.45\textwidth]{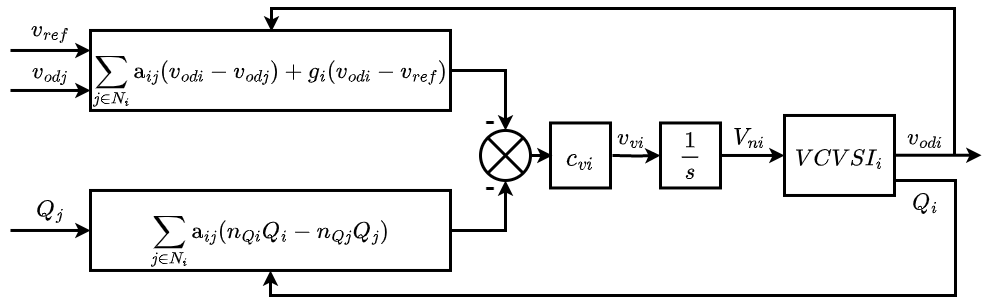}
        \label{fig:Attack_Voltage_Controller}
    } \\
    \subfloat[]{
        \includegraphics[width=0.45\textwidth]{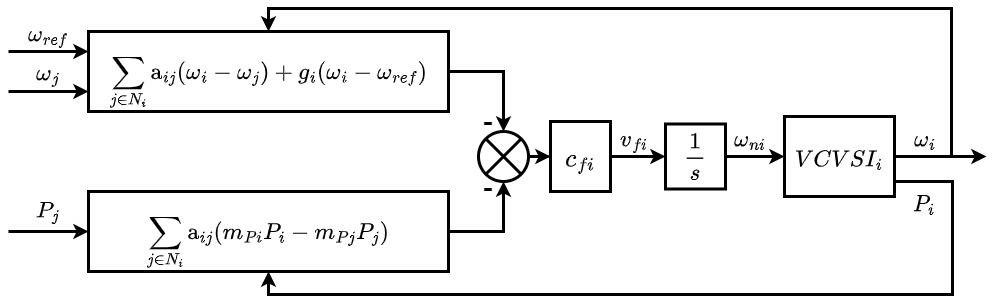}
        \label{fig:Attack_Frequency_Controller}
    } 
\vspace{-1mm}   
\caption[CR]{Distributed voltage (\subref{fig:Attack_Voltage_Controller}) and frequency (\subref{fig:Attack_Frequency_Controller}) controllers. } 
\vspace{-5mm}
\label{fig:Distributed_Voltage_Frequency_Disruption}
\end{figure}

In this paper, the threat model considers an adversary aiming to corrupt the operation of grid converters. The attack entry point could be direct, i.e., at the DG, or indirect, e.g., by modifying the firmware of the inverters. First, the attacker intends to undermine the communication links of the MG by infiltrating compromised data that can harm substantionally the voltage and frequency of all the inverters. The act of introducing malicious data to the inputs of each DG will affect the dynamics of the grid and hence, create instabilities in the MG. \emph{The aim of the attack is to maximize the effect of the disruption at the input control signals of the MG secondary controller, shown in Fig. \ref{fig:Distributed_Voltage_Frequency_Disruption}, and thus the states of the MG,  while remaining stealthy}.  
Specifically, the attacker would aim to modify control actions ${act}_{i} \in \mathcal{A}_i, \forall i \in  \mathcal{G}$, of inverter-based DG control under normal conditions  to actions ${act}^{'}_{i} \in \mathcal{A^\prime}_i, {act}^{'}_{i}\neq {act}_{i}, \forall i \in  \mathcal{G}$, that would harm DGs and/or MG/distribution and transmission systems. Following this modification, the production and consumption at every node are changed. Using this mechanism, the attacker will aim to coordinate actions of  compromised signals in the MG, i.e.,  $  \big\{ {act}^{'}_{it}: \prod_{i \in \mathcal{G}} \mathcal{S}_i \rightarrow \prod_{i \in \mathcal{G}} \mathcal{A^\prime}_i\big\}$, defined over a discrete  time horizon ($\mathcal{T} =\big[ 1, 2, \ldots,t, \ldots, N_T \big]$), i.e., $\forall t \in \mathcal{T}$, to maximize potential damage in order to cancel secondary control objective in removing deviations in both MG global frequency and local voltage, with the action space $\mathcal{A}$ representing  control actions  (e.g., secondary control signal disruption, etc.) that influence state space  $\mathcal{S}$ and regulate interactions of DGs with the MG system (e.g., withdraw/inject power from/to the system or idle). The attack model assumes that the attacker possesses some, but not necessarily exact, knowledge of the system design (e.g., topology, parameters of lines, buses, generators, etc), which can be obtained from publicly available data sets (e.g., ISO New England \cite{7039273} and WECC-240 \cite{6039476}).

Throughout this work, stealthiness is considered with respect to an anomaly detector $\mathcal{D}$, characterized by the residual, and a threshold $\eta$ for the identification of an anomaly. Current detection approaches within the communication of local controllers (LCs) (primary) and MG {centralized} controller (MGCC) ({secondary}) involve using residual-based methods \cite{Zografopoulos2022}. The residual is often determined from the difference between measurements ${y}$ and their estimated values $\hat{{y}}$. As an example, the  $\chi^2$-distribution with $n-m$ degrees of freedom and confidence interval can determine the threshold $\eta$ equal to $\sigma\surd{\chi_{n-m, \alpha}^{2}}$. In the event that $\| {y}-\hat{{y}} \|_{2} > \eta$, an alarm will be raised.

\begin{proposition}\label{Proposition-equivalent-system}
Let us consider the attacked system of Eq. \eqref{complete} becomes as follows based on the aforementioned attack model at which an attack starts at an unknown time $t=T_0$ (i.e., and hence $a(t) = 0$  for $0 \leq t<T_0$) targeting a subset of the input vector $u(t)$ from the neighboring DGs.
\begin{equation}{
\label{attack_subsystem-u}
\ W_{a} : \left\{
\begin{array}{ll}
       \boldsymbol{\dot{x}} = A\boldsymbol{x} + \boldsymbol{f(x)} + B\boldsymbol{u} + B\Gamma_{u}\boldsymbol{a_{u}} \\
       \boldsymbol{y}= C\boldsymbol{x}
      
\end{array} 
\right. \
}
\end{equation}
\noindent where $\Gamma_u \in \mathbb{B}^{n_u \times\left|K_u\right|}$, $\mathbb{B} \triangleq\{0,1\}$, is the binary incidence matrix mapping the attack signal along with the mapping of the disturbance vector to the respective channels. The attack signal, incorporating also the disturbance, is $\boldsymbol{a_{u}}=\left[a_{u, 1}(t), \ldots, a_{u,\left|K_u\right|}(t)\right]^T \in \mathbb{R}^{\left|K_u\right|}$, where $K_u \subseteq\left\{1, \ldots, n_u\right\}$ represents the disruption resources available to the attacker able to corrupt the neighboring inputs of the system. For each $i \in\left\{1, \ldots,\left|K_u\right|\right\}$, $a_{u, i}(t)=0$ for $t \in \mathbb{R}_{+}$ if no attack occurs on the $i$-th transmission channel of $\mathcal{A}_{a_{u_{i}}}$. In the specific model of our MG, $n_u = 9$ and {$\left|K_u\right| = 2$}.
\end{proposition}

Taking into account Proposition \ref{Proposition-equivalent-system} and the threat model described so far, we consider that the attacker has partial model knowledge ($A, B, \Gamma, C$), available disruption resources ($K$), and {overall some limited but yet sufficient knowledge about the internal system signals as well as a restrict understanding of the network topology} in order to create an adverse disturbance. The goal is to generate a stealthy intermittent integrity attack in which first the attack activation time is identified, and then, the integrity model of the attack is constructed. 

\vspace{-2mm}
\subsection{Scheduling Intervals for Intermittent Integrity Attacks}\label{s:scheduling_attacks}

The time $\forall t \in \mathcal{T}$ in which the attack signal will be activated or not is assumed to be determined by the adversary at the time instants $t_{1}, \dots, t_{Na}$ $\in \mathcal{T} $, where $N_{a}$ $\in$ $\mathbb{Z}^+$. Let us consider a specific time slot $k$ in which the signal is active for a time interval $\tau_{k}$, i.e., $0 < \tau_{k}  \leq t_{k + 1} - t_{k}$. Thus, the attack is inactive when $t \geq t_{k} + \tau_{k}$.
Consequently, the time interval where the attack signal will be activated ($\Theta^{ac}_{k}$) and deactivated ($\Theta^{de}_{k}$) for the $k$-th attack slot is the following:
\begin{equation}
\label{attack slot}
 \left\{
\begin{array}{ll}
       \Theta^{ac}_{k}, ~~~~  t_{k} \leq t < t_{k} + \tau_{k}\\
       \Theta^{de}_{k}, ~~~~  t_{k} + \tau_{k} \leq t < \infty
\end{array} 
\right. \
\end{equation}
\noindent The time interval where the $k$-th attack is implemented can be defined as $\Theta^{}_{k} = \Theta^{ac}_{k} \cup \Theta^{de}_{k}$, and an auxiliary time interval where the attack slot is conducted can be formulated as:
{
\begin{equation}
\label{auxiliary_time}
\left\{
\begin{array}{ll}
       \Theta_{k}, ~~~~  t_{k} \leq t < t_{k + 1}, \forall k \in \{1, \dots, N_{a} - 1\}\\
       \Theta_{N_{a}}, ~~ t_{N_{a}} \leq t < \infty
\end{array} 
\right. \
\end{equation}
}

\subsection{Attack Model for Stealthy Intermittent Integrity Attacks}\label{s:_attacks}

Given the attack interval where the signal is activated and deactivated $\Theta^{}_{k}$, the complete attack model for the $k$-th slot can be described as:
\begin{subequations}\label{attack model 1}
\begin{align}
    & \dot{\zeta_{k}}(t) = (A + B\Gamma_{u} Q_{k})\zeta_{k}(t) + B\Gamma_{u}L_{a}l(t) \\
    & \zeta_{k}(t_{k})= -\Delta z_{k} \label{dzk-initial}
\end{align}
\end{subequations}
\begin{equation}
\label{attack model 2}
\ a_{k}(t) = \left\{
\begin{array}{ll}
       \ Q_{k} \zeta_{k}(t) + L_{a}l(t), ~ \forall t \in \Theta^{ac}_{k} \\
       \ 0, ~~~~~~~~~   \forall t \in \Theta^{de}_{k}
\end{array} 
\right. \
\end{equation}
\begin{equation}
\label{attack model 3}
\ a(t) = \sum_{i=1}^{k} a_{i}(t), ~~   \forall t \in \Theta_{k}
\end{equation}

\noindent where the matrix $Q_{k}$ $\in \mathbb{R}^{(n_u + n_y)\times n_x}$ is generated using Theorem \ref{T-stealth}. The dimensions of the matrix $L_{a}$ depend on the orthonormal basis for the inverse map of Im($V_{a}$) in B$\Gamma_{u}$, and can be computed with Theorem \ref{T-stealth}. The vector $l(t)$ can be any signal with proper dimensions. $\Delta z_{k} \in \mathbb{R}^{n_x}$ shows the difference between the real values of the system state and the states known by the adversary. The initial conditions for $\Delta z_{k}$ are usually bounded, where the upper and lower bounds are not specific requirements that the adversary needs to compute to execute the disturbance in the system. Accordingly, the initial condition from \eqref{dzk-initial}, $\Delta z_{k}$, is not zero. This strategy provides realism to the model since the attacker is not required to know the exact value of each state. Furthermore, we consider time slots when the attack signal is deactivated, $t \geq t_{k} + \tau_{k} $, because we want to guarantee continuity in the output signal $y(t)$. The convergence of $a(t)$ is also guaranteed if the attack model is built as \eqref{attack slot}, while the design of $a_{k}(t)$ ensures that the generated attacks can pass through some statistical anomaly detectors  $\mathcal{D}$ \cite{Zografopoulos2022,  basseville1993detection}. It is clear that the stealthiness of $a(t)$ is based on $\Delta z_{k}$, which is required to be sufficiently small not to be perceived. $\Delta z_{k}$ is valid under Assumption \ref{Dzk-assumption}.

\begin{assumption}\label{Dzk-assumption}
There exists two scalars $c_{1}$ \textgreater \ 0 and $c_{2}$ \textgreater \ 0 such that the initial condition, described in \eqref{dzk-initial}, is bounded: 
{
\begin{equation}
\label{initial condition}
c_{1} \leq \|\Delta z_{k}\| \leq c_{2}, \forall k \in (1, \dots, N_{a})\\
\end{equation}
}where $c_{1}$ and $c_{2}$ are sufficiently small and not required to be known by the adversary.
\end{assumption}

\begin{remark}
The inferior bound for $\|\Delta z_{k}\|$ shows that the attacker is not supposed to know the exact value of the states vector because the measurements are usually affected by the noise that comes to the system. The superior bound, on the other hand, is chosen in such a way that it will not be distinguished by the detector $\mathcal{D}$. 
\end{remark}

In \cite{Zhang2023}, the input signal $u(t)$ and $\Delta z_{k}$ are taken into consideration to evaluate the stealthiness of the attacks. These two parameters are important to establish the convergence of the output $y(t)$ when the disruption starts. 
Following the structure in \cite{ZhangK2022}, we split the system $W$ in \eqref{complete} 
during the time interval $\Theta_{k}$ where the attack holds:
\begin{equation}
\label{Incremental system1}
\ W_{1} : \left\{
\begin{array}{ll}
       \dot{x}_{1}(t) = Ax_{1}(t), ~ \forall t \in \Theta_{k}\\
        y_{1}(t) = Cx_{1}(t), ~ \forall t \in \Theta_{k}
\end{array} 
\right. \
\end{equation}
\begin{equation}
\label{Incremental system2}
\ W_{2} : \left\{
\begin{array}{ll}
       \dot{x}_{2}(t) = Ax_{2}(t) + f(t,x_{2}) + Bu(t) ,  \forall t \in \Theta_{k}\\
        y_{2}(t) = Cx_{2}(t),  \forall t \in \Theta_{k}
\end{array} 
\right. \
\end{equation}
\noindent where $W_{1}$ and $W_{2}$ correspond to the case where the system is working under nominal conditions. The initial conditions are ${x}_{1}(t_{k}) = -\Delta z_{k}$ and ${x}_{2}(t_{k})= {x}(t_{k})$. The way the split is done satisfies the condition of ${x}(t)= {x}_{1}(t) + {x}_{2}(t)$. In the same way, ${y}(t)= {y}_{1}(t) + {y}_{2}(t)$. Following the same technique, we can split our system $W_{a}$ in \eqref{attack_subsystem-u}, when the attacker injects a malicious signal $a(t)$ into the communication 
of the DGs:
\begin{equation}
\label{Attack 1}
\ W_{1_{a}} : \left\{
\begin{array}{ll}
       \dot{x}_{1_{a}}(t) = Ax_{1_{a}}(t) + B\Gamma_{u}a(t)\\
        y_{1_{a}}(t) = Cx_{1_{a}}(t) + D'a(t)
\end{array} 
\right. \
\end{equation}
\begin{equation}
\label{Attack 2}
\ W_{2_{a}} : \left\{
\begin{array}{ll}
       \dot{x}_{2_{a}}(t) = Ax_{2_{a}}(t) + f(t,x) - f(t,x_{n}) + Bu(t) \\
        y_{2_{a}}(t) = Cx_{2_{a}}(t)
\end{array} 
\right. \
\end{equation}
where $x_{n}$ represents the state vector in nominal conditions and $D' = \left[0_{n_y \times\left|K_u\right|}, \Gamma_y\right]$, where $\Gamma_y \in \mathbb{B}^{n_y \times\left|K_y\right|}$.  $a(t) = \beta_{i}(t- T_{o}) a_{u}(t) $, where $\beta_{i}(t- T_{o})$ describes the attack function dynamics  and is given by:
\begin{equation}
\label{attack interval 2}
\ \beta_{i}(t- T_{o}) \equiv \left\{
\begin{array}{ll}
       0, ~~~~~~~~~~~~~~~~ \forall t \in \Theta^{de}_{k} \\
       1 - e^{-b_{i}(t-T_{o})},~\forall t \in \Theta^{ac}_{k},
\end{array} 
\right. \
\end{equation}
with parameter $b_{i} \in \mathbb{R}$ representing the attack evolution rate. The initial conditions are ${x}_{1}(t_{k})= -\Delta z_{k}$ and ${x}_{2}(t_{k})= {x}(t_{k}) -\Delta z_{k}$. The splitting only works if the summation of $W_{a_{1}}$ and $W_{a_{2}}$ yields $W_{a}$. The computation of the matrix parameter $Q_{k}$ and the generation of several subspaces are discussed below.

In order to calculate the largest controlled invariant subspace of $W_{1_{a}}$ under an attack signal $a(t)$, the nonlinear function $f(t,x)$ of $W_{2_{a}}$ in \eqref{Attack 2} needs to satisfy:
\begin{equation} 
\label{incremental_nonlinear}
f(t,x) = f(t,x_{n} + x_{2_{a}}) 
\end{equation}Using \eqref{incremental_nonlinear} and the mean value theorem in \cite{Zemouche2005,ZhangK2022}, the difference $\Delta f = f(t,x) - f(t,x_{n} + x_{2_{a}})$ can be described as $\Delta f = F(t, \rho)x_{1_{a}}(t)$, where $\rho = \phi(x, x_{n} + x_{2_{a}}) \triangleq [\rho_{1}, \dots, \rho_{n_{x}}$] $\in$ $\mathbb{R}^{n_x\times n_x}$ with $\rho_{i} \in \operatorname{convh}(x, x_{n} + x_{2_{a}}$) for $i= 1, \dots, n_x$, and ${F(t, \rho)}$ as:

\hspace{1cm}

${F(t, \rho)} = \begin{bmatrix}
\frac{\mathrm{\partial f_{1}}}{\mathrm{\partial}x_{1}}(t, \rho_{1}) & \dots  & \frac{\mathrm{\partial f_{1}}}{\mathrm{\partial}x_{n_{x}}}(t, \rho_{1})\\
\vdots & \ddots & \vdots\\
\frac{\mathrm{\partial f_{n_{x}}}}{\mathrm{\partial}x_{1}}(t, \rho_{n_{x}}) & \dots  & \frac{\mathrm{\partial f_{n_{x}}}}{\mathrm{\partial}x_{n_{x}}}(t, \rho_{n_{x}}) 
\end{bmatrix}$

\hspace{1cm}

\noindent The theory behind the decoupling of the largest controlled invariant subspace and $F(t, \rho)$ can be found in \cite{ZhangK2022}.

\begin{theorem}[Stealthiness] \label{T-stealth}
Let us consider there is a weakly unobservable subspace generated by the incremental system  $W_{a_{1}}$ in \eqref{Attack 1}, where this subspace can be represented as ${V}(W^{a}_{1}$), the unobservable subspace created by the pair $(C, A)$, denoted by $H \subset \mathbb{R}^{n_x}$, and the largest controlled invariant subspace of $W^{a}_{1}$ in $H$, denoted by $V(H)$.
\begin{equation}
\label{Va}
V_{a}= V(W_{a_{1}}) \cap V(H)\\
\end{equation}
Then, there is a matrix $Q_{k}$, such that:
\begin{equation}
\label{matrixQ}
(A + B\Gamma_{u}Q_{k})V_{a} \subset V_{a}\\
\end{equation}
\begin{equation} 
\label{zk}
\Delta z_{k} \in V_{a}, \forall k \in (1, \dots, N_{a})\\
\end{equation}
Finally, there is a matrix $L_{a}$ that represents the attack generator gain as follows:
\begin{equation}
\label{matrixLa}
\operatorname{Im}(L_a) = \operatorname{ker}(D') \cap (B\Gamma_{u})^{-1}V_{a}\\
\end{equation}
\end{theorem}
\begin{proof}
The proof for Theorem 1 can be found in \cite{Zhang2023, ZhangK2022}.
\end{proof}

\vspace{-3mm}
\subsection{{Numerical Simulation Example}}

This section provides an illustration of a numerical simulation using a linear time-invariant state space representation of \eqref{attack_subsystem-u}. The matrices of the system are specified as follows: 
\begin{center}
$\boldsymbol{A} = \begin{bmatrix}
 0 & 1\\
 -0.2 & -0.1
\end{bmatrix}, \ 
\boldsymbol{B} = \begin{bmatrix}
 0\\
 1
\end{bmatrix},  \
\boldsymbol{C} = \begin{bmatrix}
1 & 0\\
\end{bmatrix}$
\end{center}where the pair ($A, C$) is observable. 

\begin{figure}[t]
    \centering
        \includegraphics[width=0.3\textwidth]{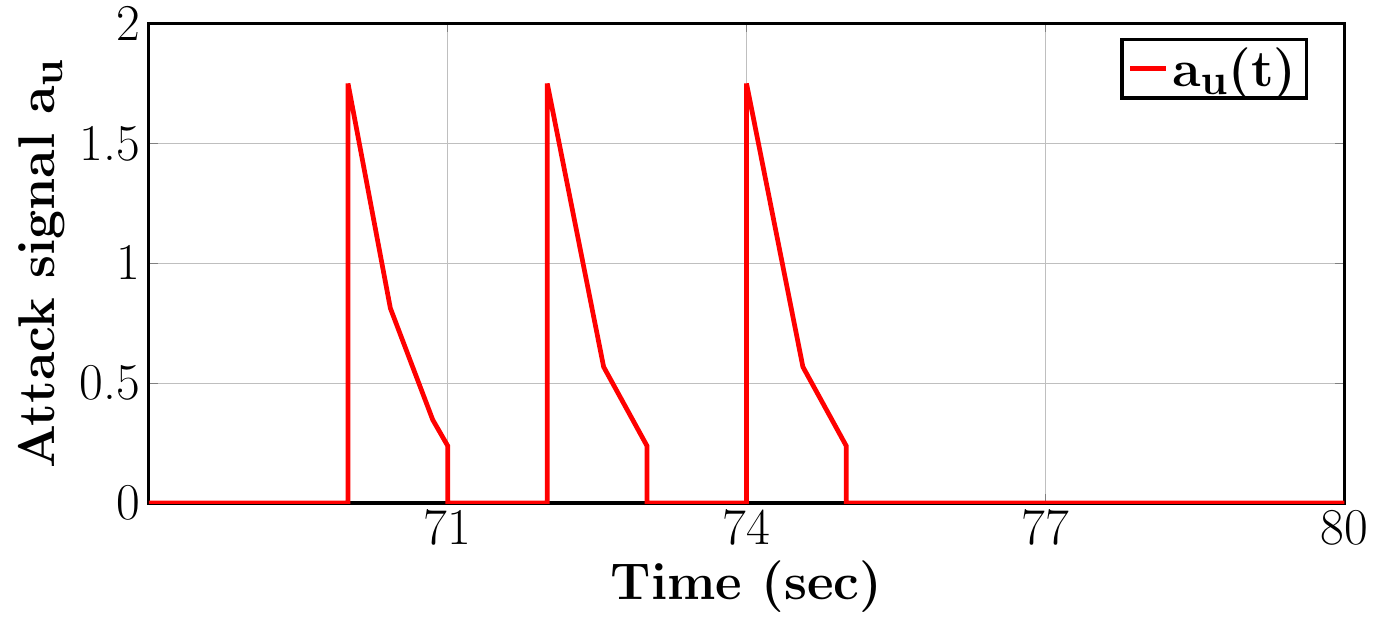}
        \vspace{-1mm}
    \caption{Time responses of the attack signal $a_{u}$(t).}
    \vspace{1mm}
    \label{input_attack_signal_car}
        \includegraphics[width=0.3\textwidth]{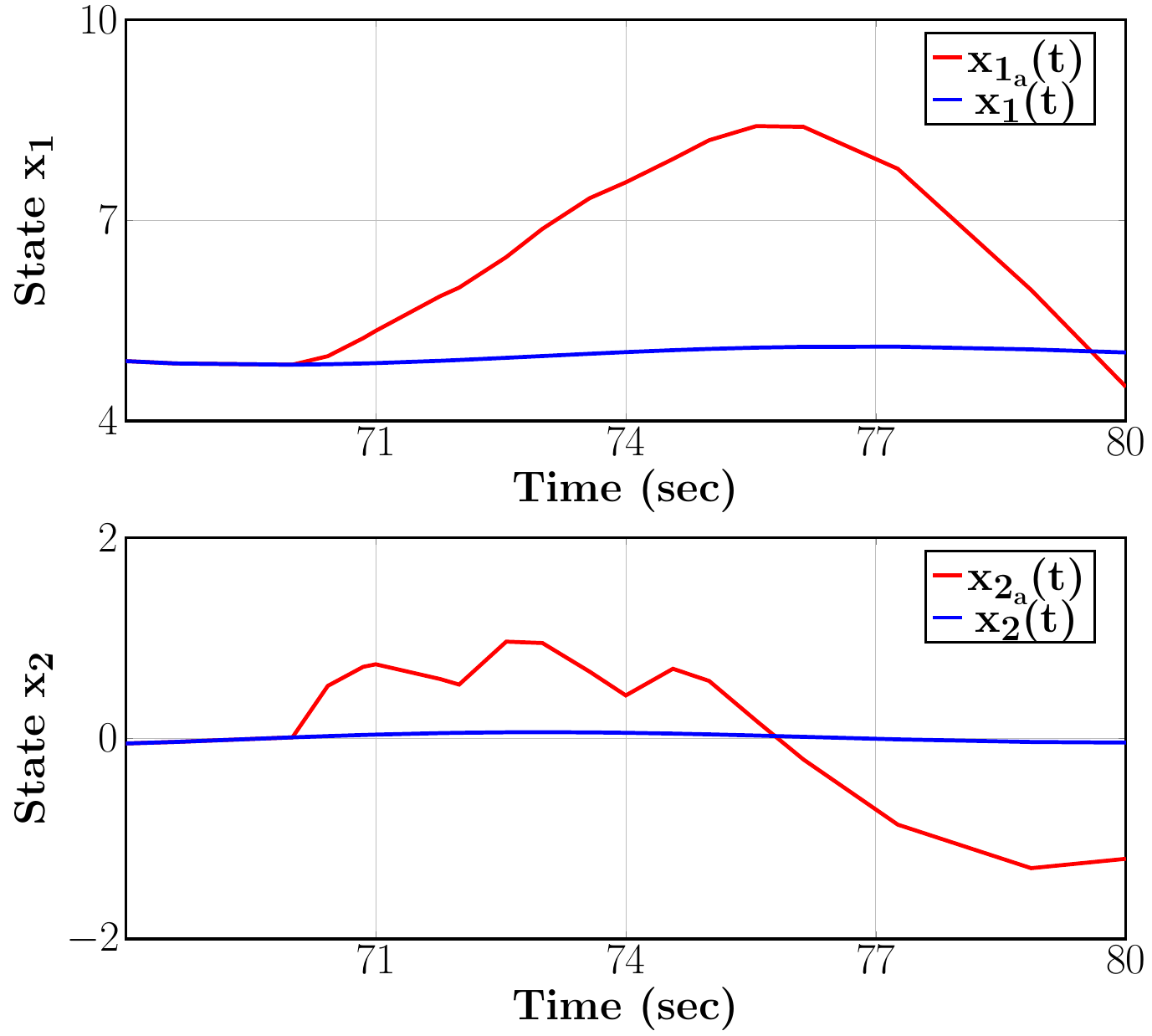}
        \vspace{-1mm}
    \caption{Time responses of the state vector $x$(t) under nominal $x_{n}(t)$ and attack $x_{n_{a}}(t)$ scenarios.}
    \label{states_attack_car}
    \vspace{-5mm}
\end{figure}

In order to demonstrate the nature of the stealthy intermittent integrity attack, we consider an adversary able to corrupt all the input-related disruption resources, i.e., $\Gamma_{u} = 1$. We also assume that the attack activating time instants $t_{k}$ are given first as follows: $t_{1} = 70 s$, $t_{2} = 72 s$, $t_{3} = 74 s$, with the same interval time being used when the attack signal is inactive, i.e., $\tau_{k} = 1 s$,  $\forall k$ $\in$ \{1, 2, 3\}. Following Theorem \ref{T-stealth} and utilizing the geometric approach toolbox \cite{geo-toolbox}, we can obtain the weakly unobservable subspace $\boldsymbol{V_{a}}$ and the matrix $\boldsymbol{Q_{k}}$ as follows:
\begin{center}
$\boldsymbol{V_{a}} = \begin{bmatrix}
-0.4472\\
0.8944
\end{bmatrix}, \ 
\boldsymbol{Q_{k}} = \begin{bmatrix}
 0.8 & -1.6\\
 0.4 & -0.8
\end{bmatrix}$
\end{center}The initial conditions $\Delta_{z_{k}}$ to generate the stealthy attack while satisfying \eqref{zk} are chosen as:
\begin{center}
$\boldsymbol{\Delta_{z_{1}}} = \begin{bmatrix}
-0.2236\\
0.4472
\end{bmatrix}, \ 
\boldsymbol{\Delta_{z_{2}}} = \begin{bmatrix}
 -0.1118\\
 0.2236 
\end{bmatrix}, \
\boldsymbol{\Delta_{z_{3}}}= \begin{bmatrix}
 -0.0559\\
 0.1118
\end{bmatrix}$
\end{center}The effects of the attack signal on the system are shown in Figs. \ref{input_attack_signal_car} and \ref{states_attack_car}. The attack signal $a_{u}$ introduced into the system input is intermittent and activated during specific time intervals. By comparing the system states at nominal and during attack scenarios, we can observe that, in this numerical example, the stealthy attack compromises the stability of the state vector. It can be observed from Fig. \ref{states_attack_car} that the system state $x(t)$ and its variation do not exhibit any abrupt changes during the attack pausing time instants.  Detecting this attack can be challenging, especially when measurement noises and disturbances are present in the system.

\vspace{-5mm}
\section{Residual-Based Detection Observer}\label{s:methodology}
{In this section, a framework is proposed to detect and mitigate the attacks previously introduced. The proposed detection architecture is illustrated in Fig. \ref{fig_det_arch}. An observer is deployed at each DG. On the left of the figure, the measurements transmitted from the connected DGs are represented. The detection method is residual-based: it compares the value of the measured variables, which is possibly affected by an attack, {with the estimation of the value that the variable would have in the absence of the attack}, as obtained through the use of the proposed observer. {In nominal conditions, the two values coincide. When an attack occurs, the error between the two increases, making it possible to detect the presence of the attack.}
The observer's inputs include the measured variables, $\boldsymbol{y}$, and the inputs of the DG, $\boldsymbol{u}$, {possibly subject to attacks.}}

Inspired by the works in  \cite{Rajamani_2016}, \cite{Rajamani_1998_tac}, 
which address also observers for systems having a structure like \eqref{complete}, this work proposes to use the following Luenberger-like observer: 
\begin{equation}\label{observer}
    \begin{cases}
          \dot{\hat{\boldsymbol{x}}} = A\hat{\boldsymbol{x}} + \boldsymbol{f}(\hat{\boldsymbol{x}},\boldsymbol{y}) + B\boldsymbol{u} + L(\hat{\boldsymbol{x}})(\boldsymbol{y} - \hat{\boldsymbol{y}})\\
      \hat{\boldsymbol{y}} = {C}\hat{\boldsymbol{x}},
    \end{cases}       
\end{equation}where $\boldsymbol{u}$ and $\boldsymbol{y}$ are, respectively, the inputs and the outputs of the DG, which are both input signals for the observer. Later on, in \ref{ss:detectability-subsection} we will present the nomenclature to distinguish the attack case from the normal operation. Notice that the observer is essentially a copy of the system, plus a correction term. Notice also that $f$ and $L$ in the observer may depend in general on the output of the system and on the state of the observer. Different observers arise depending on the choice of $f$ and $L$. 
Differently from the above-mentioned works \cite{Rajamani_2016}, \cite{Rajamani_1998_tac}, in this paper we consider a nonlinear gain $L(\hat{\boldsymbol{x}})$ in the correction term in \eqref{observer}, which will allow reducing the number of boundings done in the convergence proof for the observer, leading to less conservative results.
It is proved in the following that, with a particular choice of $L(\hat{\boldsymbol{x}})$, in the absence of attacks on the state and the input measurements, the state of the observer, $\hat{\boldsymbol{x}}$, tracks the state of the system, $\boldsymbol{x}$. This is used then to reveal the presence of attacks, by analyzing the residual vector $\boldsymbol{r} \triangleq \boldsymbol{y} - \hat{\boldsymbol{y}}$.

\begin{figure}[t]
     \includegraphics[width=0.95\linewidth]{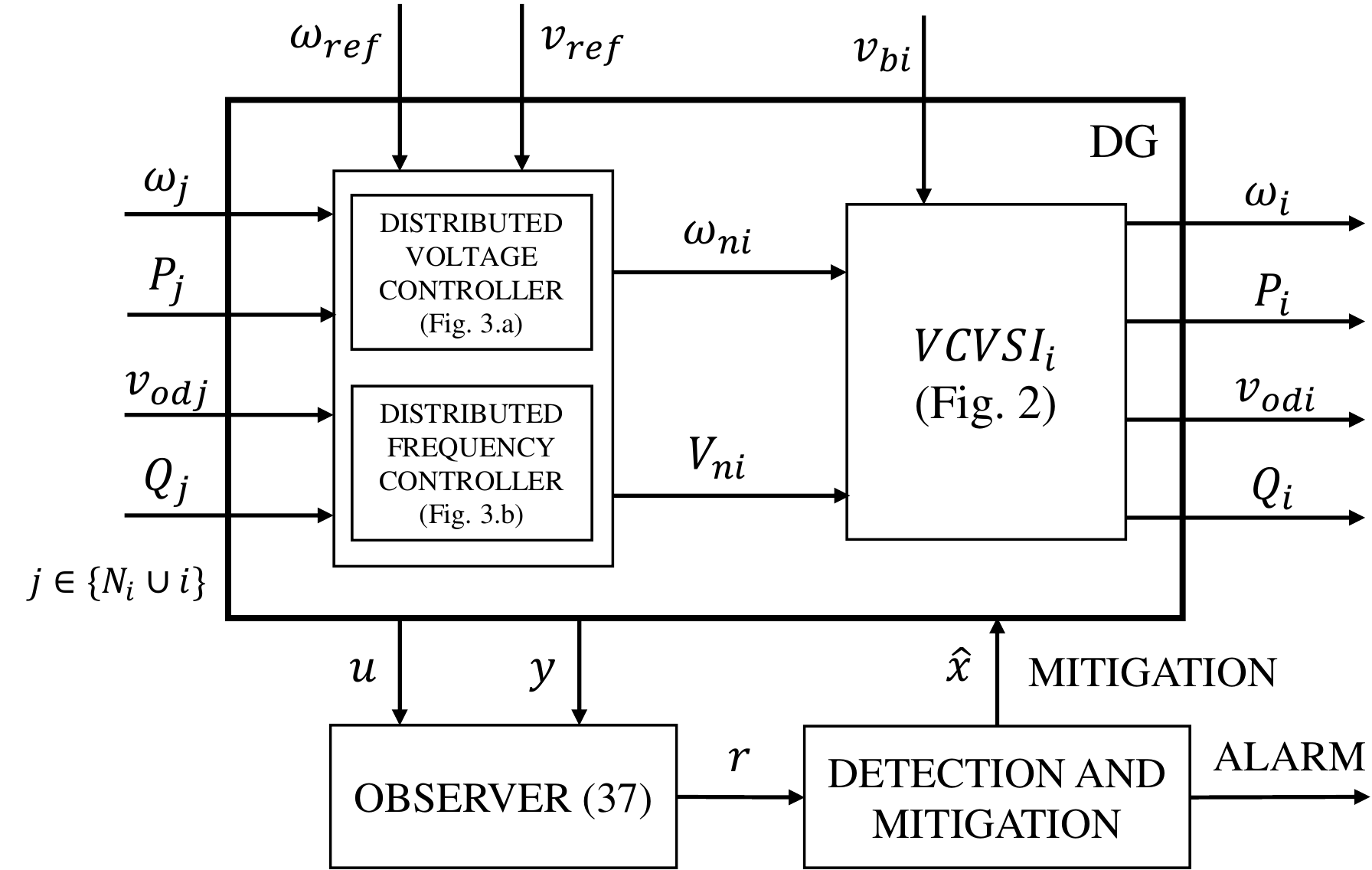}
     \vspace{-1mm}
     \caption{Proposed observer-based attack detection and mitigation architecture for the generic DG.}
     \label{fig_det_arch}
     \vspace{-5mm}
\end{figure}

\begin{theorem}[Convergence of the observer estimate]\label{theorem_observer}
There exists a choice of $L(\hat{\boldsymbol{x}})$ in \eqref{observer} such that the state $\hat{\boldsymbol{x}}$ of the observer \eqref{observer} globally asymptotically tracks the state $\boldsymbol{x}$ of system \eqref{complete} in absence of attacks.
\end{theorem}

\begin{proof}
The result is proved by showing that there exists a choice of $L(\hat{\boldsymbol{x}})$ such that the dynamics of the {estimation} error $\boldsymbol{\xi}\triangleq \boldsymbol{x} - \hat{\boldsymbol{x}}$ is globally asymptotically stable (i.e., $\boldsymbol{\xi}$ converges to $\boldsymbol{0}$, for any initial condition).

The derivative of $\boldsymbol{\xi}$ is:\begin{equation}\label{dotxi}
 \dot{\boldsymbol{\xi}} = \dot{\boldsymbol{x}} - \dot{\hat{\boldsymbol{x}}} = A\boldsymbol{\xi} + \boldsymbol{f}(\boldsymbol{x}) - \boldsymbol{f}(\hat{\boldsymbol{x}}) -  L(\hat{\boldsymbol{x}}){C}\boldsymbol{\xi}.  
\end{equation}Without loss of generality, one can write the nonlinear gain term $L(\hat{\boldsymbol{x}})$ as $L(\hat{\boldsymbol{x}}) \triangleq L' + L''(\hat{\boldsymbol{x}})$, so that \eqref{dotxi} becomes
%
\begin{equation}
 \dot{\boldsymbol{\xi}} = (A-L'{C})\boldsymbol{\xi} + \boldsymbol{f}(\boldsymbol{x}) - \boldsymbol{f}(\hat{\boldsymbol{x}}) -  L''(\hat{\boldsymbol{x}}){C}\boldsymbol{\xi}.
\end{equation}
The stability of the dynamics of $\boldsymbol{\xi}$ can proved via classical Lyapunov arguments (see, e.g., \cite{khalil2002nonlinear} [Theorem 4.2]). {The proof is constructive, i.e., a particular choice for $L'$ and $L''$ will be derived, which results in a stable observer}. Take the positive definite {and radially unbounded} function $V(\boldsymbol{\xi})=\frac{1}{2}\boldsymbol{\xi}^T\boldsymbol{\xi}$. 
Its time derivative is:
\small
\begin{equation}\label{dotV}
\begin{aligned}
    \dot{V} = &\boldsymbol{\xi}^T\dot{\boldsymbol{\xi}} = \boldsymbol{\xi}^T(A-L'{C})\boldsymbol{\xi} + \boldsymbol{\xi}^T\big(\boldsymbol{f}(\boldsymbol{x}) - \boldsymbol{f}(\hat{\boldsymbol{x}})\big) - \boldsymbol{\xi}^T L''(\hat{\boldsymbol{x}}){C} \boldsymbol{\xi}.
\end{aligned}
\end{equation}
\normalsize
The result is proved if it can be shown that there exists a choice of {$L'$ and} $L''$ for which $\dot{V}$ is negative definite. 

{First, it can be verified that, after removing \eqref{dotx_13}, the couple $(A,C)$ is detectable, hence} the term $\boldsymbol{\xi}^T(A-L'{C})\boldsymbol{\xi}$ can be made negative-definite by a proper selection of $L'$ {(i.e., any which makes $(A-L'{C})$ Hurwitz)}. {This part is addressed later in the proof}. 

Focusing on the term ${N(\boldsymbol{\xi},\boldsymbol{x},\boldsymbol{\hat{x}}) \triangleq }\boldsymbol{\xi}^T\big(\boldsymbol{f}(\boldsymbol{x}) - \boldsymbol{f}(\hat{\boldsymbol{x}})\big) - \boldsymbol{\xi}^T L''{C}(\hat{\boldsymbol{x}}) \boldsymbol{\xi}$, the strategy of the proof will be first to select $L''$ so that the resulting terms in $-\boldsymbol{\xi}^T L''{C}(\hat{\boldsymbol{x}}) \boldsymbol{\xi}$ cancel as many terms as possible of $\boldsymbol{\xi}^T\big(\boldsymbol{f}(\boldsymbol{x}) - \boldsymbol{f}(\hat{\boldsymbol{x}})\big)$, and then to bound the contribution to $\dot{V}$ of any residual term in ${N(\boldsymbol{\xi},\boldsymbol{x},\boldsymbol{\hat{x}})}$ with a proper choice of $L'$.
{First of all,}
the non-zero components of $\boldsymbol{f}$ are (see Section \ref{ss:fullstatemodel}):
\begin{itemize}
    \item $f_1(x) = \omega_{ci}x_9x_{11} + \omega_{ci}x_{10}x_{12}$
    \item $f_2(x) = - \omega_{ci}x_9x_{12} + \omega_{ci}x_{10}x_{11}$
    \item $f_7(x) = - m_{P_i}x_1x_8 + x_8x_{14}$
    \item $f_8(x) = m_{P_i}x_1x_7 - x_7x_{14}$
    \item $f_9(x) = - m_{P_i}x_1x_{10} + x_{10}x_{14}$
    \item $f_{10}(x) = m_{P_i}x_1x_{9} - x_{9}x_{14}$
    \item $f_{11}(x) = - m_{P_i}x_1x_{12} + x_{12}x_{14}$
    \item $f_{12}(x) = m_{P_i}x_1x_{11} - x_{11}x_{14}$
\end{itemize}
By noticing that $x_ix_j - \hat{x}_i\hat{x}_j = \xi_i\xi_j +\xi_i\hat{x}_j + \hat{x}_i\xi_j = 
\xi_ix_j + \hat{x}_i\xi_j$, the term $\boldsymbol{\xi}^T\big(\boldsymbol{f}(\boldsymbol{x}) - \boldsymbol{f}(\hat{\boldsymbol{x}})\big)$ in \eqref{dotV} can be written, after simplifications, as:
\small
\begin{equation}\label{xiTf_exact_v1}
\begin{aligned}
&\boldsymbol{\xi}^T\big(\boldsymbol{f}(\boldsymbol{x}) - \boldsymbol{f}(\hat{\boldsymbol{x}})\big) =  \\
& \omega_{ci}\xi_1\xi_9x_{11}  + \omega_{ci}\xi_1\hat{x}_9\xi_{11}
+\omega_{ci}\xi_1\xi_{10}x_{12}  + \omega_{ci}\xi_1\hat{x}_{10}\xi_{12} + \\
&-\omega_{ci}\xi_2\xi_9x_{12} - \omega_{ci}\xi_2\hat{x}_9\xi_{12} + \omega_{ci}\xi_2\xi_{10}x_{11} + \omega_{ci}\xi_2\hat{x}_{10}\xi_{11}\\
& -m_{P_i}\xi_7\xi_1x_{8} + \xi_7\hat{x}_{8}\xi_{14} + \\
& + m_{P_i}\xi_8\xi_1x_{7} - \xi_8\hat{x}_{7}\xi_{14} + \\
& - m_{P_i}\xi_9\xi_1x_{10} + \xi_9\hat{x}_{10}\xi_{14} + \\
& + m_{P_i}\xi_{10}\xi_1x_{9} - \xi_{10}\hat{x}_{9}\xi_{14} + \\
& - m_{P_i}\xi_{11}\xi_1x_{12} + \xi_{11}\hat{x}_{12}\xi_{14} + \\
& + m_{P_i}\xi_{12}\xi_1x_{11} - \xi_{12}\hat{x}_{11}\xi_{14}.
\end{aligned}
\end{equation}
\normalsize
The choice of $L''$ to cancel as many terms in \eqref{xiTf_exact_v1} as possible can now be done in two steps. First, {a}  proper selection of elements in $L''$ allows to cancel all the entries in \eqref{xiTf_exact_v1} that explicitly depend on $\hat{\boldsymbol{x}}$ (i.e., $\omega_{ci}\xi_1\hat{x}_9\xi_{11}$, $\omega_{ci}\xi_1\hat{x}_{10}\xi_{12}$, etc.). 
{In fact, first of all, given the structure \eqref{C} of $C$ it is easy to see that the generic $i_{th}$ row of $L''C$ is given by: 
\small{\begin{equation}
[L''_{i1}, L''_{i2}, 0,0,0,0, L''_{i3}, L''_{i4}, L''_{i5}, L''_{i6}, L''_{i7}, L''_{i8},0, L''_{i9}, L''_{i10}]
\end{equation}
}}
{Hence}, the generic term $c\xi_i\hat{x}_j\xi_k$ in \eqref{xiTf_exact_v1} (with $c$ a constant and $i,j,k\in\{1,...,15\}$), can be cancelled by choosing the $(i,k)$ entry of $L''{C}$  as $[L''{C}]_{i,k} = -c\hat{x}_j$. Hence, one can take: 
$L''_{1,7} = -\omega_{ci}\hat{x}_9$,
$L''_{1,8} = -\omega_{ci}\hat{x}_{10}$,
$L''_{2,8} = \omega_{ci}\hat{x}_9$,
$L''_{2,7} = -\omega_{ci}\hat{x}_{10}$,
$L''_{7,9} = -\hat{x}_8$,
$L''_{8,9} = \hat{x}_7$,
$L''_{9,9} = -\hat{x}_{10}$,
$L''_{10,9} = \hat{x}_9$,
$L''_{11,9} = -\hat{x}_{12}$,
$L''_{12,9} = \hat{x}_{11}$.

Then, by a proper selection of {the remaining} $L''$ {entries}, it is also possible to cancel the six terms in \eqref{xiTf_exact_v1} which depend on parameter $m_{P_i}$. This is done by noticing that they can be divided into three couples with a similar structure: $- m_{P_i}\xi_7\xi_1x_{8} + m_{P_i}\xi_8\xi_1x_{7}$,
$- m_{P_i}\xi_9\xi_1x_{10} + m_{P_i}\xi_{10}\xi_1x_{9}$, and $- m_{P_i}\xi_{11}\xi_1x_{12} + m_{P_i}\xi_{12}\xi_1x_{11}$. The first couple can be eliminated by selecting $L''_{1,3} = m_{P_i}\hat{x}_{8}$ and $L''_{1,4} = -m_{P_i}\hat{x}_{7}$, the second couple by selecting 
$L''_{1,5} = m_{P_i}\hat{x}_{10}$ and $L''_{1,16} = -m_{P_i}\hat{x}_{9}$, 
and the third one by selecting 
$L''_{11,1} = m_{P_i}\hat{x}_{12}$ and $L''_{12,1} = -m_{P_i}\hat{x}_{11}$. 

With the above choices, the term $\boldsymbol{\xi}^T\big(\boldsymbol{f}(\boldsymbol{x}) - \boldsymbol{f}(\hat{\boldsymbol{x}})\big) + \boldsymbol{\xi}^T L''(\hat{\boldsymbol{x}}){C} \boldsymbol{\xi}$ in \eqref{dotV} reduces to:
\small{\begin{equation}\label{xiTf_exact1}
\begin{aligned}
&\boldsymbol{\xi}^T\big(\boldsymbol{f}(\boldsymbol{x}) - \boldsymbol{f}(\hat{\boldsymbol{x}})\big) + \boldsymbol{\xi}^T L''(\hat{\boldsymbol{x}}){C} \boldsymbol{\xi} =  \\
& \omega_{ci}\xi_1\xi_9x_{11} +\omega_{ci}\xi_1\xi_{10}x_{12}  -\omega_{ci}\xi_2\xi_9x_{12}  + \omega_{ci}\xi_2\xi_{10}x_{11}
\end{aligned}
\end{equation}}
\normalsize which consists of the sum of terms of the kind $cx_i\xi_j\xi_k$, with $c$ being a coefficient. The generic state component can be upper bounded as $x_i\leq \bar{x}_i$, since the state space is contained in a compact set. From the Young's inequality, it is $c\bar{x}_i\xi_j\xi_k\leq |c||\bar{x}_i|(\xi_j^2 + \xi_k^2)$. 
Hence, {with simple calculations,} \eqref{xiTf_exact1} can be upper-bounded as:
\begin{equation}\label{xiTf_bound}
\begin{aligned}
&\boldsymbol{\xi}^T\big(\boldsymbol{f}(\boldsymbol{x}) - \boldsymbol{f}(\hat{\boldsymbol{x}})\big) + \boldsymbol{\xi}^T L''(\hat{\boldsymbol{x}}){C} \boldsymbol{\xi} \leq \\
&\leq \omega_{ci}(\bar{x}_{11} + \bar{x}_{12})(\xi_1^2 + \xi_2^2 + \xi_9^2 + \xi_{10}^2) = \boldsymbol{\xi}^T D \boldsymbol{\xi},
\end{aligned}
\end{equation}
with $D$ a diagonal, positive semi-definite matrix, such that, $D_{i,i}=\omega_{ci}(\bar{x}_{11} + \bar{x}_{12})$ for $i=1,2,9,10$, and zero otherwise.
In conclusion, by plugging \eqref{xiTf_bound} in \eqref{dotV}, $\dot{V}$ can be bounded as:
\begin{equation}\label{dotV_bound}
    \dot{V} \leq  \boldsymbol{\xi}^T(A + D - L'{C}s)\boldsymbol{\xi}.
\end{equation}
All that is left is to select $L'$ in \eqref{dotV_bound} so that matrix $(A + D - L'C)$ is negative definite. Since $(A+D,C)$ is detectable (after removing \eqref{dotx_13},), this can be done by selecting (via a standard eigenvalue assignment problem) any $L'$ such that $(A + D - L'C)$ is Hurwitz. The eigenvalues of $(A + D - L'C)$ can be chosen to control the speed of convergence to zero of the estimation error.

{In conclusion, it has been proven that the observer \eqref{observer} with $L(\hat{\boldsymbol{x}}) \triangleq  L' + L''(\hat{\boldsymbol{x}})$, and any $L'$, $L''(\hat{\boldsymbol{x}})$ chosen as:
\begin{itemize}
\item $L''_{1,7} = -\omega_{ci}\hat{x}_9$,
$L''_{1,8} = -\omega_{ci}\hat{x}_{10}$,
$L''_{2,8} = \omega_{ci}\hat{x}_9$,
$L''_{2,7} = -\omega_{ci}\hat{x}_{10}$,
$L''_{7,9} = -\hat{x}_8$,
$L''_{8,9} = \hat{x}_7$,
$L''_{9,9} = -\hat{x}_{10}$,
$L''_{10,9} = \hat{x}_9$,
$L''_{11,9} = -\hat{x}_{12}$,
$L''_{12,9} = \hat{x}_{11}$,
and all the other entries of $L''(\hat{\boldsymbol{x}})$ equal to zero,
\item $L'$  such that $(A + D - L'C)$ is Hurwitz, with $D$ a diagonal matrix with entries $D_{i,i}=\omega_{ci}(\bar{x}_{11} + \bar{x}_{12})$ for $i=1,2,9,10$, and zero otherwise,
\end{itemize}
}
globally asymptotically reconstructs the state $x$.
\end{proof}
\noindent {The key steps in the proof are summarized in Fig. \ref{proof_steps}}.
\begin{figure}[t]
\centering
\includegraphics[width=0.5\textwidth]{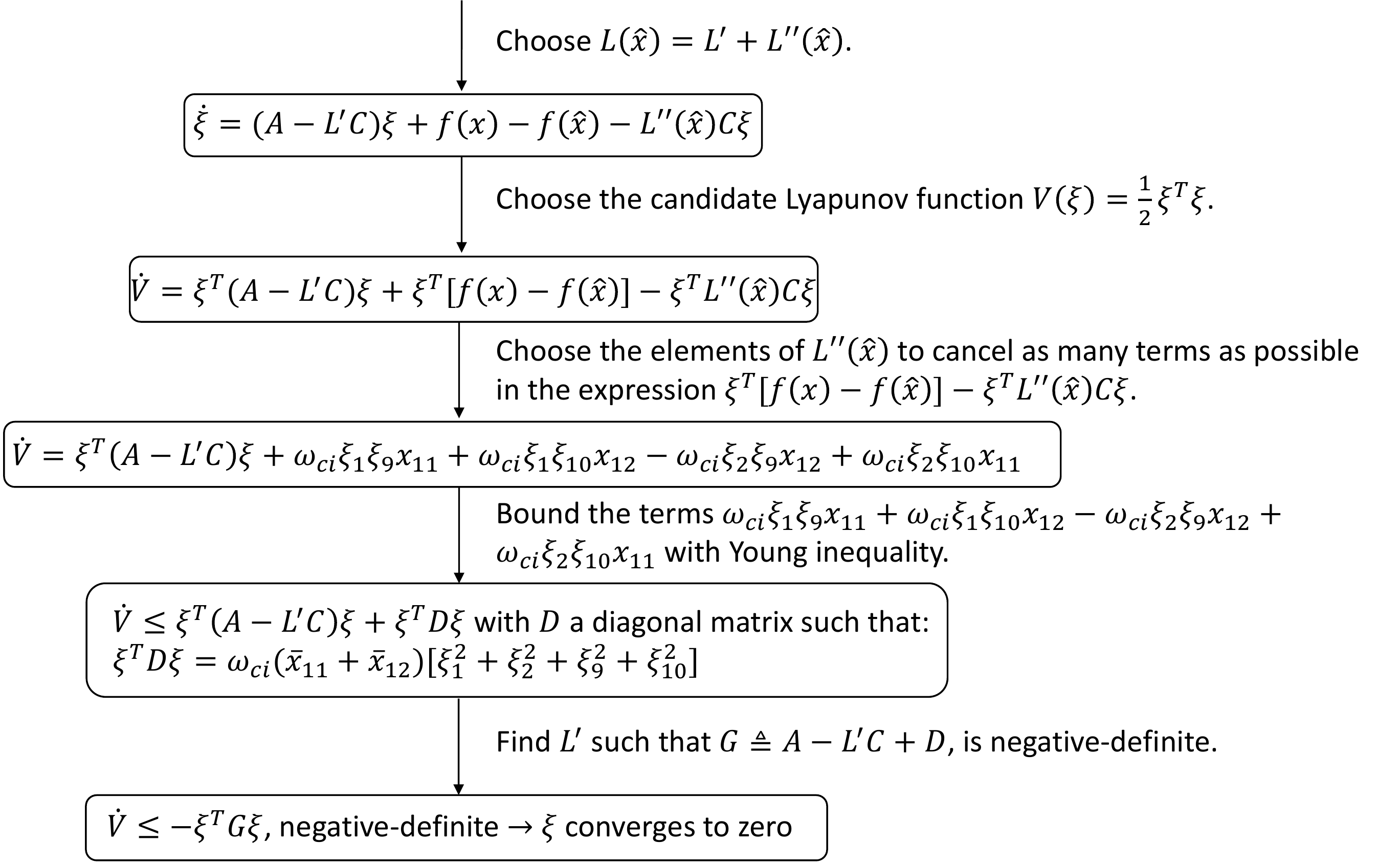}
\vspace{-2mm}
\caption{Main steps of the proof of Theorem \ref{theorem_observer}.}
\vspace{-5mm}
\label{proof_steps}
\end{figure}

The detection and mitigation strategy towards resilient MG operation is based on the analysis of the residual vector $\boldsymbol{r} \triangleq \boldsymbol{y} - \hat{\boldsymbol{y}}$, which represents the difference between the measurements, $\boldsymbol{y}$, affected by the attacks, and the estimated measurements, built on the basis of the estimated state, $\hat{\boldsymbol{x}}$, which tracks $\boldsymbol{x}$ when there are no attacks. A threshold-based detection scheme is adopted, in which an alarm is raised when $\|\boldsymbol{r}\|$ exceeds a pre-specified threshold $\eta$. 
When this occurs, the impacted state variables are replaced with their estimates from the observer. The same detection and mitigation strategy can be used on any other signal that is a function of the state variables. 
Finally, the following remark presents an alternative observer based on output injection \cite{bernard2019observer}. 

{
\begin{remark}[Output Injection Observer \cite{bernard2019observer}]\label{remark_obs1}
A first simpler and alternative observer design based on ``output injection'' can be made by noticing that all the state variables appearing in  $\boldsymbol{f}$ (see Section \ref{ss:fullstatemodel}) are actually measurable for our particular system (see Remark \ref{RemarkMeasures}), and thus  available to the observer from the output of the system. Hence, by following the simple design suggested, e.g., in \cite[Section 3.1.2]{bernard2019observer}, we can take in \eqref{observer} $\boldsymbol{f}(\hat{\boldsymbol{x}},\boldsymbol{y}) = \boldsymbol{f}(\boldsymbol{x})$, and $L(\hat{\boldsymbol{x}},\boldsymbol{y}) = L$ any constant matrix which makes $(A-LC) $ Hurwitz. As proved in \cite{bernard2019observer}, this simple design choice results in an observer which correctly estimates the DG state in normal operation (as it can be seen also from a simple adaptation of the proof of Theorem \ref{theorem_observer}).  
\end{remark}
}

\vspace{-4mm}
\subsection{Attack Detectability Analysis}\label{ss:detectability-subsection}

In this part, we present the attack detectability analysis for the two observers introduced previously and then we analyze the conditions for the aforementioned detection scheme. 
We start by analyzing the impact that the attack vector has on the dynamics of the states and on the residual variable. Let $x_n(t)$ denote the state of the system at time $t$ in normal conditions (i.e., when there is no attack according to \eqref{complete}), $x(t)$ denote the state of the system under attack (according to \eqref{attack_subsystem-u}), and $B\Gamma_{u}$ = $B_{a}$. We want to characterize the deviation in the state caused by the attack, i.e., the dynamics of the incremental state variables $x_a(t)\triangleq x(t)-x_n(t)$: 
\begin{equation}
\dot{x}_a = A{x}_a + f(x) - f(x_n)  + 
B_a{a_{u}}
\end{equation}

Moving now to the analysis of the impact on the residual, let $r_n(t)$ denote the residual at time $t$, in normal conditions, and denote with $r(t)$ the residual when the system is under attack. Define the incremental residual, i.e., the impact of the attack on the residual, as $r_a(t) \triangleq r(t)-r_n(t)=C(x(t) - \hat{x}(t)) - C(x_n(t) - \hat{x}_n(t))=C(\xi(t) - \xi_n(t)) \triangleq C\xi_a$, where $\xi$ and $\xi_n$ are, respectively, the estimation error in the attack and in the normal scenario (i.e., the deviation between the state of the system and the observer estimate).

Consider first the output injection observer of Remark \ref{remark_obs1}. In this case, simple calculations show that the increment to the detector state caused by the attack is governed by the following:
\begin{equation}
\dot{\hat{x}}_a = (A-LC)\hat{x}_a + f(x)-f(x_n) + LC(x-x_n).
\end{equation}
Also, the dynamics of $\xi_n(t)$ is:
\begin{equation}
\dot{\xi}_n \triangleq \dot{x}_n - \dot{\hat{x}}_n = (A - LC)\xi_n,
\end{equation}
and similarly the dynamics of $\xi(t)$ is:
\begin{equation}
\dot{\xi} \triangleq \dot{x} - \dot{\hat{x}}  = (A-LC)\xi  + B_a{a_{u}}
\end{equation}
From the above calculations, we can compute the dynamics of variable $r_a$, which characterizes the deviation that the attack causes  to the residual, fundamental to evaluate the detectability of the attacks.
\begin{equation}\label{dotra}
\dot{r}_a = C( \dot{\xi} - \dot{\xi}_n ) \triangleq C\dot{\xi}_a = C(A-LC)\xi_a  + 
CB_a{a_{u}}
\end{equation}
\begin{remark}[Undetectability condition for observer in Remark \ref{remark_obs1}]
From \eqref{dotra} it is seen that the attack is perfectly undetectable if $a_u\in \operatorname{ker}(CB_a)$, since then the residual does not depend on the attack. In particular, given the structure of \eqref{C}, it follows that any attack targeting only the state variables which are not measured is undetectable under the observer of Remark \ref{remark_obs1}. Hence, the use of this observer is not recommended in practice. 
Furthermore, even attacks such that $CB_a\neq 0$ could be easily made undetectable. As a matter of fact, dynamics \eqref{dotra} are linear, and the solutions to the equation $\dot{\xi}_a = (A-LC)\xi_a  + B_a{a_u}$ are also solutions to \eqref{dotra}. They are given by:
\begin{equation}\label{sol_xia}
\xi_a(t) = e^{(A-LC)t}\xi_a(0)  + \int_{0}^t e^{[A-LC](t-\tau)}B_a a_u(\tau)d\tau.
\end{equation}
Considering now that $\xi_a(0)=0$ (at the initial time of the attack the effect on the state is zero), and that $\dot{r}_a = C( \dot{\xi} - \dot{\xi}_n ) \triangleq C\dot{\xi}_a$, from \eqref{sol_xia} we find that if: 
\begin{equation}
\|C\int_{0}^t e^{[A-LC](t-\tau)}B_a a_u(\tau)d\tau\| \leq \tau
\end{equation}
then, the attack is not detected.
\end{remark}

Finally, if instead the proposed observer is used (Theorem \ref{theorem_observer}), the dynamics of the variation of the residual due to the attack are given by:
\begin{equation}\label{dotraTheorem2}
\begin{aligned}
\dot{r}_a := C\dot{\xi}_a & = C\{A\xi_a + [f(x)-f(x_n)] +\\
& -[f(\hat{x})-f(\hat{x}_n)] - L(\hat{x})C\xi_a + B_a a_u\}.
\end{aligned}
\end{equation}
Notice that 
\eqref{dotraTheorem2} is fundamentally different from \eqref{dotra}, since they are nonlinear and, beyond $\xi_a$, they include also $x$, $x_n$, $\hat{x}$, and $\hat{x}_n$. Thus, even when the attack satisfies $CB_a a_u = 0$ \cite{7001709}, the residual is affected by the impact of the attack on $x$. 

\subsubsection{Conditions for Attack Detection}

The residual scheme is designed and analyzed based on the derivation of a suitable observer which globally asymptotically tracks the state $\boldsymbol{x}$ of the system in the absence of attacks. An important related question is determining the class of attacks that can be detected. This part focuses on deriving the conditions for the aforementioned detection scheme. The analysis provides a theoretical result that characterizes quantitatively and implicitly the class of attacks detectable by the proposed scheme.

From \eqref{attack_subsystem-u}, and incorporating the attack function dynamics according \eqref{Attack 2} and \eqref{attack interval 2}, the state space model of each inverter can be described as: 
\begin{equation}
\small
\label{attack_subsystem-new}
\ \Tilde{W_{a}} : \left\{
\begin{array}{ll}
       \dot{x}(t) = k(x, u) + f(x(t)) + \zeta(t) \\
       + B_{a}\beta_{i}(t - To)a_{u}(t) \\
       y(t)= Cx(t)  
\end{array} 
\right. \
\end{equation} \normalsize where $k(x, u)$ represents the nominal function dynamics of each DG considering both the state and input vactors, and $\zeta(t) = E w(t)$ describes the known disturbances of the system according to \eqref{complete} and \eqref{attack_subsystem-u}. 
From \eqref{observer}, the state of the observer can be re-written as: 
\begin{equation} \label{obs}
\small
\begin{split}
\dot{\hat{x}}= k(\hat{x},u) + f(\hat{x}) + \Phi(\hat{x}, x)
\end{split}
\end{equation}\normalsize where $\Phi(\hat{x}, x)$ incorporates the nonlinear gain L($\hat{x}$) term from the Luenberger-like observer.

Following similar approaches in literature in which they use a filtering scheme to address unknown disturbances in nonlinear systems \cite{Keliris2013}, by filtering the output signal $y(t)$ of each DG, we can compute the filtered output $z(t)$:
\begin{equation} \label{output_signal_filtered}
\small
z(t) = H(s)[y(t)] = sH_{p}(s)[x(t)]
\end{equation} \normalsize where $H(s)$ and $H_{p}$ are asymptotically stable and, hence, BIBO stable. The conditions for selecting such filters can be found in \cite{Keliris2013}.
Now, by using $s[x(t)]= \dot{x}(t) + x(0)$, 
we obtain: 
\begin{equation} \label{output_signal_filter}
\small
\begin{split}
z(t) & = H_{p}(s)\left[\dot{x}(t)\right] + H_{p}(s)\left[x(0)\delta(t)\right] \\
 & = H_{p}(s)\left[k(x, u) + f(x) + B\Gamma_{u}a(t)\right] \\
 & + H_{p}(s)[\zeta(t)] + h_{p}(t)x(0)
\end{split}
\end{equation} \normalsize The initial conditions $x(0) = \hat{x}$(0) are considered to be known and the term $ h_{p}(t)$ describes a exponential decay \cite{Keliris2013}, converging to zero, eventually. Filtering the output of the observer, $\hat{z}$ can be described as:
\begin{equation} \label{observer_signal_filtered}
\small
\begin{split}
\hat{z}(t) & = H(s)[\hat{x}(t)] \\
 & = H_{p}(s)[k(\hat{x},u) + f(\hat{x}) + \Phi{\hat{x}, x}] + h_{p}(t)\hat{x}(0)\\
\end{split}
\end{equation} \normalsize



The residual error {$r(t)$} to detect the attack in each DG can be calculated from \eqref{output_signal_filtered} and \eqref{observer_signal_filtered}, and defined as: 
\begin{equation} \label{residual_z}
\small
\begin{split}
r(t) = z(t) - \hat{z}(t)
\end{split}
\end{equation} \normalsize The logic for attack detection is defined as:
\begin{equation}
\small 
\label{logic_threshold}
\left\{
\begin{array}{ll}
      |r(t)|\le \eta \rightarrow No \ attack \\
      |r(t)| > \eta \rightarrow Attack \ detected
\end{array} 
\right. 
\end{equation}
\normalsize Prior to the attack ($t <T_{o}$), the residual signal can be written using \eqref{output_signal_filtered} and \eqref{output_signal_filter}  as:
\begin{equation} \label{residual_general}
\small
\begin{split}
r(t) = H_{p}(s)[\chi(t)] + H_{p}(s)[\zeta(t)] = H_{p}(s)[\chi(t)] + \zeta(t)]
\end{split}
\end{equation}\normalsize where the term $\chi(t)$ incorporates the characteristics of the system and the observer, can be defined as:
\begin{equation} \label{system_observer}
\small
\begin{split}
\chi(t)= [k(x,u) - k(\hat{x}, u)] + [f(x) - f(\hat{x})] - \Phi{(\hat{x}, x)}
\end{split}
\end{equation}\normalsize Now, by taking bounds on the residual error, we obtain:
\begin{equation} \label{bound_residual}
\small
\begin{split}
r(t) & = |H_{p}(s)[\chi(t)] + H_{p}(s)[\zeta(t)]| \\
 & \le |H_{p}(s)[\chi(t)]| + |H_{p}(s)[\zeta(t)]| \\
 & = \bigg | \int_{0}^{t} h_{p} (t - \tau) \chi (\tau) \, d\tau \bigg | + \bigg | \int_{0}^{t} h_{p} (t - \tau) \zeta (\tau) \, d\tau \bigg | \\
 & \le \int_{0}^{t} | h_{p} (t - \tau) | | \chi (\tau) | \ d\tau + \int_{0}^{t} | h_{p} (t - \tau) | | \zeta (\tau) | \ d\tau \\
 & \le \int_{0}^{t} | h_{p} (t - \tau) | \chi (\tau) d\tau +  \int_{0}^{t} | h_{p} (t - \tau) | \zeta (\tau) d\tau
\end{split}
\end{equation} \normalsize A suitable threshold $\eta(t)$ is given as:
\begin{equation} \label{threshold_definition}
\small
\begin{split}
\eta(t) = \int_{0}^{t} \bar{h}_{p} (t - \tau)  \bar{\chi} (\tau) d\tau + \int_{0}^{t} \bar{h}_{p} (t - \tau)  \bar{\zeta} (\tau) d\tau
\end{split}
\end{equation} \normalsize Finally, using the approach in \cite{Keliris2013}, the {threshold can be implemented as}:
\begin{equation} \label{threshold}
\small
\begin{split}
\eta(t) = \bar{H}_{p}(s) \bar{\chi} (t) + \bar{H}_{p}(s) \bar{\zeta} (t) = \bar{H}_{p}(s) [\bar{\chi} (t) + \bar{\zeta} (t) ]
\end{split}
\end{equation}  \normalsize
\begin{theorem}[Attack Detectability] \label{T-detect}
Consider the nonlinear system described in \eqref{attack_subsystem-new}, with the residual-based observer detection scheme in \eqref{obs} with the conditions \eqref{output_signal_filtered}--\eqref{residual_z} and \eqref{threshold}. An attack introduced at $t= T_{o}$ is detectable if the attack function $a_{u}$(t) satisfies the following inequality:
\begin{equation} \label{detectability_condition}
\small
\begin{split}
\bigg | \int_{T_{o}}^{t} h_{p} (t - \tau) (1 - e^{-b_{i}(t-T_{o})})a_{u}(t) d\tau \bigg | > 2\eta(t)
\end{split}
\end{equation}
\normalsize 
\begin{proof} In the presence of an attack at $t= T_{o}$, \eqref{residual_general} becomes:
\begin{equation} \label{residual_proof1}
\small
\begin{split}
r(t) = H_{p}(s)[\chi (t) + \beta_{i}(t - To)a_{u}(t)] + H_{p}(s)[\zeta(t)]
\end{split}
\end{equation} \normalsize By using the triangle inequality, for $t >$ $T_{o}$, $r(t)$ is computed:
\begin{equation} \label{residual_detectability_bound}
\small
\begin{split}
r(t) & \geq - | H_{p}(s)[\chi(t)] | - |H_{p}(s)[\zeta(t)]| \\
 & + | H_{p}(s)[\beta_{i}(t - To)a_{u}(t)]| \\
 & \geq - \int_{0}^{t} | h_{p} (t - \tau) | | \chi (\tau) | \, d\tau  - \int_{0}^{t} | h_{p} (t - \tau) | | \zeta (\tau) | \\
 & + |H_{p}(s)[\beta_{i}(t - To)a_{u}(t)]| \\
 & \geq - \int_{0}^{t} \bar{h}_{p} (t - \tau) \bar{\chi} (\tau) d\tau - \int_{0}^{t} \bar{h}_{p} (t - \tau) \bar{\zeta} (\tau) d\tau \\
 & + |H_{p}(s)[\beta_{i}(t - To)a_{u}(t)]| \\
 & \geq - \eta(t) + |H_{p}(s)[\beta_{i}(t - To)a_{u}(t)]|
\end{split}
\end{equation}\normalsize The final attack detectability condition, can be computed as:
\begin{equation} \label{residual_proof2}
\small
\begin{split}
|H_{p}(s)[\beta_{i}(t - To)a_{u}(t)]| > 2\eta(t)
\end{split}
\end{equation}
\end{proof}

\end{theorem}

\normalsize 

\vspace{-5mm}
\section{Bi-Level Stability-Constrained Attack-Mitigation Formulation}\label{s:bilevel}
In this section, we first present the OPF problem for AC MGs and then identify the worst-case cyber-attack in a formulation of a bi-level optimization problem. The attacker follows the threat model of Section \ref{s:attackformulation} while satisfying the OPF of the system operation.  Operators, in order to ensure a stable operation of the MG, react by incorporating both stability constraints as well as the design of the detection observer. These conditions ensure not only the MG stable and  optimal operation but also guarantee the detection of attacks. Fig. \ref{fig:bi-level-fig} presents the concept of the formulation.

\begin{figure}[t]
\centering
\includegraphics[width=2.5in]{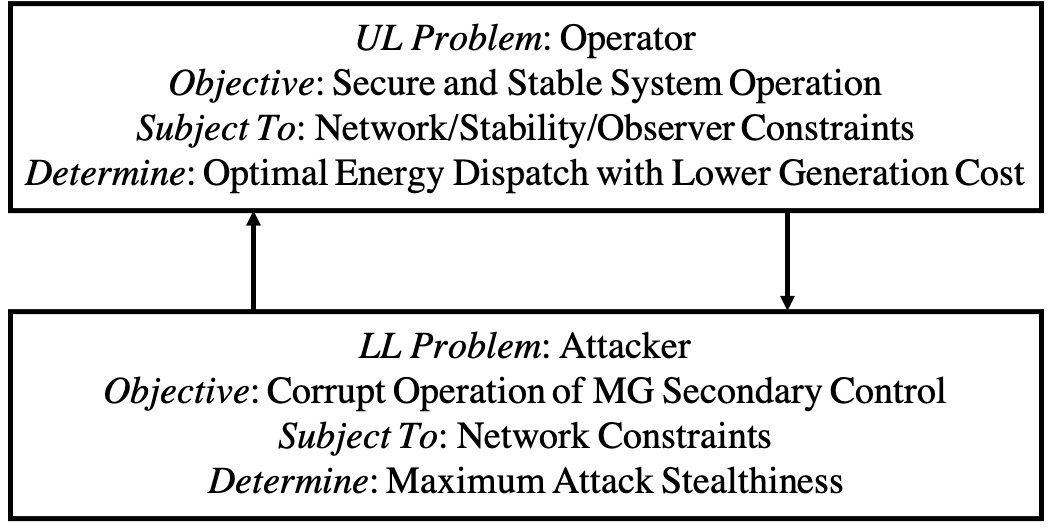}
\caption{Bi-level formulation that considers attacker and operator objectives.}
\vspace{-5mm}
\label{fig:bi-level-fig}
\end{figure}

\vspace{-3mm}
\subsection{Optimal Power Flow in AC Microgrids}

The injected active and reactive power at every inverter-based DG are $\boldsymbol{p}^g \in$ $\mathbb{R}^{n^{\mathrm{g}} \times 1}$ and $\boldsymbol{q}^{\mathrm{g}} \in \mathbb{R}^{n^{\mathrm{g}} \times 1}$. The total load demand is characterized by $\boldsymbol{d} \in \mathbb{C}^{n^{\mathrm{d}} \times 1}$ in the form of constant complex impedances $\boldsymbol{d} = [R_1 + X_1, \cdots, R_{n^{\mathrm{d}}} + X_{n^{\mathrm{d}}} ]$. 
The MG admittance matrix is $\boldsymbol{Y} \in \mathbb{C}^{n \times n}$ and the load and inverter incidence matrices are given by $\hat{\boldsymbol{D}} \in\{0,1\}^{n^{\mathrm{d}} \times n}$ and $\boldsymbol{G} \in\{0,1\}^{n^{\mathrm{g}} \times n}$, respectively. Thus, the from and to admittance matrices are represented as $\overrightarrow{\boldsymbol{Y}}, \overleftarrow{Y} \in \mathbb{C}^{n \times n}$, and their respective branch-incidence matrices as $\overrightarrow{\boldsymbol{L}}, \overleftarrow{\boldsymbol{L}} \in\{0,1\}^{l \times n}$. Finally, $\boldsymbol{v}_{o}$ is the terminal voltages at inverter/DG output terminals, and $\boldsymbol{v}=\left[\boldsymbol{v}^{\mathrm{g}}, \boldsymbol{v}^{\mathrm{b}}\right] \in \mathbb{C}^{n \times 1}$ represents all the voltages, where $\boldsymbol{v}^{\mathrm{g}} \in \mathbb{C}^{n^{\mathrm{g}} \times 1}$ is the vector of bus voltages at point of coupling, and $\boldsymbol{v}^{\mathrm{b}} \in \mathbb{C}^{n-n^{\mathrm{g}} \times 1}$ is the vector of all remaining MG buses. The OPF problem is then formulated as: 
\begin{equation}\label{min-pg}
\min~\mathrm{h}\left(\boldsymbol{p}^{\mathrm{g}}\right) = \min _{P_i} \sum_{i=1}^{n^{\mathrm{g}}} C_i\left(P_i\right)
\end{equation}
\begin{equation}\label{nodal-balance}
\textit{s.t.:}~~ \boldsymbol{G}^{\top}\left(\boldsymbol{p}^{\mathrm{g}}+\mathrm{j} \boldsymbol{q}^{\mathrm{g}}\right)=\hat{\boldsymbol{D}}^{\top} \boldsymbol{d}+\operatorname{diag}\left\{\boldsymbol{v} \boldsymbol{v}^* \boldsymbol{Y}^*\right\}
\end{equation}
\begin{equation}\label{p-inverters}
\boldsymbol{p}^{\min } \leq \boldsymbol{p}^g \leq \boldsymbol{p}^{\max }
\end{equation}
\begin{equation}\label{q-inverters}
\boldsymbol{q}^{\min } \leq \boldsymbol{q}^g \leq \boldsymbol{q}^{\max }
\end{equation}
\begin{equation}\label{line-flows}
\operatorname{diag}\left\{\overrightarrow{\overleftarrow{\boldsymbol{L}}} \boldsymbol{v} \boldsymbol{v}^* \overrightarrow{\overleftarrow{\boldsymbol{Y}}}^*\right\} \leq \boldsymbol{l}^{\max }
\end{equation}
\begin{equation}\label{voltage-constr}
\left(\boldsymbol{v}^{\min }\right)^2 \leq|\boldsymbol{v}|^2 \leq\left(\boldsymbol{v}^{\max }\right)^2
\end{equation}

\noindent Here, in order to minimize the total cost of DG output while satisfying system-wide requirements/constraints, we need to schedule each generator's active power output. The generation cost for the $i_{th}$ generator is typically approximated by quadratic cost functions\footnote{Quadratic cost functions are widely used in OPF problems even for inverter-based systems \cite{wu2016distributed}. The problem formulation, however, can be solved for any convex formulation of the cost function.} $ C_i\left(P_i\right)=\alpha_i P_i^2+\beta_i P_i+\gamma_i$, where $\alpha_i, \beta_i$, and $\gamma_i$ are the cost parameters. The system constraints include the enforcement of the nodal power balance in Eq.  \eqref{nodal-balance}, the bounding conditions of active and reactive power of the individual DG outputs in Eqs. \eqref{p-inverters} and \eqref{q-inverters}, respectively, the line flows in either direction using Eq. \eqref{line-flows}, and the voltage magnitude limits  within $\left[\boldsymbol{v}^{\min }, \boldsymbol{v}^{\max }\right]$ according to Eq. \eqref{voltage-constr}.

\vspace{-3mm}
\subsection{Lower Level Problem: Identification of Worst-Case Attack}

The aim of the attack is to maximize the effect of the disruption at the output signals of the MG secondary controller while remaining stealthy as explained in Section \ref{s:attackformulation}. The formalization of the worst-case attack model is achieved by co-optimizing it with the MG operations as given in \eqref{min-pg}--\eqref{voltage-constr} and adopting a common cybersecurity practice, which is to conservatively assume a strong (omniscient) and stealthy attacker with full knowledge of the network requirements including all nodal and line parameters, operating limits, and priorities. These assumptions lead to assessing the worst-case impact of the attack. The attacker with a generic objective, denoted as $ O^{A} ({act}'_{it})$,   can be modeled as in \eqref{eq:max-attack}, where  \eqref{eq:att_dso} presents the MG OPF model conditioned by the modified actions ${act}'_{it}$ that satisfy the stealthiness properties of \eqref{stealthiness-constraint}. 
\begin{equation}\label{eq:max-attack}
\max  O^{A} (\cap_{\forall i \in \mathcal G}{act}'_{it})
\end{equation}
\begin{equation}\label{eq:att_dso}
\textit{s.t.:}~\big\{\eqref{min-pg}-\eqref{voltage-constr}{\text{ with }}  {act}_{it}={act}'_{it}, \forall i \in \mathcal{G}, t \in \mathcal{T} \big\} 
\end{equation}
\begin{equation}\label{stealthiness-constraint}
{act}'_{it}~\textit{s.t.:}~\big\{\eqref{attack model 1}-\eqref{initial condition}~\text{and}~\eqref{Va}-\eqref{zk}  \big\}
\end{equation}

\vspace{-2mm}
\subsection{Upper Level Problem: Stability-Constrained Operation}

In this part, the MG entity operates at the top of the hierarchy and aims to  enforce the optimal operating setpoints  to the inverters while satisfying OPF conditions and  ensuring the small-signal stability of the inverter-based MG, since OPF cannot guarantee MG stability \cite{pullaguram2021small}, let alone the existence of a stealthy attack, and assuming the existence of a residual-based observer able to detect and mitigate malicious signal manipulations according \eqref{eq:max-attack}--\eqref{stealthiness-constraint}. To establish the stability restriction, the MG is initially represented as a collection of differential-algebraic equations: 
\begin{IEEEeqnarray}{lCr}\label{eq:MG-modeling-stability1}
\dot{\boldsymbol{x}} & =\boldsymbol{\Tilde{f}}(\boldsymbol{x}, \boldsymbol{z})\\
0 & =\boldsymbol{\Tilde{g}}(\boldsymbol{x}, \boldsymbol{z}) \label{eq:MG-modeling-stability2}
\end{IEEEeqnarray}
\noindent In the above equations, $\boldsymbol{\Tilde{f}}$ and $\boldsymbol{\Tilde{g}}$ denote the nonlinear differential and algebraic equation vectors of a MG, respectively. $\boldsymbol{x}$ and $\boldsymbol{z}$ are the vectors of state and algebraic variables with sizes $n_{\mathrm{x}}$ and $n_z$, respectively. A third-order inverter model, as stated in \cite{luo2014spatiotemporal}, is employed to minimize the computational burden. Fig. \ref{fig-vcvsi} depicts voltage and current controllers that are equipped with an $L C$ filter, which possess a significantly higher closed-loop bandwidth compared to the power controller module. It can be assumed that these control loops achieve quasi-steady-state quickly. Therefore, the vector of differential equations, denoted by $\boldsymbol{\Tilde{f}}$, includes state variables represented by $\boldsymbol{x}=[\boldsymbol{p}^{\top}, \boldsymbol{q}^{\top}, \boldsymbol{\delta}^{\top}]^{\top} \in \mathbb{R}^{n_{\mathrm{x}} \times 1}$.

\begin{equation}\label{eq:f(x)-mg-p}
 \dot{\boldsymbol{x}}^{\top}=\left[\begin{array}{c}-\boldsymbol{\omega}_c \cdot \boldsymbol{p}+\boldsymbol{\omega}_c \cdot \operatorname{Re}\left\{\boldsymbol{v}_{o} \cdot\left(\boldsymbol{i}_{o}\right)^*\right\} \\ -\boldsymbol{\omega}_c \cdot \boldsymbol{q}+\boldsymbol{\omega}_c \cdot \operatorname{Im}\left\{\boldsymbol{v}_{o} \cdot\left(\boldsymbol{i}_{o}\right)^*\right\} \\ \left(\boldsymbol{\omega}-\omega_{com}\right) {\omega}_{b} \end{array}\right]   
\end{equation}
\noindent where $\boldsymbol{\omega}_c$, $\boldsymbol{p}$ and $\boldsymbol{q}$ $\in \mathbb{R}^{n^{\mathrm{g}} \times 1}$,  $\boldsymbol{v}_{o} \in \mathbb{C}^{n^{\mathrm{g}} \times 1}=$ $v_{\mathrm{od}}+\mathrm{j} v_{\mathrm{oq}}$, and $i^{\mathrm{o}} \in \mathbb{C}^{n^{\mathrm{g}} \times 1}=i_{\mathrm{od}}+\mathrm{j} i_{\mathrm{oq}}$. 
The operating inverter frequency $\boldsymbol{\omega}$ is obtained using:
$\boldsymbol{\omega}=\omega_{b}-\boldsymbol{m}_{{P}} \cdot \boldsymbol{p}+\boldsymbol{m}_{{P}} \cdot \boldsymbol{p}_{\mathrm{opf}}$, 
where $\boldsymbol{p}_{\mathrm{opf}}$ is the active power setpoint provided by the OPF. 

\begin{table}[t]
\captionsetup{font=small}
\centering
\caption{{Parameters of the test MG.}}
\vspace{-2mm}
\begin{tabular}{||c|c||c|c||} 
\hline \hline
\multicolumn{2}{||c||}{DGs 1 \& 2~}     & \multicolumn{2}{c||}{DGs 3 \& 4}      \\ 
\hline \hline
Parameter & Value & Parameter & Value \\
\hline \hline
$m_{P}$  & $9.4 \times 10^{-5}$ & $m_{P}$  & $12.5 \times 10^{-5}$  \\ 
\hline
$n_{Q}$  & $1.3 \times 10^{-3}$  & $n_{Q}$  & $1.5 \times 10^{-3}$   \\ 
\hline
$R_{c}$  & $0.03$ $\Omega$                         & $R_{c}$  & $0.03$ $\Omega$                         \\ 
\hline
$L_{c}$  & $0.35$ mH                       & $L_{c}$  & $0.35$ mH                      \\ 
\hline
$R_{f}$  & $0.1$ $\Omega$                           & $R_{f}$  & $0.1$ $\Omega$                          \\ 
\hline
$L_{f}$  & $1.35$ mH                       & $L_{f}$  & $1.35$ mH                      \\ 
\hline
$C_{f}$  & $50$ $\mu$F                         & $C_{f}$  & $50$ $\mu$F                        \\ 
\hline
$K_{PV}$ & $0.1$                           & $K_{PV}$ & $0.05$                         \\ 
\hline
$K_{IV}$ & $420$                           & $K_{IV}$ & $390$                          \\ 
\hline
$K_{PC}$ & $15$                            & $K_{PC}$ & $10.5$                         \\ 
\hline
$K_{IC}$ & $20000$                         & $K_{IC}$ & $16000$                       \\ 
\hline
$\omega_{b}$  & $314.16$ rad/s                       & $\omega_{b}$  & $314.16$ rad/s                      \\ 
\hline
$F$   & $0.75$                          & $F$   & $0.75$                         \\ 
\hline
$\omega_{c}$  & $31.41$ Hz                        & $\omega_{c}$  & $31.41$ Hz                        \\
\hline
$c_{f}$  & $30$                        & $c_{f}$  & $30$                        \\
\hline
$c_{v}$  & $30$                         & $c_{v}$  & $30$                        \\
\hline \hline
\end{tabular}
\vspace{-2mm}
\label{table_I}
\end{table}
\begin{table}[t]
\captionsetup{font=small}
\centering
\caption{{Parameters of MG lines and loads.}}
\vspace{-2mm}
\begin{tabular}{||c|c|c|c|c|c|c|c||} 
\hline \hline
\multicolumn{4}{||c|}{Lines 1, 3}                             & \multicolumn{4}{c||}{Line 2}                                 \\ 
\hline \hline
\multicolumn{2}{||c|}{$R_{line}$}    & \multicolumn{2}{c|}{$0.23$ $\Omega$}   & \multicolumn{2}{c|}{$R_{line}$}  & \multicolumn{2}{c||}{$0.35$ $\Omega$}     \\ 
\hline
\multicolumn{2}{||c|}{$L_{line}$}    & \multicolumn{2}{c|}{$318$ $\mu$H} & \multicolumn{2}{c|}{$L_{line}$}  & \multicolumn{2}{c||}{$1847$ $\mu$H}  \\ 
\hline \hline
\multicolumn{2}{||c|}{Load 1}   & \multicolumn{2}{c|}{Load 2} & \multicolumn{2}{c|}{Load 3} & \multicolumn{2}{c||}{Load 4}   \\ 
\hline \hline
$R$ & $30$ $\Omega$                         & $R$ & $20$ $\Omega$                      & $R$ & $25$ $\Omega$                      & $R$ & $25$ $\Omega$                        \\ 
\hline
$X$ & $15$ $\Omega$                         & $X$ & $10$ $\Omega$                      & $X$ & $10$ $\Omega$                     & $X$ & $15$ $\Omega$                       \\
\hline \hline
\end{tabular}
\label{table_II}
\vspace{-3mm}
\end{table}

\begin{figure}[t]
    \centering
        \includegraphics[width=0.335\textwidth]{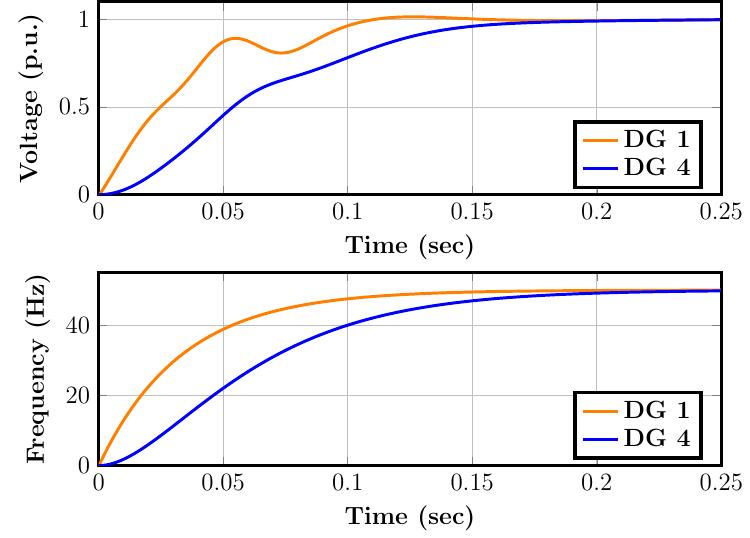}
        \vspace{-2mm}
    \caption{Attack-free scenario.}
    \label{no_attack}
\end{figure}

\begin{figure}[t]
\vspace{-2mm}
\centering
    \subfloat[]{
        \includegraphics[width=0.2315\textwidth]{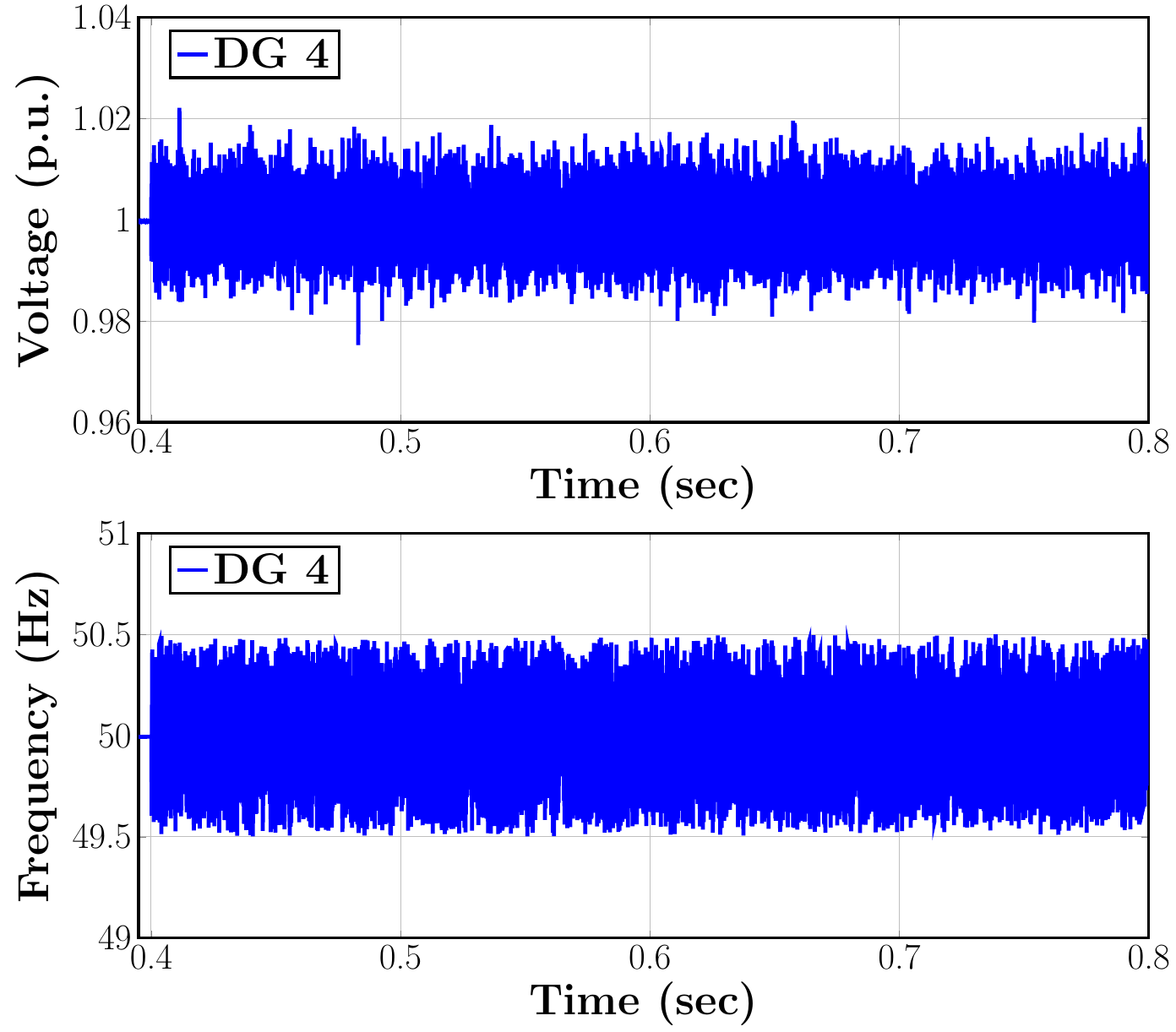}
        \label{fig:FDI_attack_DG4_INPUTS}
    } 
    \subfloat[]{
        \includegraphics[width=0.2315\textwidth]{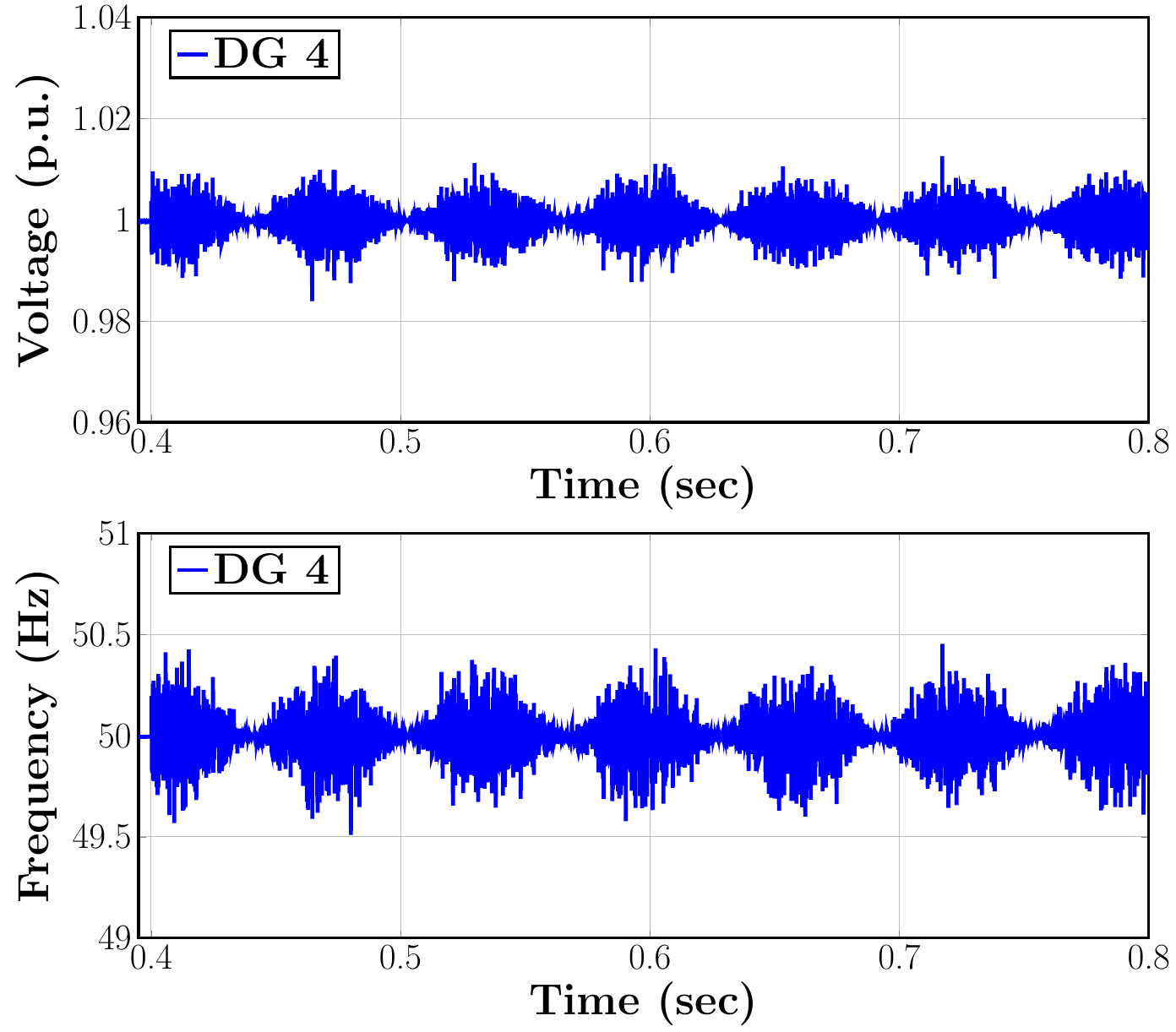}
        \label{fig:SINUSOIDAL_attack_DG4_INPUTS}
    } \\
\vspace{-2mm}
\caption[CR]{Effects (at the output of DG 4) of an arbitrary attack vector following for both voltage and frequency:  (\subref{fig:FDI_attack_DG4_INPUTS}) a uniform distribution with $\mu[-0.01, 0.01]$, and (\subref{fig:SINUSOIDAL_attack_DG4_INPUTS}) a normal distribution with $\mathcal{N}(0, 0.00001)$ multiplied with a sine carrier of reference values ($1$ p.u., $50$ Hz).}
\vspace{-5mm}
\label{fig:attacks_inputs}
\end{figure}

The vector of algebraic equations, denoted by $\boldsymbol{\Tilde{g}}$, includes the algebraic variables represented by $\boldsymbol{z}=[\left(\boldsymbol{i}_{\mathrm{od }}\right)^{\top},\left(\boldsymbol{i}_{\mathrm{oq}}\right)^{\top}]^{\top} \mathbb{R}^{n_{\mathrm{z}} \times 1}$.
\begin{equation}\label{eq:iod}
 {\boldsymbol{z}}^{\top}=\left[\begin{array}{c} \operatorname{Re}\left\{\check{\boldsymbol{Y}}\left(\boldsymbol{v}_{o}-\boldsymbol{i}_{o} \cdot \boldsymbol{z}_c\right)\right\} \\ \operatorname{Im}\left\{\check{\boldsymbol{Y}}\left(\boldsymbol{v}_{o}-\boldsymbol{i}_{o} \cdot \boldsymbol{z}_c\right)\right\}  \end{array}\right]   
\end{equation}
\noindent where $\check{Y}$ is the Kron-reduced admittance matrix of the system \cite{luo2014spatiotemporal} and $\boldsymbol{z}_{c}$ is the line impedance connecting an inverter to the MG. The inverter terminal's voltage is given as: 
\begin{equation}\label{eq:vo}
\boldsymbol{v}_{\mathrm{o}}=\left(\boldsymbol{v}_{\mathrm{opf}}+\boldsymbol{n}_Q \cdot \boldsymbol{q}_{\mathrm{opf}}-\boldsymbol{n}_Q \cdot \boldsymbol{q}\right) \cdot(\cos \boldsymbol{\delta}+\mathrm{j} \sin \boldsymbol{\delta})
\end{equation}
where $\boldsymbol{n}_Q$ is $q-v$ droop constant, and $\boldsymbol{q}_{\mathrm{opf}}$ and $\boldsymbol{v}_{\mathrm{opf}}$ are the OPF resulted values for the reactive power and voltage setpoints, respectively.  

\begin{proposition}
Let us consider a MG described by \eqref{eq:MG-modeling-stability1}--\eqref{eq:MG-modeling-stability2}. The state matrix, denoted by $\hat{\boldsymbol{A}} \in \mathbb{R}^{n_{\mathrm{x}} \times n_{\mathrm{x}}}$, is defined as  $\hat{\boldsymbol{A}}=\boldsymbol{A}\left(\boldsymbol{i}_{\mathrm{od}}, \boldsymbol{i}_{\mathrm{oq}}, \boldsymbol{\delta}, \boldsymbol{q}\right)-\boldsymbol{B}(\boldsymbol{\delta}, \boldsymbol{q})(\boldsymbol{D})^{-1} \boldsymbol{C}(\boldsymbol{\delta}, \boldsymbol{q})
$, where $\boldsymbol{A}\left(\boldsymbol{i}_{\mathrm{od}}, \boldsymbol{i}_{\mathrm{oq}}, \boldsymbol{\delta}, \boldsymbol{q}\right)=\frac{\partial \boldsymbol{\Tilde{f}}}{\partial \boldsymbol{x}}$, $\boldsymbol{B}(\boldsymbol{\delta}, \boldsymbol{q})=\frac{\partial \boldsymbol{\Tilde{f}}}{\partial \boldsymbol{z}}$, $\boldsymbol{C}(\boldsymbol{\delta}, \boldsymbol{q})=\frac{\partial \boldsymbol{\Tilde{g}}}{\partial \boldsymbol{x}}$, and $\boldsymbol{D}=\frac{\partial \boldsymbol{\Tilde{g}}}{\partial \boldsymbol{z}}$. The expressions for these matrices are given in the Appendix by \eqref{eq:df-dx} -- \eqref{eq:dg-dz}, along with the definition of the matrix $\boldsymbol{M}^{\mathrm{p}}$ in \eqref{eq:Mmatrix}. Furthermore, $\check{G}$ and $\check{B}$ represent the real and imaginary parts, respectively, of the Kron-reduced admittance matrix $\check{\boldsymbol{Y}}$. According to \cite{boyd1994linear}, the small-signal stability of the MG can be ensured -- incorporating the Lyapunov stability, with a minimum decay rate (i.e., damping ratio) of $\eta$, if there exists a symmetric positive definite matrix $M$ that satisfies:
\begin{equation}\label{eq:proposition-stability}
\hat{\boldsymbol{A}}^{\top} \boldsymbol{M}+\boldsymbol{M} \hat{\boldsymbol{A}} \preceq-2 \eta \boldsymbol{M} .
\end{equation}
\end{proposition}
\noindent The proof for Proposition 1 can be found in \cite{pullaguram2021small}.

Additional constraints are formulated to connect the inverter's internal variables with the OPF variables:
\begin{IEEEeqnarray}{lCr}\label{eq:21}
& \boldsymbol{v}^{\mathrm{g}}=\left(\boldsymbol{v}_{\mathrm{ref}}-\boldsymbol{n}_Q \cdot \boldsymbol{q}\right) \cdot(\cos \boldsymbol{\delta}+\mathrm{j} \sin \boldsymbol{\delta})-\boldsymbol{i}_{o} \cdot \boldsymbol{z}_c \\
& \boldsymbol{i}_{\mathrm{o}}=\breve{\boldsymbol{Y}} \boldsymbol{v}^{\mathrm{g}} \label{eq:22} \\
& \boldsymbol{p}+\mathrm{i} \boldsymbol{q}=\operatorname{diag}\left\{\boldsymbol{i}_{o}\left(\boldsymbol{i}_{o}\right)^*\left[\boldsymbol{z}_{\mathrm{c}}\right]\right\}+\boldsymbol{p}^g+\mathrm{i} \boldsymbol{q}^g \label{eq:23}
\end{IEEEeqnarray}

The OPF represented by \eqref{min-pg}-\eqref{voltage-constr}, along with the additional constraints in \eqref{eq:proposition-stability}-\eqref{eq:23}, constitutes a small-signal stability-constrained OPF for an inverter-dominant MG. In combination with the utilization of the designed observer \eqref{observer}, with $L(\hat{\boldsymbol{x}}) \triangleq  L' + L''(\hat{\boldsymbol{x}})$, and any $L'$, $L''(\hat{\boldsymbol{x}})$ chosen as described in the proof of Theorem \ref{theorem_observer} and able to globally asymptotically reconstruct the state $\boldsymbol{x}$, ensures that the operator can counter stealthy attacks while ensuring a sufficient stability margin during optimal generation. It should be noted that the constraints in \eqref{nodal-balance}-\eqref{voltage-constr} are quadratic functions of the bus voltage $\boldsymbol{v}$, while the state matrix $\hat{\boldsymbol{A}}$ in \eqref{eq:proposition-stability} and the constraints \eqref{eq:21} and \eqref{eq:23} are nonlinear functions of $i_{o}$, $q$, and $\delta$. These bilinear matrix inequalities (BMIs) and nonlinearities make the problem nonconvex which can be solved using computationally tractable semidefinite programming (SDP) or parabolic relaxation techniques \cite{kheirandishfard2018convex1}, or even sequentially through objective penalization \cite{pullaguram2021small}.

\vspace{-2mm}
\section{Experimental and Simulation Results}\label{s:result}

\begin{figure*}[htbp!]
\centering
    \subfloat[]{
        \includegraphics[width=0.23\textwidth]{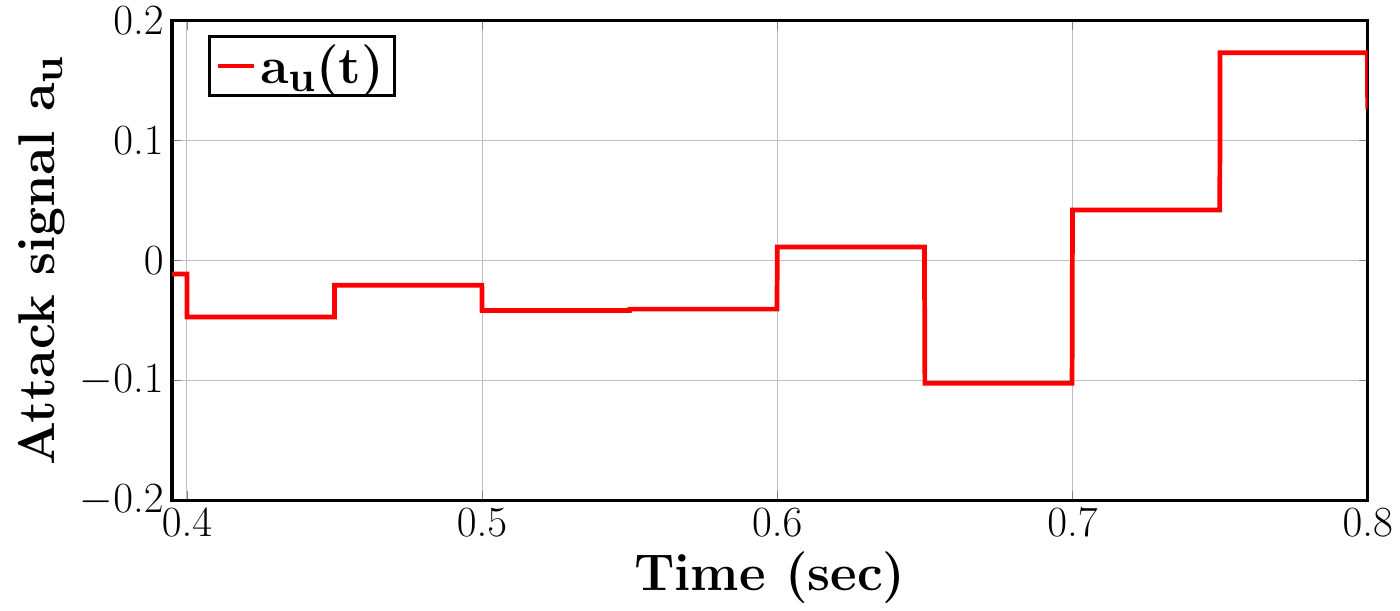}
        \label{fig:stealthy_attack_au_b_0.2}
    } 
    \subfloat[]{
        \includegraphics[width=0.23\textwidth]{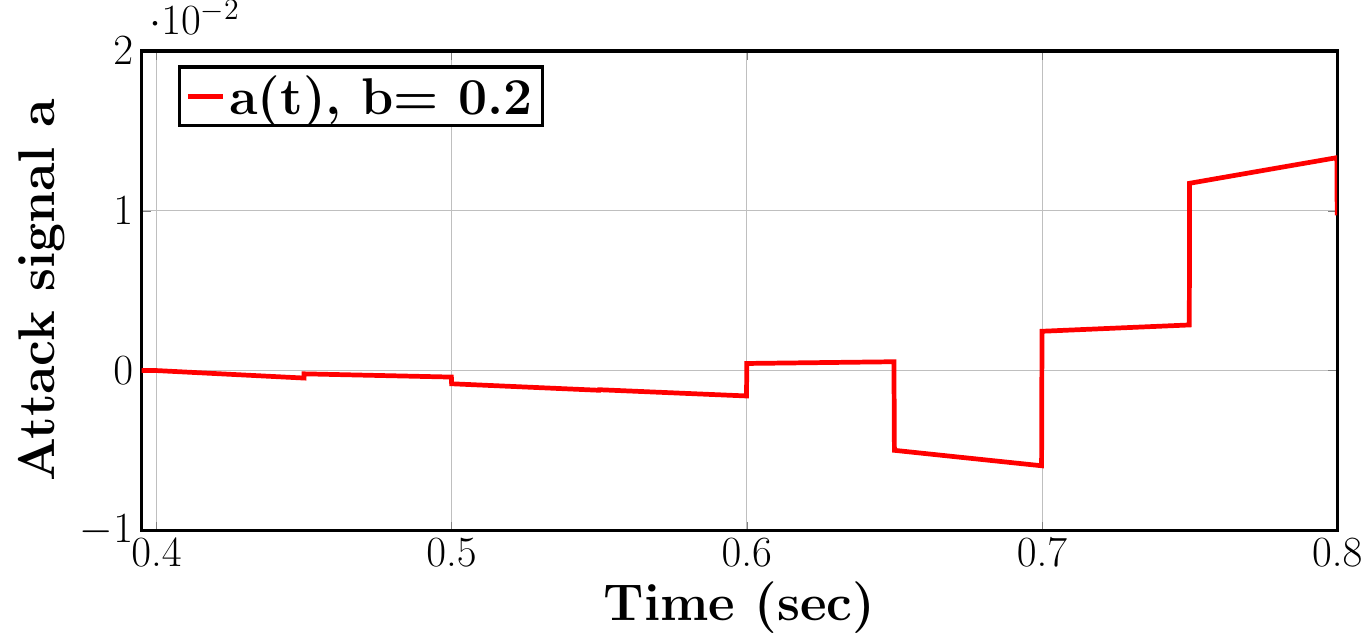}
        \label{fig:stealthy_attack_a_b_0.2}
    } 
    \subfloat[]{
        \includegraphics[width=0.23\linewidth]{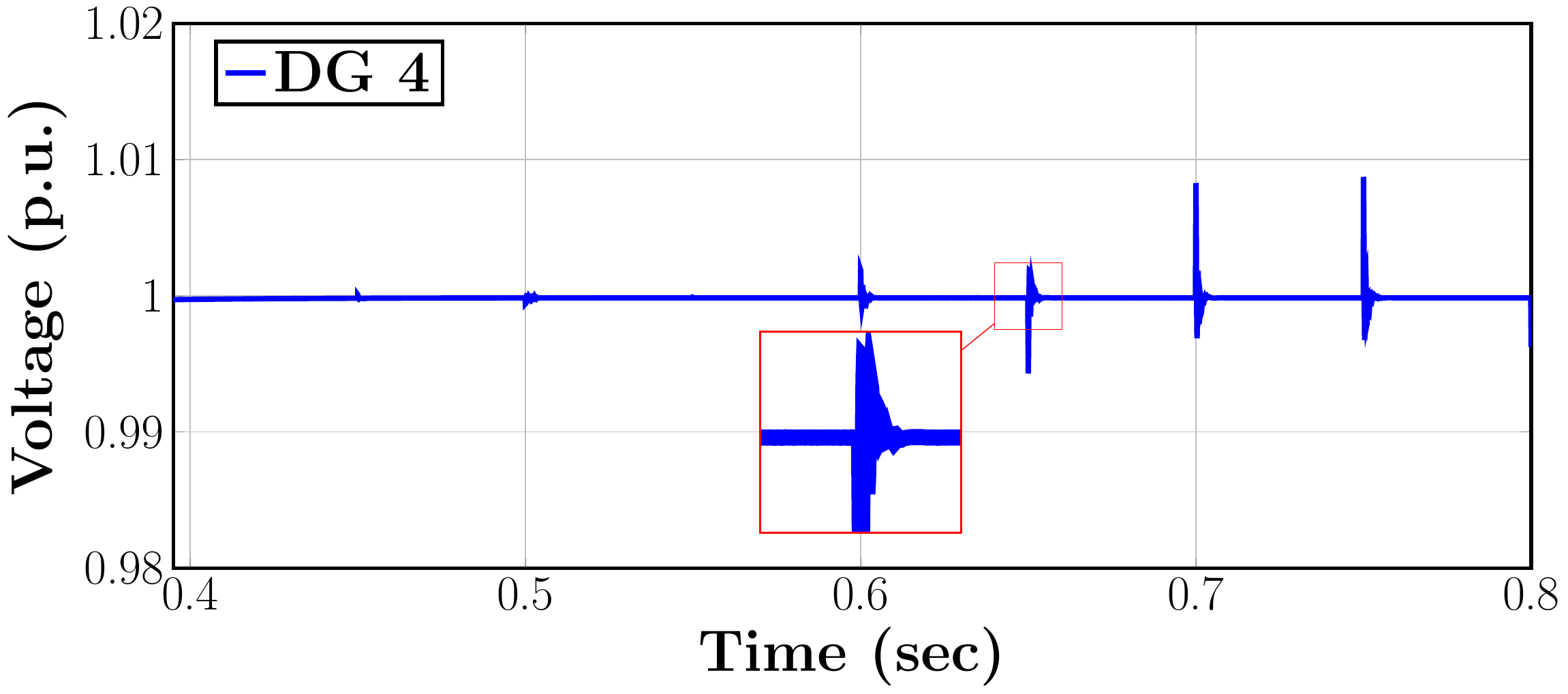}
        \label{fig:stealthy_attack_voltage_b_0.2}
    } 
   \subfloat[]{
        \includegraphics[width=0.23\linewidth]{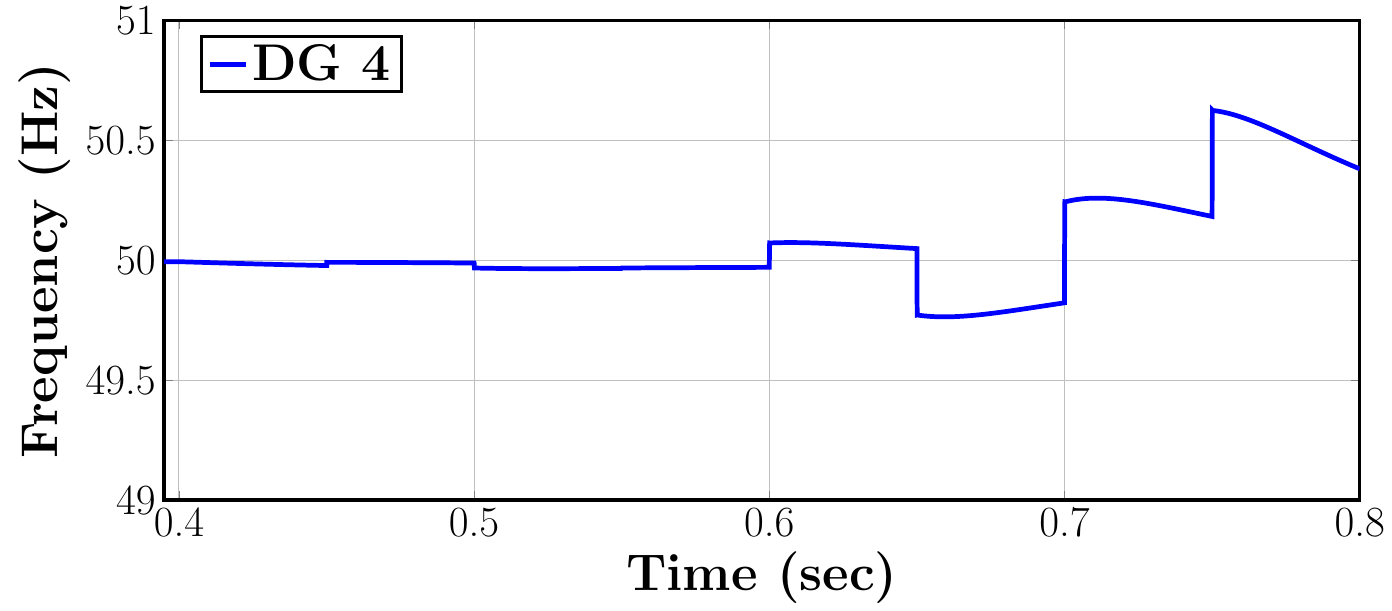}
        \label{fig:stealthy_attack_frequency_b_0.2}
    } \\
\vspace{-2mm}
\caption[CR]{Demonstration and effects of a stealthy attack at the outputs of DG 4 with $b_{i} = 0.2$ as the attack evolution rate.} 
\vspace{-3mm}
\label{fig:stealthy_attack_input}
\end{figure*}

\begin{figure*}[hbtp!]
\centering
    \subfloat[]{
        \includegraphics[width=0.23\textwidth]{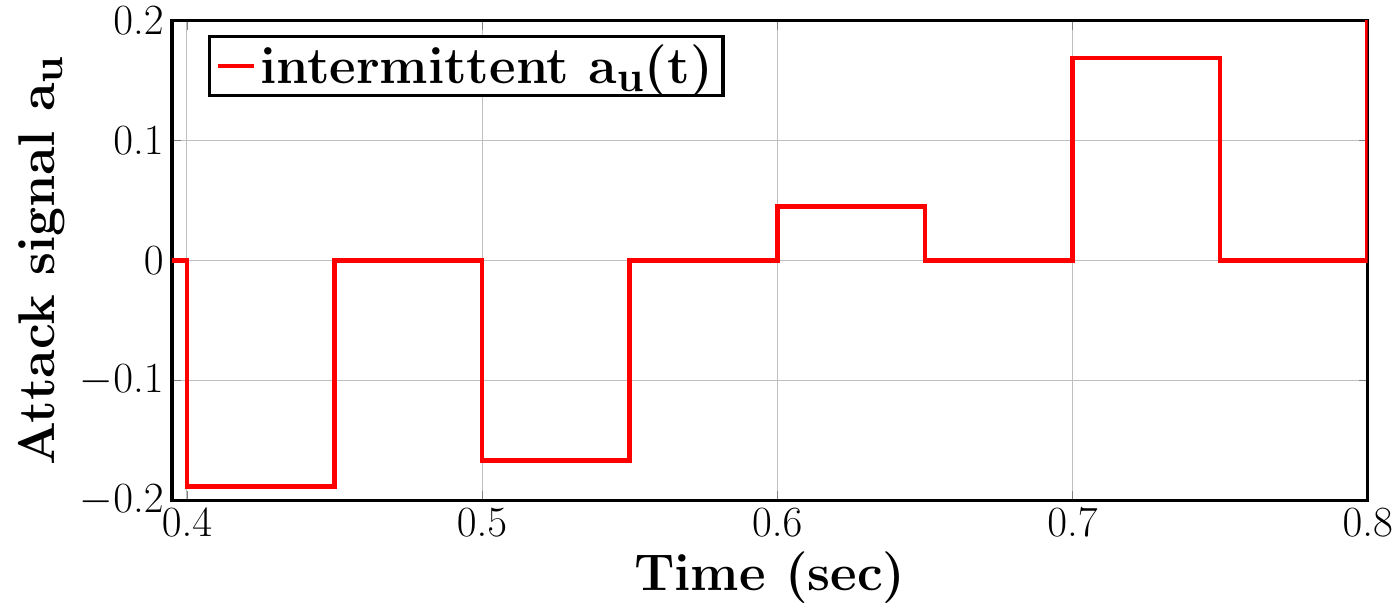}
        \label{fig:stealthy_intermittent_attack_au_b_0.2}
    } 
    \subfloat[]{
        \includegraphics[width=0.23\textwidth]{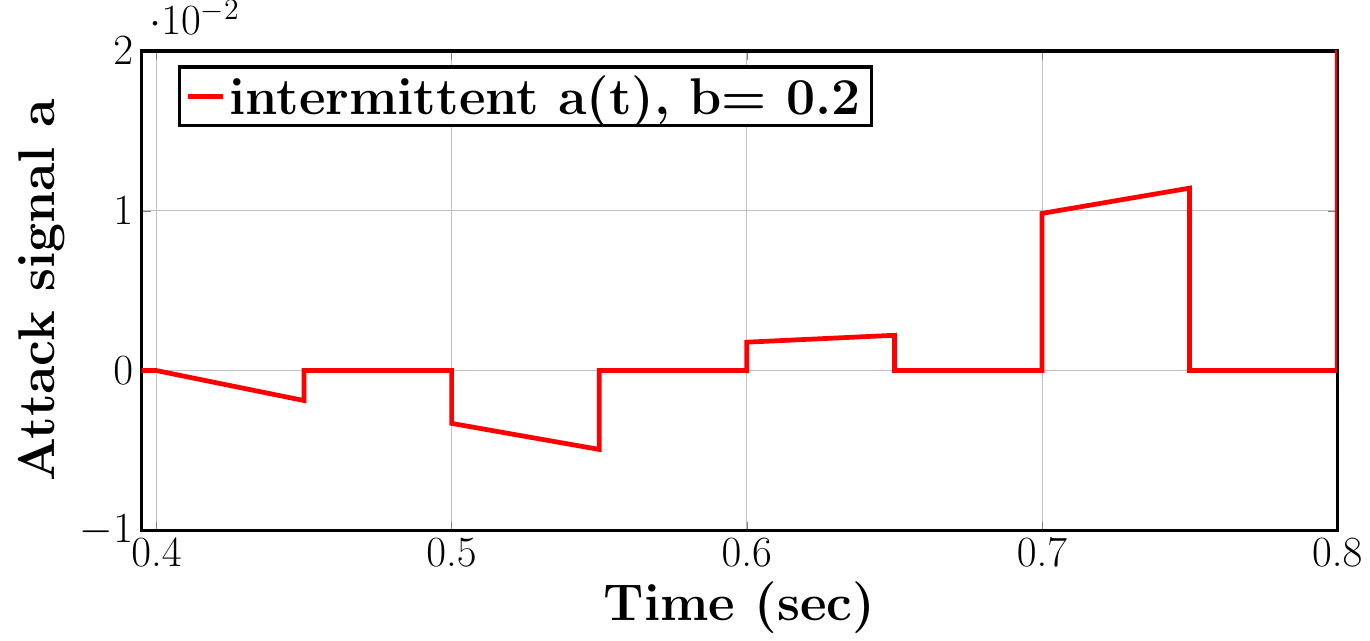}
        \label{fig:stealthy_intermittent_attack_a_b_0.2}
    } 
    \subfloat[]{
        \includegraphics[width=0.23\linewidth]{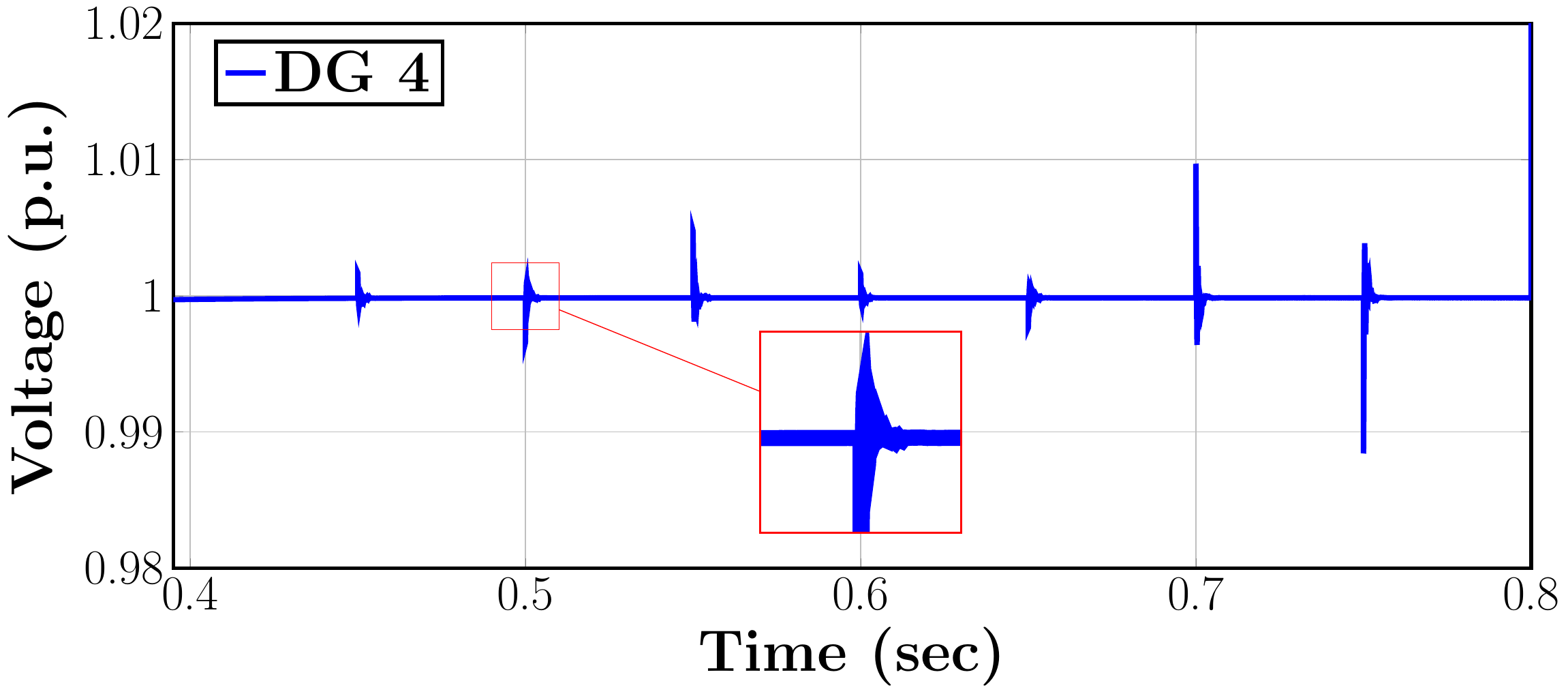}
        \label{fig:stealthy_intermittent_attack_voltage_b_0.2}
    } 
   \subfloat[]{
        \includegraphics[width=0.23\linewidth]{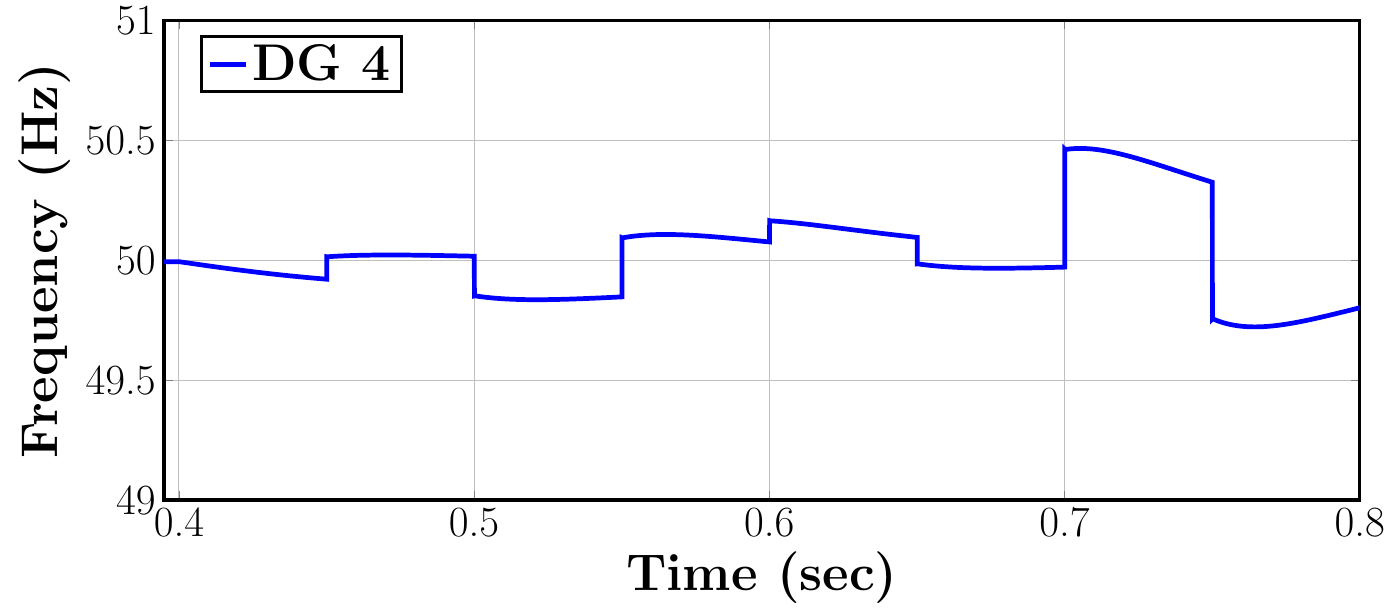}
        \label{fig:stealthy_intermittent_attack_frequency_b_0.2}
    } \\
\vspace{-2mm}
\caption[CR]{Demonstration and effects of a stealthy and intermittent attack at the outputs of DG 4 with $b_{i} = 0.2$ as the attack evolution rate.} 
\vspace{-5mm}
\label{fig:stealthy_intermittent_integrity_attack_input}
\end{figure*}

In this section, the simulations results are presented demonstrating the impact of attacks discussed in Section \ref{s:attackformulation} and the effectiveness of the proposed detection and mitigation strategy of Section \ref{s:methodology}. The MG adopted for testing is shown in Fig.  \ref{fig1a}, and the relevant MG modeling parameters are detailed in Table \ref{table_I} and Table \ref{table_II}. Simulations have been performed with MATLAB 2020.b, on a Macbook Pro 2021, with 16 GB of RAM and a 512 GB hard disk drive.

\vspace{-2mm}

\subsection{Effectiveness of Attack Models}


Fig. \ref{no_attack} displays the voltage and frequency values in the attack-free scenario. The reference values are set to $1$ per unit (p.u.) and $50$ Hz, respectively. The figure shows that the secondary controller operates properly, i.e., voltage and frequency values reach the reference values. {The magnitudes of voltage and frequency in the DGs present instabilities after attacks are introduced into the system. This effect is produced by tampering the input references $\omega_j$ and $v_{odj}$ provided to the secondary controller. The attack duration is assumed to be bounded and last from $0.4s$ to $0.8s$. This action leads to distorting the active and reactive power values of the DGs inside the MG. The plots presented next are only for DG 4 for brevity. Similar results can be obtained for other DGs.}


\begin{figure}[t]
\centering
    \subfloat[]{
        \includegraphics[width=0.2315\textwidth]{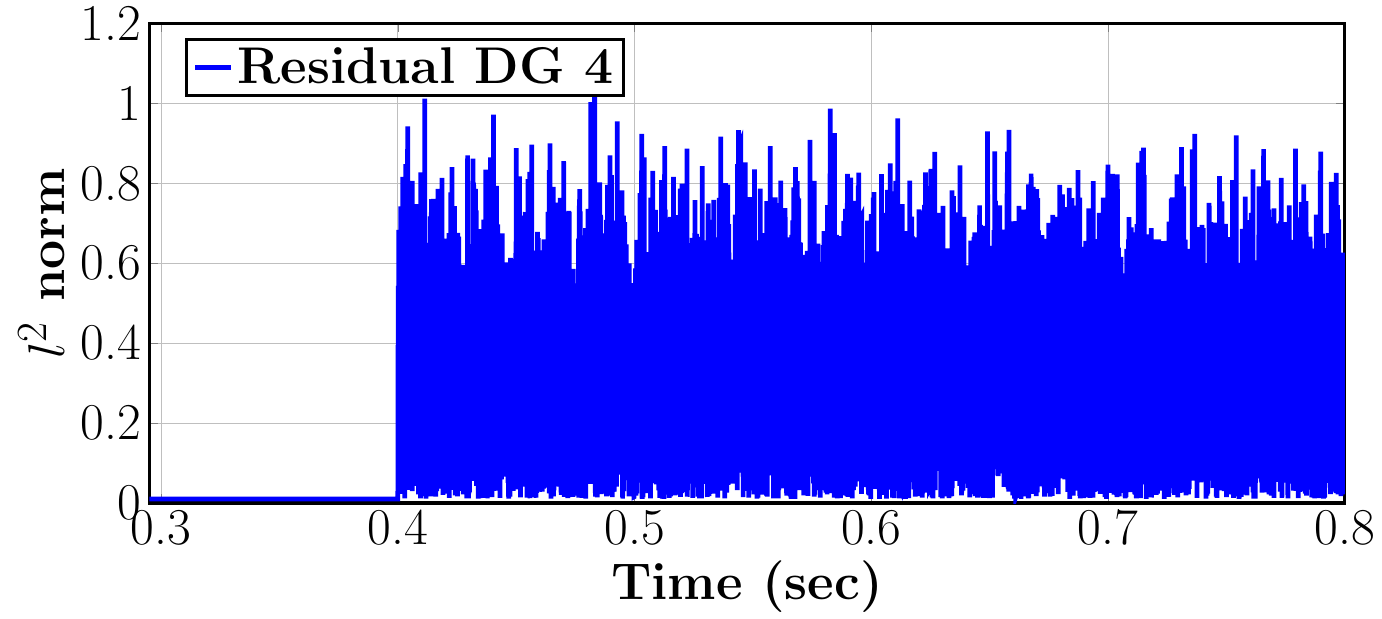}
        \label{fig:residual_FDI_inputs_DG4}
    } 
    \subfloat[]{
        \includegraphics[width=0.2315\textwidth]{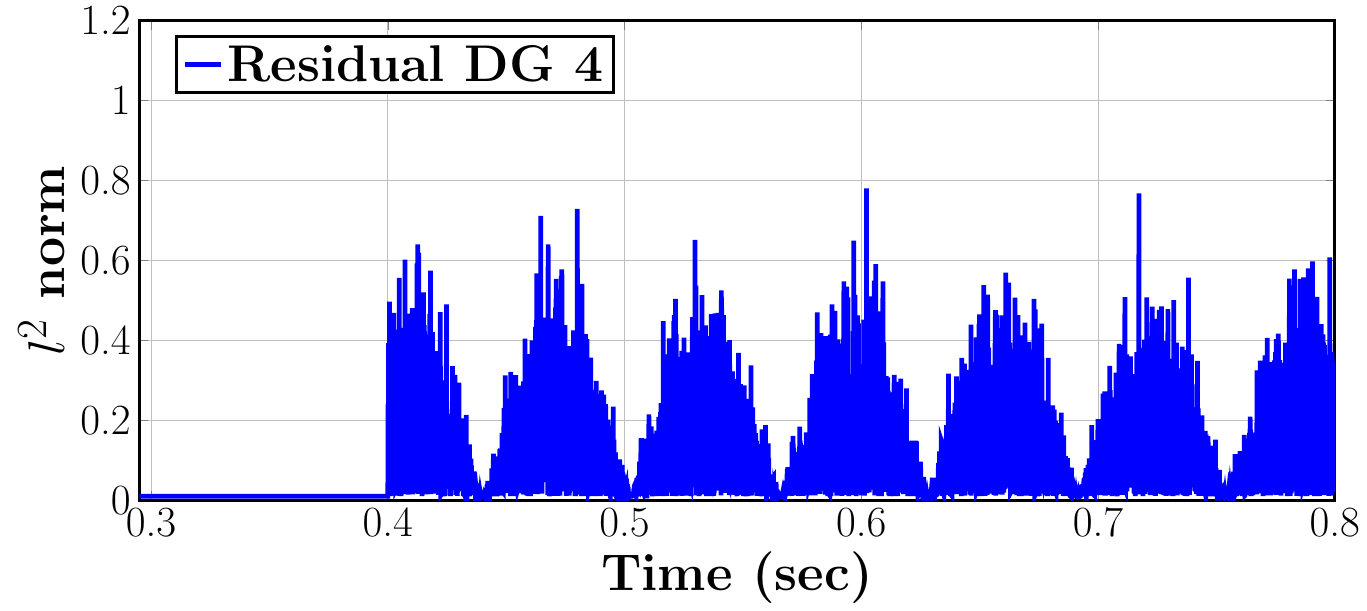}
        \label{fig:residual_sinusoidal_normal_inputs_DG4}
    } \\
   
\vspace{-1mm}  
\caption[CR]{${l}^2$ norm of the proposed residual under an arbitrary attack vector  
following for both voltage and frequency the attack configuration demonstrated in Fig. \ref{fig:attacks_inputs}.} 
\vspace{-2mm}
\label{fig:residual_random_attack_vector}
\end{figure}

\begin{figure}[t]
\centering
    \subfloat[]{
        \includegraphics[width=0.2315\textwidth]{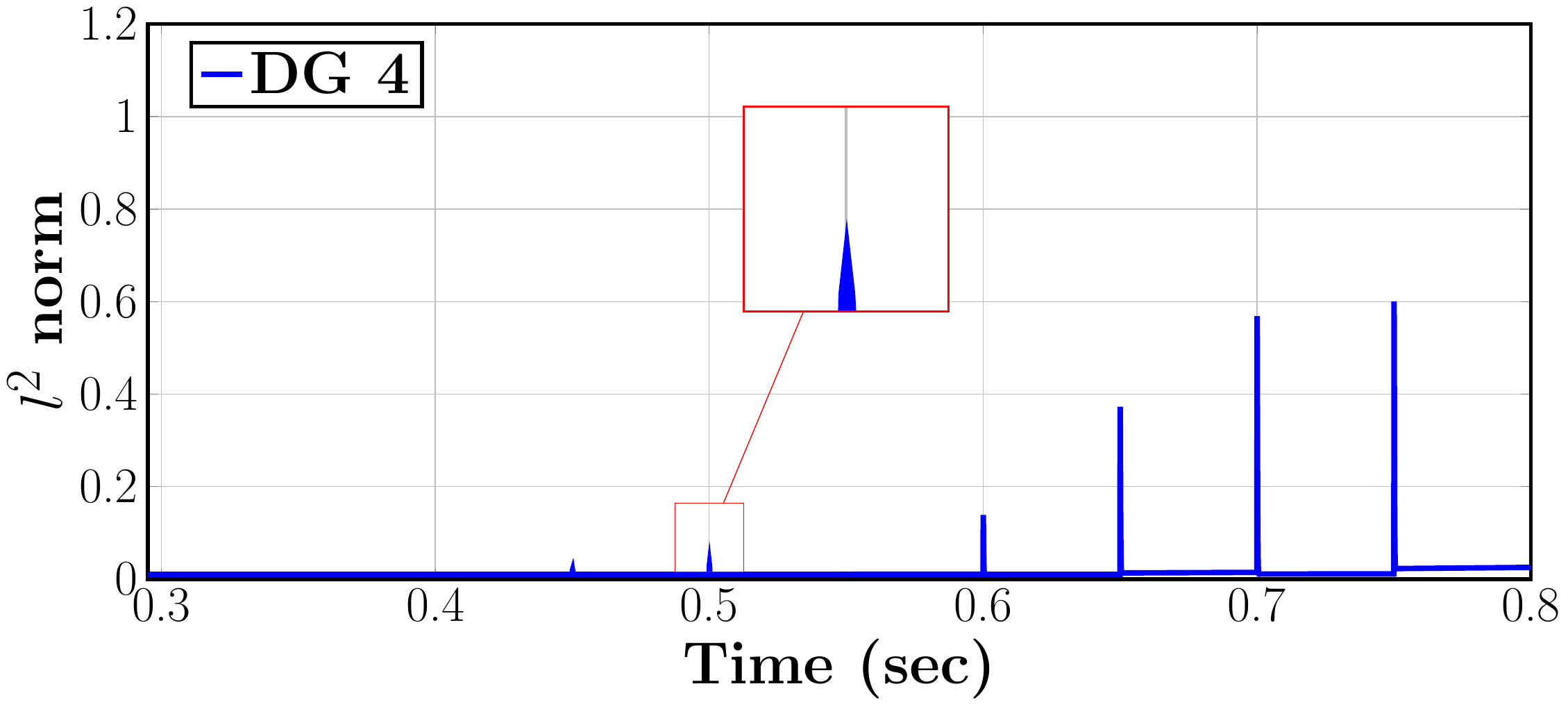}
        \label{fig:residual_stealthy_0.2}
    } 
    \subfloat[]{
        \includegraphics[width=0.2315\textwidth]{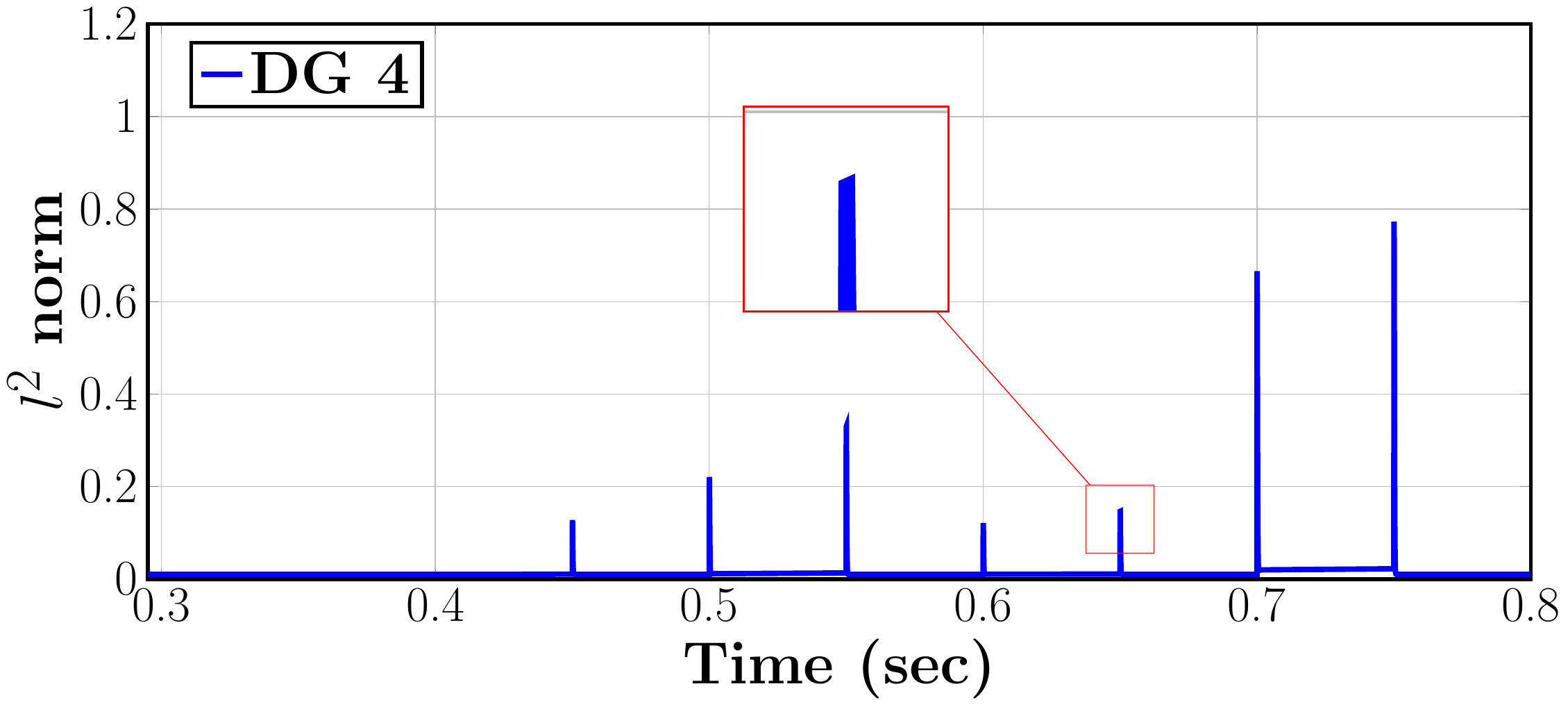}
        \label{fig:residual_stealthy_intermittent_0.2}
    } 
\\
\vspace{-2mm}    
\caption[CR]{${l}^2$ norm of the proposed residual under: (\subref{fig:residual_stealthy_0.2}) a stealthy only attack, and ((\subref{fig:residual_stealthy_intermittent_0.2})) a stealthy and intermittent attack, with $b_{i} = 0.2$. } 
\vspace{-3mm}
\label{fig:residuals_inputs_stealthy_attack_vector}
\end{figure}

Figs. \ref{fig:attacks_inputs}, \ref{fig:stealthy_attack_input}, \ref{fig:stealthy_intermittent_integrity_attack_input} present the simulation results under three different attack models: \textit{(a)} an arbitrary attack vector following a uniform and normal-sinusoidal distribution (Fig. \ref{fig:attacks_inputs}), \textit{(b)} an attack following the stealthy but not the intermittent character of the attack modeling of Section \ref{s:attackformulation} (Fig. \ref{fig:stealthy_attack_input}), and \textit{(c)} an attack following the complete proposed model of Section \ref{s:attackformulation} (Fig. \ref{fig:stealthy_intermittent_integrity_attack_input}). In Figs. \ref{fig:stealthy_attack_input} and \ref{fig:stealthy_intermittent_integrity_attack_input}, we selected empirically  within the attack function dynamics, $\beta_{i}(t- T_{o})$ in \eqref{attack interval 2}, as an attack evolution rate $b_{i}=0.2$. Increasing substantially the evolution rate of the attack will cause the attack to lead to much higher residual errors. On the other hand, lower values of $b_i$ will have minimal effect on the output voltage and frequency of the secondary controller, i.e., will minimize the attack model objectives. Overall, the attack vectors  intend to affect the system's voltage and frequency stability while tampering the pre-defined setpoints of the MG. It is clear that the constraints of the scheduling interval as well as the stealthiness of the attack slightly reduce  the effect on the secondary control output (especially on the frequency range). 

\vspace{-2mm}
\subsection{Residual-based Observer Comparison}


\begin{figure}[t]
\centering
    \subfloat[]{
        \includegraphics[width=0.2315\textwidth]{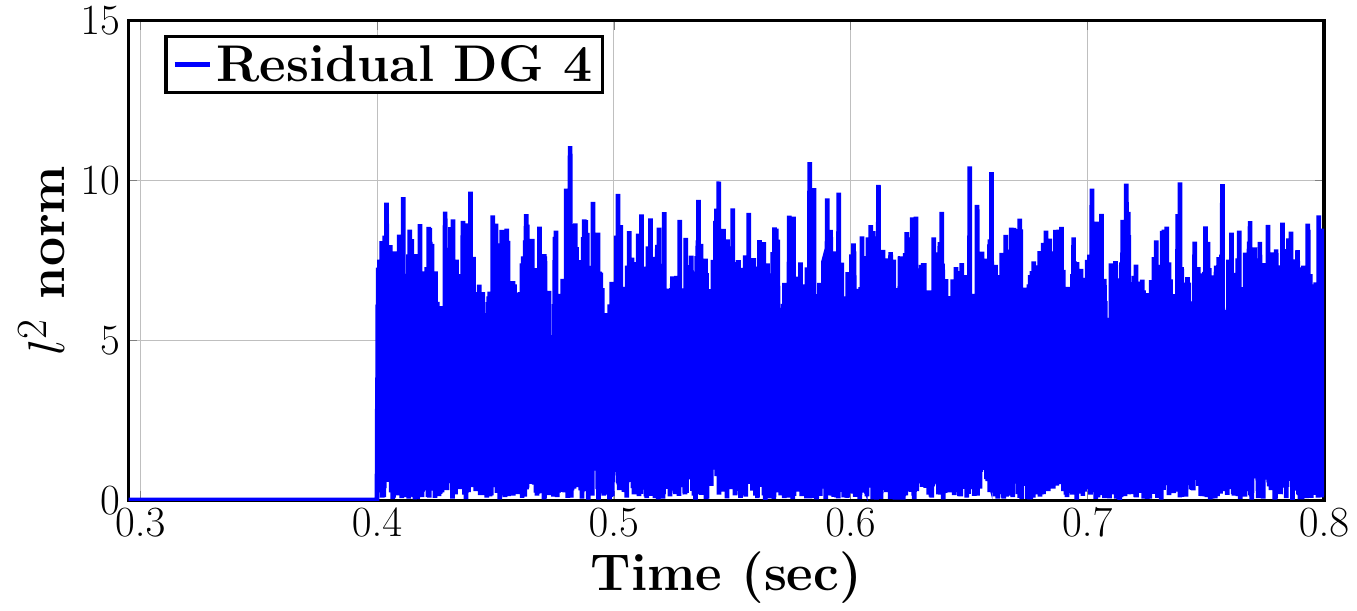}
        \label{fig:FDI_L2norm_simple_obs}
    } 
    \subfloat[]{
        \includegraphics[width=0.2315\textwidth]{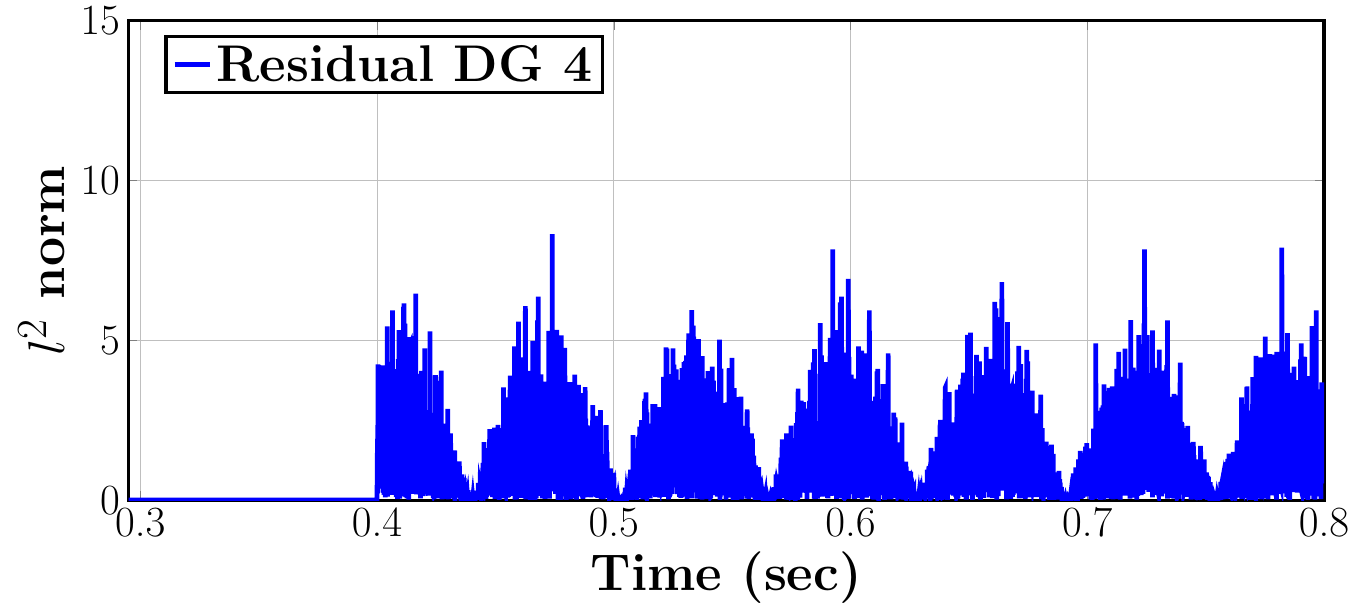}
        \label{fig:sinusoidal_L2norm_simple_obs}
    } \\
\vspace{-2mm}    
\caption[CR]{${l}^2$ norm of the residual acquired via the output injection observer (Remark \ref{remark_obs1}) under an arbitrary attack vector following for both voltage and frequency the attack configuration demonstrated in Fig. \ref{fig:attacks_inputs}.} 
\label{fig:l2norm_random_attack_simple_obs}
\end{figure}

\begin{figure}[t]
\vspace{-2mm}
\centering
    \subfloat[]{
        \includegraphics[width=0.2315\textwidth]{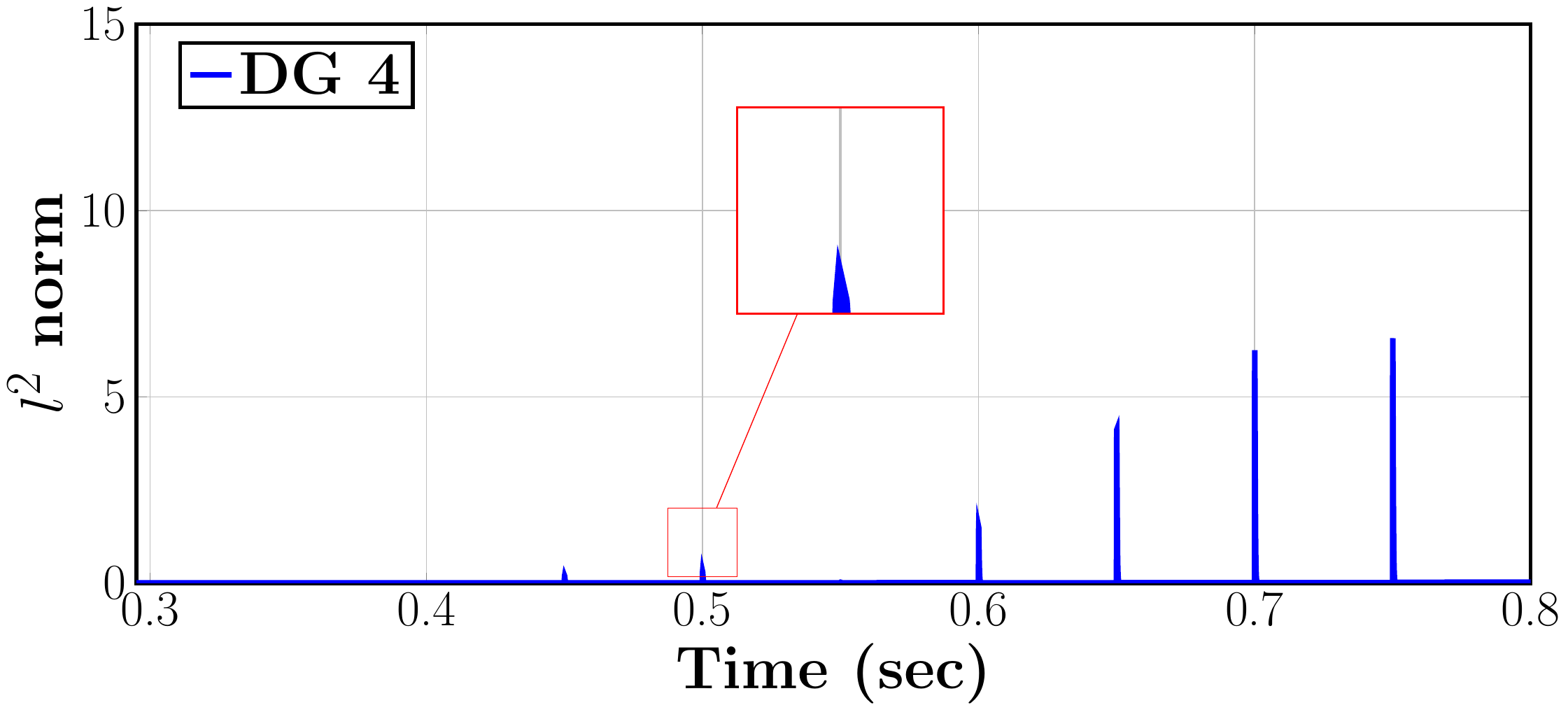}
        \label{fig:residual_stealthy_0.2_simple_observer}
    }
    \subfloat[]{
        \includegraphics[width=0.2315\textwidth]{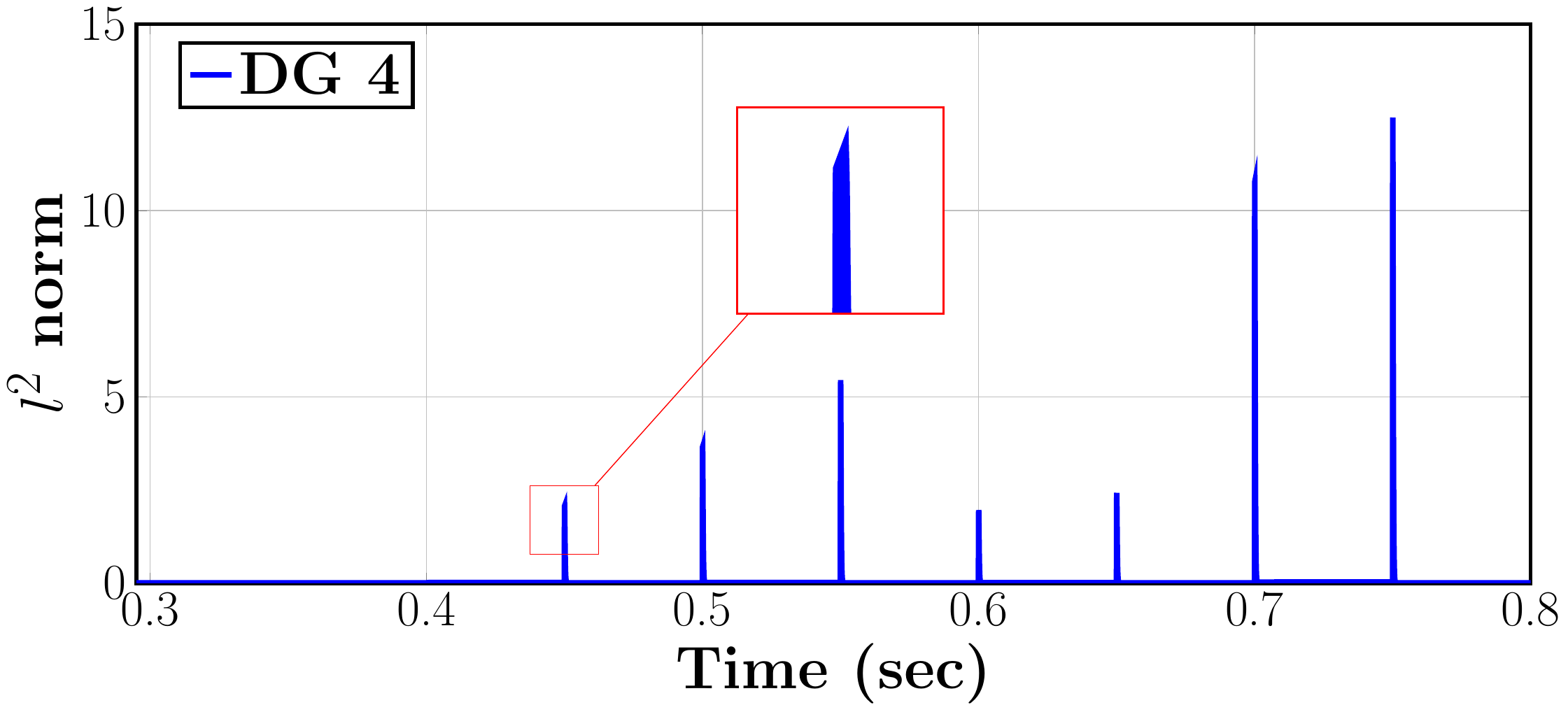}
        \label{fig:residual_stealthy_intermittent_0.2_simple_observer}
    } 
\\
\vspace{-2mm}    
\caption[CR]{${l}^2$ norm of the residual acquired via the output injection observer (Remark \ref{remark_obs1}) under: (\subref{fig:residual_stealthy_0.2_simple_observer}) a stealthy only attack, and ((\subref{fig:residual_stealthy_intermittent_0.2_simple_observer})) a stealthy and intermittent attack, with $b_{i} = 0.2$ as the attack evolution rate. }
\vspace{-3mm}
\label{fig:residuals_inputs_stealthy_attack_vector_simple_obs}
\end{figure}

The importance of observers relies on computing an error, described as the residual variable $\boldsymbol{r}$, presented in Theorem \ref{theorem_observer}. This error indicates that the MG is experimenting an anomalous behavior. The norm of the error vector, e.g., the Euclidean ${l}^2$ norm (i.e., $\|\boldsymbol{r}\|$), can be computed to measure the error. The error acts as an alarm for this work's residual-based detection approach.  This paper focuses on effectively detecting attacks, and leverages a residual-threshold based approach. Similar to presented literature, where corrupted system states could be recovered with a selected detection threshold, in our experiments we are able to correctly identify corrupted values with a similar empirically-selected maximum error threshold. The ${l}^2$ response of the residual is presented, without any mitigation strategy, for three set of results: \textit{(a)} the proposed residual-based observer (Theorem \ref{theorem_observer}) under an arbitrary attack vector of effects presented in Fig. \ref{fig:attacks_inputs} (Fig. \ref{fig:residual_random_attack_vector}), and \textit{(b)} the proposed residual-based observer (Theorem \ref{theorem_observer}) under the proposed stealthy and intermittent attack model of effects presented in Fig. \ref{fig:stealthy_intermittent_integrity_attack_input} (Fig. \ref{fig:residuals_inputs_stealthy_attack_vector}), \textit{(c)} an alternative residual-based observer from literature (Remark \ref{remark_obs1}) under an arbitrary attack vector of effects presented in Fig. \ref{fig:attacks_inputs} (Fig. \ref{fig:l2norm_random_attack_simple_obs}), and \textit{(d)} an alternative residual-based observer from literature (Remark \ref{remark_obs1}) under the proposed stealthy and intermittent attack model of effects presented in Fig. \ref{fig:stealthy_intermittent_integrity_attack_input} (Fig. \ref{fig:residuals_inputs_stealthy_attack_vector_simple_obs}). In all cases, the initial state of the observer is set equal to the measured state at the initial time (for the measurable state variables, while it is set to zero for all the other ones). 
From Figs. \ref{fig:residual_random_attack_vector}--\ref{fig:residuals_inputs_stealthy_attack_vector_simple_obs}, it can be observed that the DGs behavior captured by the residuals is sensitive to attack perturbations, especially when the error diverges more critically. In addition, the ${l}^2$ norm of the residual error of the proposed observer compared with the output injection observer (
Remark \ref{remark_obs1} - see Section \ref{ss:detectability-subsection}) indicates that the proposed one is more accurately estimating the system's state.

\vspace{-3mm}
\subsection{Mitigation Strategy}


\begin{figure}[t]
\centering
    \subfloat[]{
        \includegraphics[width=0.2315\textwidth]{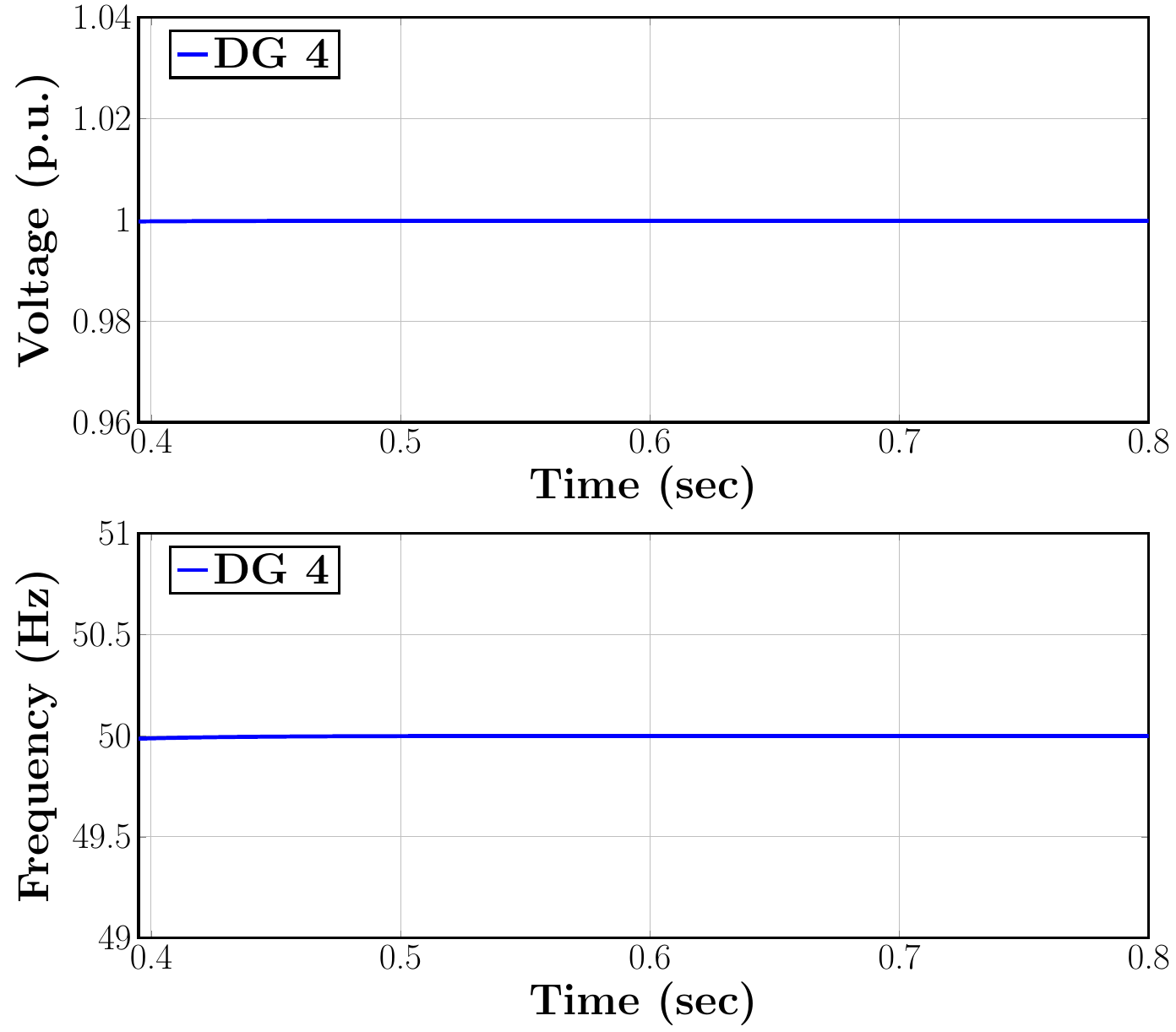}
        \label{fig:FDI_mitigation_input}
    }
    \subfloat[]{
        \includegraphics[width=0.2315\textwidth]{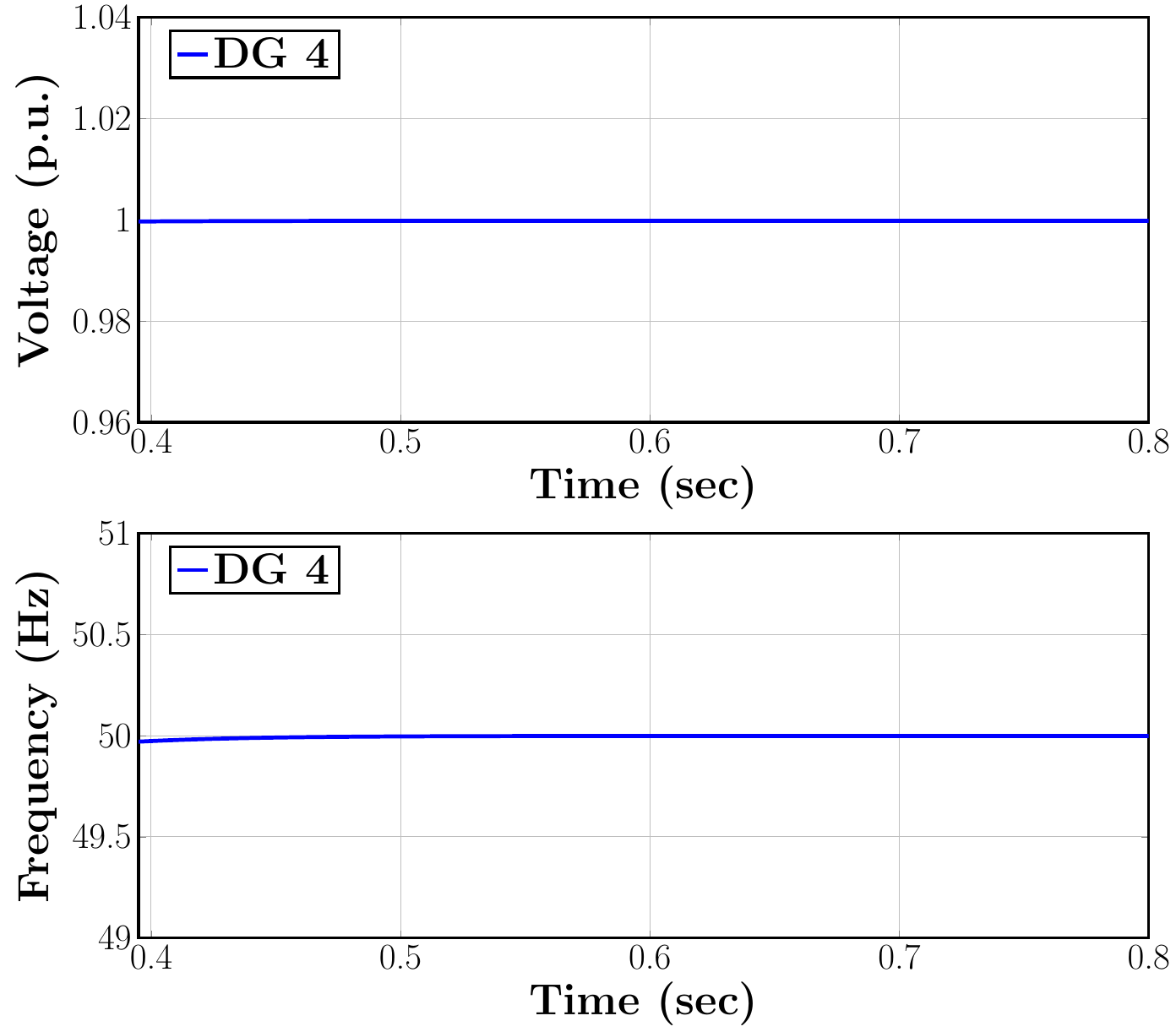}
        \label{fig:sinusoidal_mitigation_input}
    } \\ 
\vspace{-2mm}    
\caption[CR]{Detection and mitigation strategy applied at DG 4 under an arbitrary attack vector following for both voltage and frequency the attack configuration demonstrated in Fig. \ref{fig:attacks_inputs}.} 
\vspace{-3mm}
\label{fig:mitigation_inputs_random_attack_vector}
\end{figure}

\begin{figure}[t]
\centering
    \subfloat[]{
        \includegraphics[width=0.2315\textwidth]{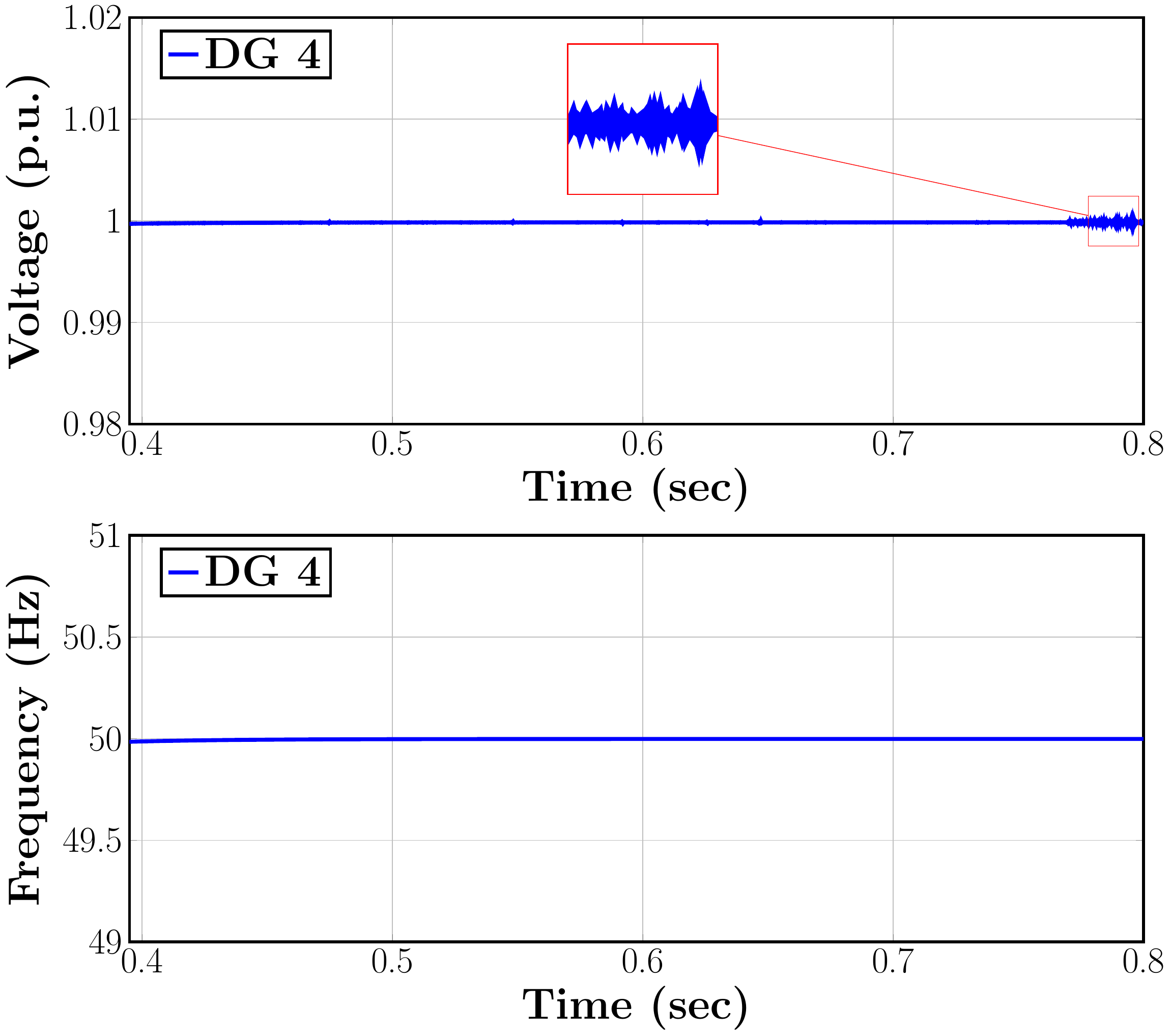}
        \label{fig:mitigation_stealthy_volt_0.2}
    } 
    \subfloat[]{
        \includegraphics[width=0.2315\textwidth]{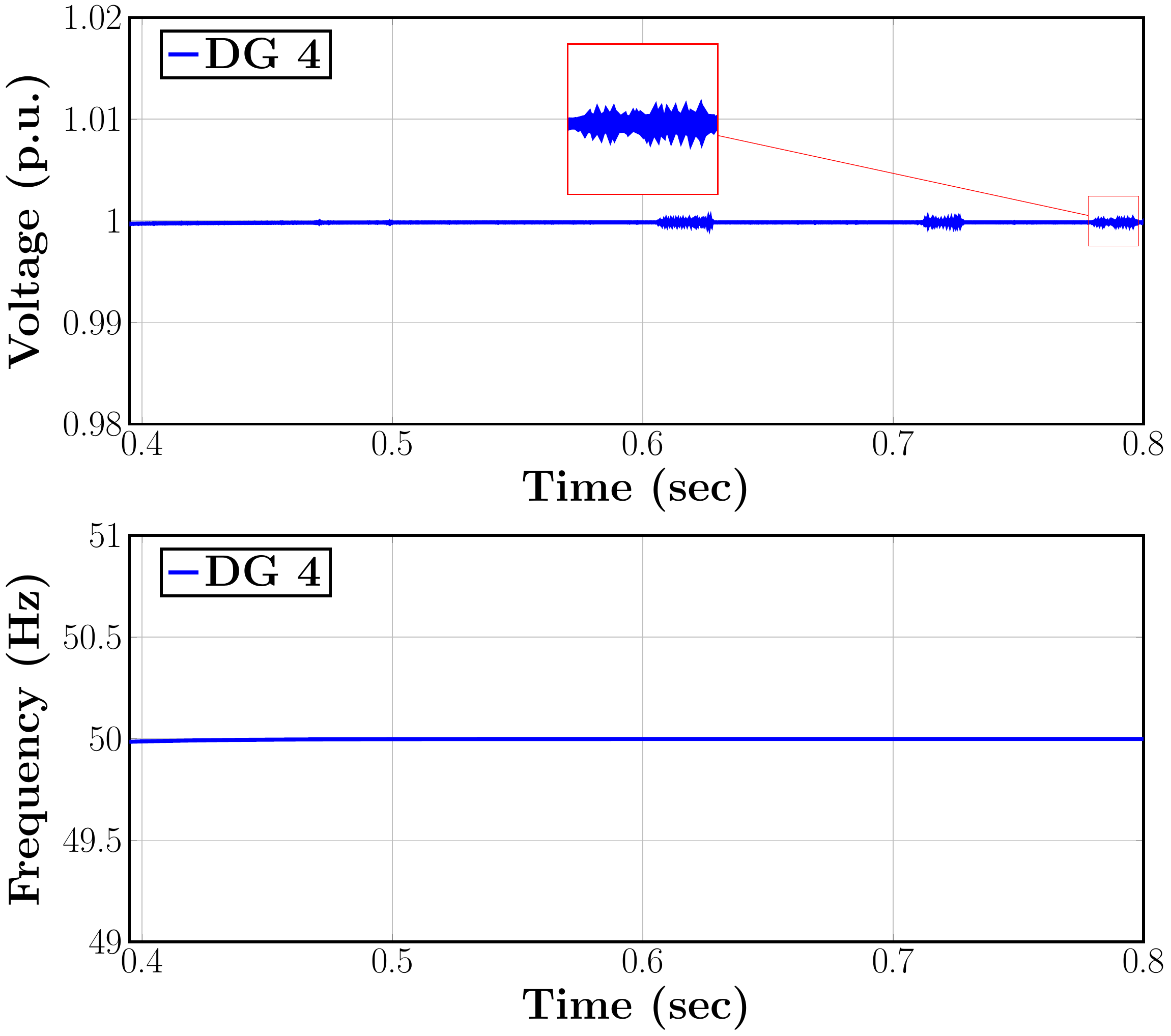}
        \label{fig:mitigation_stealthy_interm_volt_0.2}
    } 
\\
\vspace{-2mm}    
\caption[CR]{Detection and mitigation strategy applied at DG 4 under: (\subref{fig:mitigation_stealthy_volt_0.2}) a stealthy only attack, and (\subref{fig:mitigation_stealthy_interm_volt_0.2}) a stealthy and intermittent attack, with $b_{i} = 0.2$ as the attack evolution rate.} 
\vspace{-5mm}
\label{fig:mitigation_inputs_stealthy_attack_vector}
\end{figure}

\begin{figure*}[t]
\centering
    \subfloat[]{
        \includegraphics[width=0.235\textwidth]{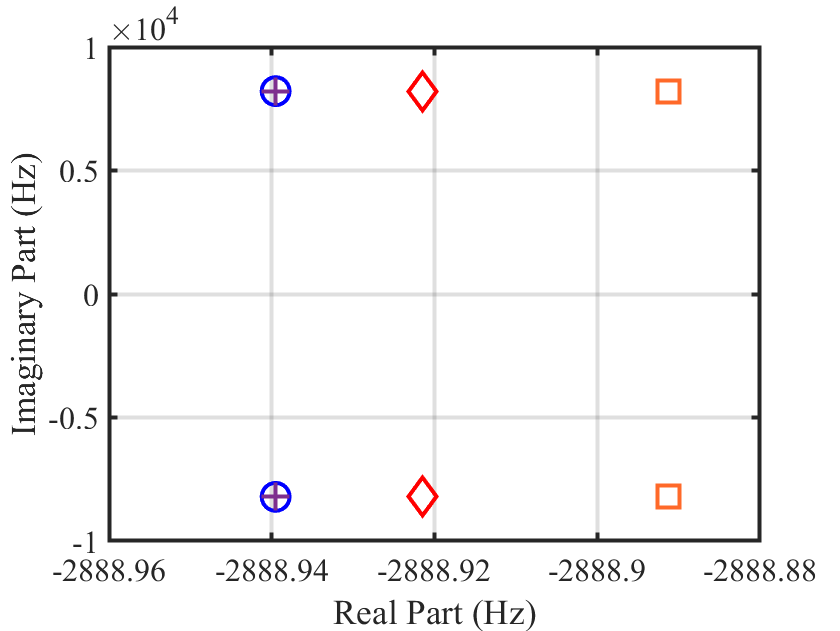}
        \label{fig:eigen1}
    }
    \subfloat[]{
        \includegraphics[width=0.235\textwidth]{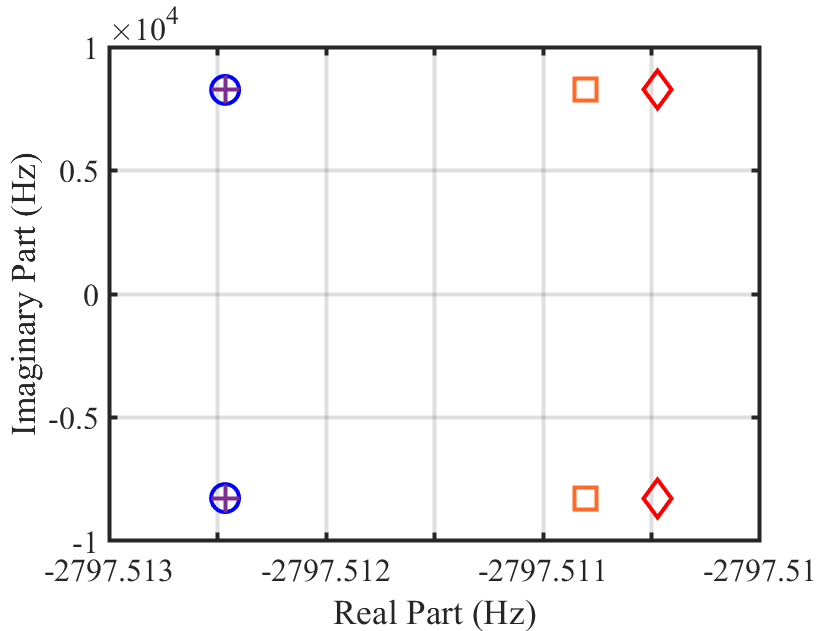}
        \label{fig:eigen2}
    } 
    \subfloat[]{
        \includegraphics[width=0.235\textwidth]{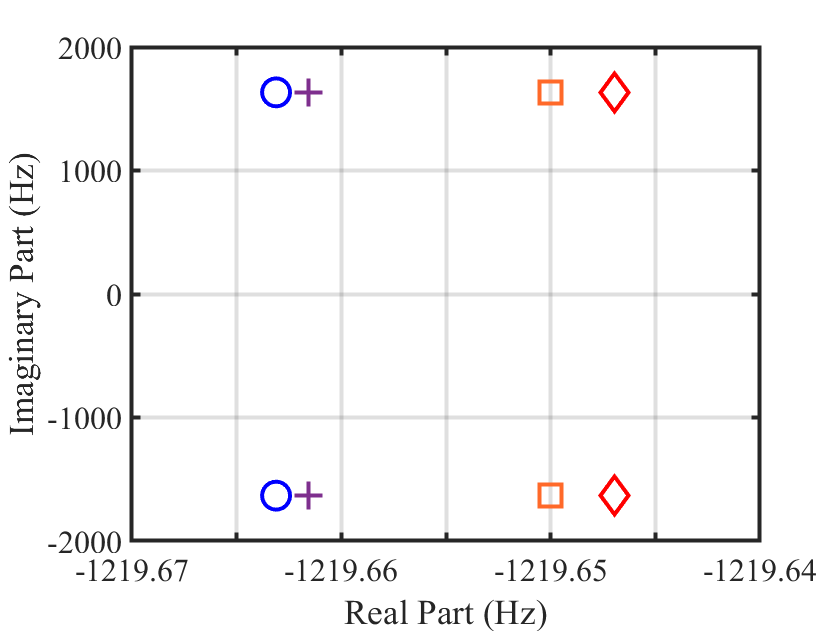}
        \label{fig:eigen3}
    } 
    \subfloat[]{
        \includegraphics[width=0.235\textwidth]{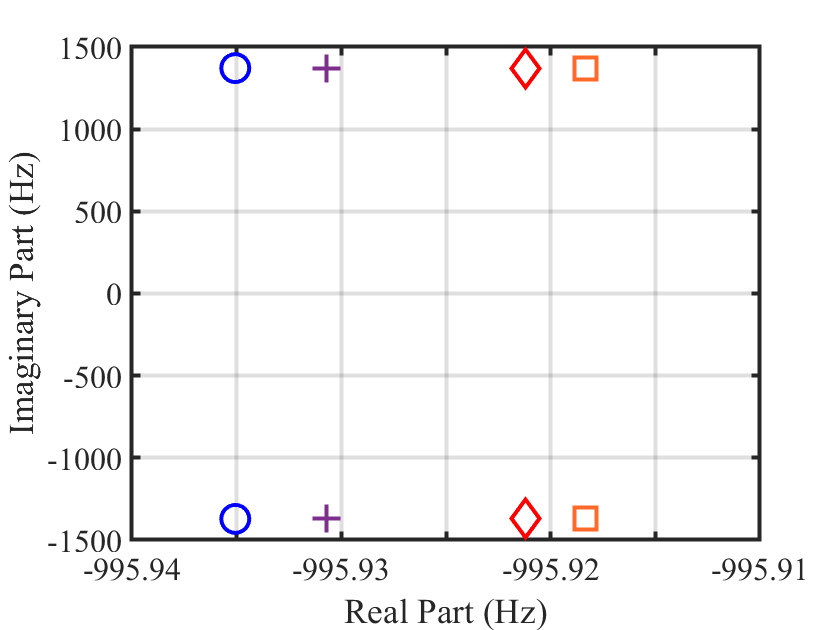}
        \label{fig:eigen4}
    }
    \\
    \subfloat[]{
        \includegraphics[width=0.235\textwidth]{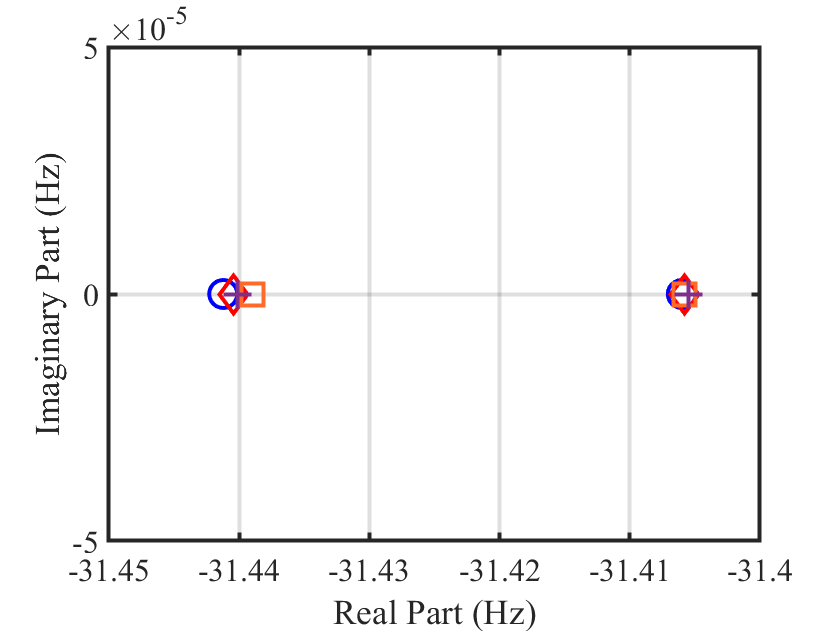}
        \label{fig:eigen5}
    }
    \subfloat[]{
        \includegraphics[width=0.235\textwidth]{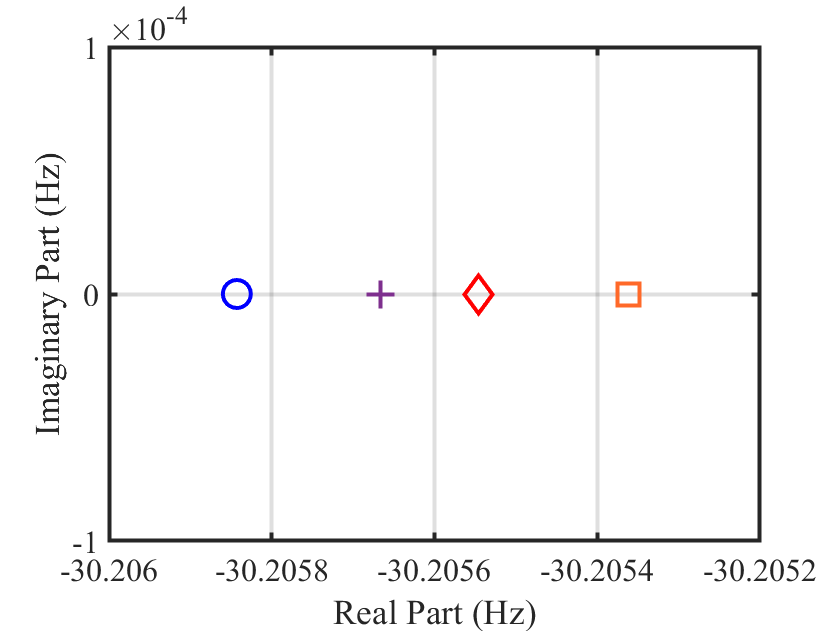}
        \label{fig:eigen6}
    }
    \subfloat[]{
        \includegraphics[width=0.235\textwidth]{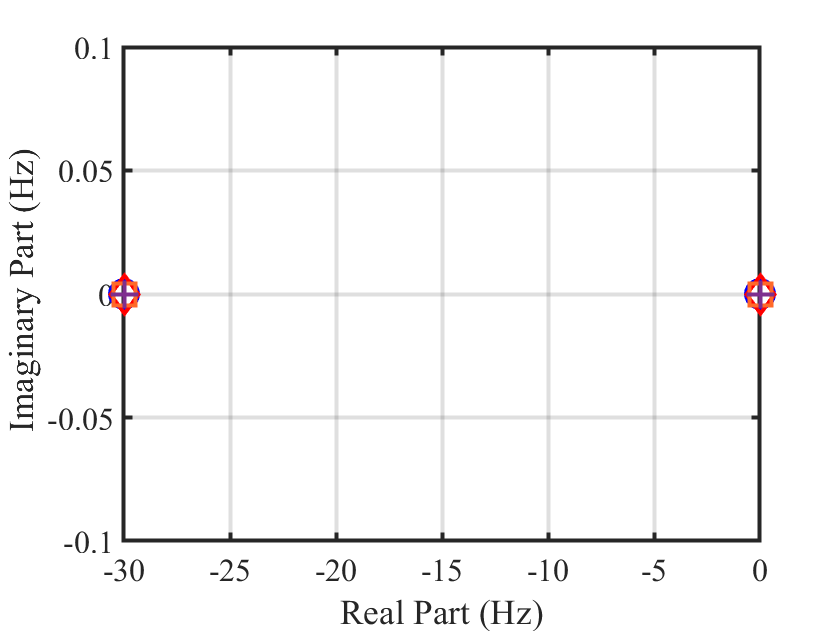}
        \label{fig:eigen7}
    }
    \subfloat[]{
        \includegraphics[width=0.235\textwidth]{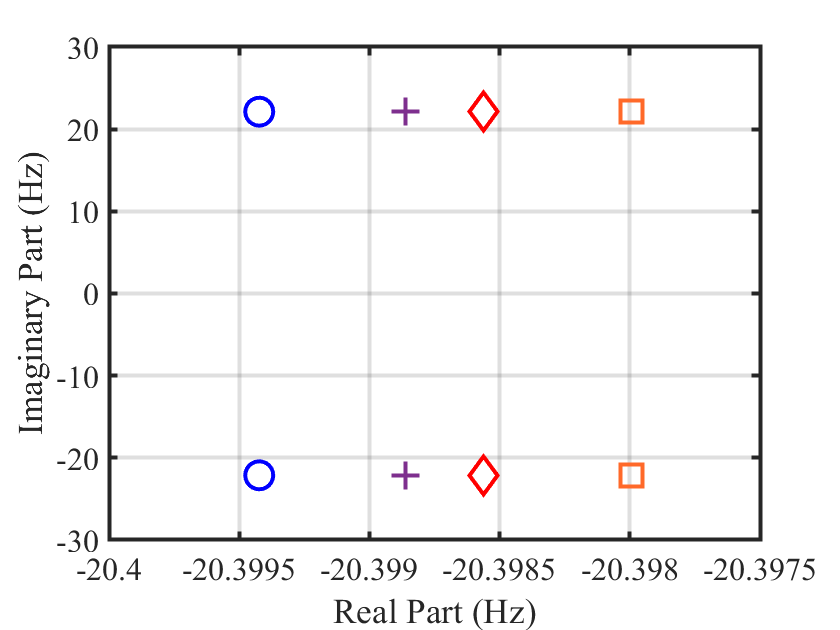}
        \label{fig:eigen8}
    }
    \\
\vspace{-2mm}    
\caption[CR]{MG eigenvalues under the:  \textcolor{blue}{$\circ$} attack-free case, \textcolor{red}{$\Diamond$} stealthy attack scenario, \textcolor{orange}{$\square$} stealthy-intermittent attack scenario, \textcolor{violet}{$+$} stability-constrained mitigation formulation is employed while the stealthy and intermittent attack is active. } 
\vspace{-2mm}
\label{fig:eigen}
\end{figure*}

In this subsection, the mitigation strategy is applied to deal with the attack scenarios. A threshold value is used as an alarm in case the DG inverter is experiencing abnormal operating values for frequency and voltage; in the time instants when one or more components of the residual go beyond the respective threshold $\tau_i$, the DG controllers will use the estimated state values from the observer, instead of the measured values that triggered the residual-based detection condition. Consequently, protection of the MG is effectively guaranteed.

Here, we present three sets of mitigation results following the proposed observer of our work: \textit{(a)} under the effect of the arbitrary attacks (Fig. \ref{fig:mitigation_inputs_random_attack_vector}), \textit{(b)}  under the proposed stealthy and intermittent attack (Fig. \ref{fig:mitigation_inputs_stealthy_attack_vector}), and \textit{(c)} under the proposed stealthy and intermittent attack while complying to the stability-constrained network modeling and overall formulation structure of Section \ref{s:bilevel} ({Fig. \ref{fig:eigen}}). 
According to IEEE 1547.4-2011 guide \cite{5960751}, the MG control strategy necessitates that the MG's voltage and frequency deviations meet the standards acceptable to all parties involved. The MG's ability to maintain a narrow frequency range will determine its effectiveness in load following. This is clearly demonstrated by the results in Figs. \ref{fig:mitigation_inputs_random_attack_vector}--\ref{fig:eigen}. Specifically, after the residual-based observer mitigation is implemented in Figs. \ref{fig:mitigation_inputs_random_attack_vector}--\ref{fig:mitigation_inputs_stealthy_attack_vector}, the effect on both the voltage and frequency for all of the presented scenarios can be considered negligible. The results confirm that the mitigation strategy effectively counters the impact of the attacks, even in the stealthy and intermittent attack case. Furthermore, Fig. \ref{fig:eigen} illustrates the eigenvalues of the MG obtained under four scenarios: the attack-free case, the stealthy attack scenario, the stealthy and intermittent attack scenario, and  after the stability-constrained mitigation formulation is employed while the stealthy and intermittent attack is active. We observe that four of the eigenvalues remain stationary, however, the remaining eleven, in the attack cases, they are  located closer to the imaginary axis which can potentially cause instabilities in the MG in response to load perturbations or other disruptive events. In contrast, when stability and network constraints are employed during the stealthy and intermittent attack, the eigenvalues move  further away from the imaginary axis, providing a better stability margin.
\vspace{-3mm}
\section{Conclusions}\label{s:conclusion}
\vspace{-1mm}
This paper presents a detection and mitigation strategy for the operation of islanded MGs under disruptive and stealthy attacks. It depicts the consequences of such attacks at the output of the secondary controller. The work further develops a theoretical formulation of a detection strategy based on the analysis of the residual error between the actual state of the DGs, and the estimated one, as provided by a nonlinear observer. {Finally, the paper considers both network and stability constraints to ensure secure MG operation, while it} provides simulation-based experiments to demonstrate the effectiveness of the proposed mitigation strategy. The results demonstrate that unexpected disturbances can affect the reference  of each DG provided by the secondary controller. However, an accurate mitigation strategy can significantly increase the effort required by the adversary to affect MG operation. 

\vspace{-3mm}
\section*{Appendix}

The matrix $\boldsymbol{M}^{\mathrm{p}}$ is defined as:
\begin{equation}\label{eq:Mmatrix}
\small
  \boldsymbol{M}^{\mathrm{p}}=\left[\begin{array}{ccccc}
-m_{P_1} & m_{P_2} & 0 & \cdots & 0 \\
-m_{P_1} & 0 & m_{P_3} & \cdots & 0 \\
\vdots & \vdots & \vdots & \ddots & 0 \\
-m_{P_1} & 0 & 0 & \cdots & m_{P_{n^g}} 
\end{array}\right]  
\end{equation}where each $m_{P_i}$ indicates the droop constant of the $i_{th}$ inverter, $\forall i \in \mathcal{G}$.

\begin{definition}[Convexity]
A set $\mathrm{P} \subseteq \mathbb{R}^{\mathrm{d}}$ is convex iff $\sum_{i=1}^n \lambda_i p_i \in \mathrm{P}$, for all $n \in \mathbb{N}$, $p_1, \ldots, p_n \in P$, and $\lambda_1, \ldots, \lambda_n \geqslant 0$ with $\sum_{i=1}^n \lambda_i=1$  
\end{definition}

\begin{definition}[Convex Hull]
The convex hull $\operatorname{convh}(\mathrm{P})$ of a set $\mathrm{P} \subseteq \mathbb{R}^{\mathrm{d}}$ is the intersection of all convex supersets of $\mathrm{P}$.
\end{definition}

\begin{figure*}[t]
 \begin{equation} \label{eq:df-dx}
     \small
    \frac{\partial \mathbf{\tilde{f}}}{\partial \boldsymbol{x}}=\left[\begin{array}{ccc}-\left[\boldsymbol{\omega}_{b}\right] & -\left[\boldsymbol{\omega}_{b} \cdot \boldsymbol{n}_Q \cdot\left(\boldsymbol{i}_{\mathrm{oq}} \cdot \sin \boldsymbol{\delta}+\boldsymbol{i}_{\mathrm{od}} \cdot \cos \boldsymbol{\delta}\right)\right] & -\left[\boldsymbol{\omega}_{b} \cdot\left(\boldsymbol{v}_{\mathrm{ref}}-\boldsymbol{n}_Q \cdot \boldsymbol{q}\right) \cdot\left(\boldsymbol{i}_{\mathrm{od}} \cdot \sin \boldsymbol{\delta}-\boldsymbol{i}_{\mathrm{oq}} \cdot \cos \boldsymbol{\delta}\right)\right] \\ 0^{n^{\mathrm{g}}} & {\left[\boldsymbol{\omega}_{b} \cdot \boldsymbol{n}_{Q} \cdot\left(\boldsymbol{i}_{\mathrm{oq}} \cdot \cos \boldsymbol{\delta}-\boldsymbol{i}_{\mathrm{od}} \cdot \sin \boldsymbol{\delta}\right)-\boldsymbol{\omega}_{b}\right]} & {\left[\boldsymbol{\omega}_{b} \cdot\left(\boldsymbol{v}_{\mathrm{ref}}-\boldsymbol{n}_{Q} \cdot \boldsymbol{q}\right) \cdot\left(\boldsymbol{i}_{\mathrm{od}} \cdot \cos \boldsymbol{\delta}+\boldsymbol{i}_{\mathrm{oq}} \cdot \sin \boldsymbol{\delta}\right)\right]} \\ \boldsymbol{M}^{\mathrm{p}} & 0^{n^{\mathrm{g}}-1} & 0^{n^{\mathrm{g}}-1}\end{array}\right]
 \end{equation}
    \begin{equation}\label{eq:df-dz} 
        \small
\frac{\partial \mathbf{\tilde{f}}}{\partial \boldsymbol{z}}=\left[\begin{array}{cc}{\left[\boldsymbol{\omega}_{b} \cdot\left(\boldsymbol{v}_{\mathrm{ref}}-\boldsymbol{n}_{Q} \cdot \boldsymbol{q}\right) \cdot \cos \delta\right]} & {\left[\boldsymbol{\omega}_{b} \cdot\left(\boldsymbol{v}_{\mathrm{ref}}-\boldsymbol{n}_{Q} \cdot \boldsymbol{q}\right) \cdot \sin \delta\right]} \\ {\left[\boldsymbol{\omega}_{b} \cdot\left(\boldsymbol{v}_{\mathrm{ref}}-\boldsymbol{n}_{Q} \cdot \boldsymbol{q}\right) \cdot \sin \delta\right]} & -\left[\boldsymbol{\omega}_{b} \cdot\left(\boldsymbol{v}_{\mathrm{ref}}-\boldsymbol{n}_{Q} \cdot \boldsymbol{q}\right) \cdot \cos \delta\right] \\ 0^{n \mathrm{~g}-1} & 0^{n \mathrm{~g}-1}\end{array}\right]
    \end{equation}
\begin{equation}\label{eq:dg-dx}
        {\small
 \frac{\partial \mathbf{\tilde{g}}}{\partial \boldsymbol{x}}=\left[\begin{array}{llll}0^{n_{\mathrm{g}}} & \left(\check{\boldsymbol{G}} \cdot\left[\boldsymbol{n}_{Q} \cdot \cos \boldsymbol{\delta}\right]^{n^{\mathrm{g}}}-\check{\boldsymbol{B}} \cdot\left[\boldsymbol{n}_{Q} \cdot \sin \boldsymbol{\delta}\right]^{n^{\mathrm{g}}}\right) & \check{\boldsymbol{G}} \cdot\left[\left(\boldsymbol{v}_{\mathrm{ref}}-\boldsymbol{n}_{Q} \cdot \boldsymbol{q}\right) \cdot \sin \boldsymbol{\delta}\right]^{n^{\mathrm{g}}}+\check{\boldsymbol{B}}\left[\left(\boldsymbol{v}_{\mathrm{ref}}-\boldsymbol{n}_{Q} \cdot \boldsymbol{q}\right) \cdot \cos \boldsymbol{\delta}\right]^{n^{\mathrm{g}}} \\ 0^{n_{\mathrm{g}}} & \left(\check{\boldsymbol{G}} \cdot\left[\boldsymbol{n}_{Q} \cdot \sin \boldsymbol{\delta}\right]^{n^{\mathrm{g}}}+\check{\boldsymbol{B}} \cdot\left[\boldsymbol{n}_{Q} \cdot \cos \boldsymbol{\delta}\right]^{n^{\mathrm{g}}}\right) & \check{\boldsymbol{B}} \cdot\left[\left(\boldsymbol{v}_{\mathrm{ref}}-\boldsymbol{n}_{Q} \cdot \boldsymbol{q}\right) \cdot \sin \boldsymbol{\delta}\right]^{n^{\mathrm{g}}}-\check{\boldsymbol{G}}\left[\left(\boldsymbol{v}_{\mathrm{ref}}-\boldsymbol{n}_{Q} \cdot \boldsymbol{q}\right) \cdot \cos \boldsymbol{\delta}\right]^{n^{\mathrm{g}}}\end{array}\right]   }
    \end{equation}

\begin{equation}\label{eq:dg-dz} {\small
 \frac{\partial \mathbf{\tilde{g}}}{\partial \boldsymbol{z}}=\left[\begin{array}{cc}I^{n^{\mathrm{g}}}+\check{\boldsymbol{G}} \cdot\left[\boldsymbol{r}_{\boldsymbol{c}}\right]^{n^{\mathrm{g}}}-\check{\boldsymbol{B}} \cdot\left[\boldsymbol{x}_{\boldsymbol{c}}\right]^{\mathrm{g}} & -\check{\boldsymbol{G}} \cdot\left[\boldsymbol{x}_{\boldsymbol{c}}\right]^{n^{\mathrm{g}}}-\check{\boldsymbol{B}} \cdot\left[\boldsymbol{r}_{\boldsymbol{c}}\right]^{n^{\mathrm{g}}} \\ \check{\boldsymbol{B}} \cdot\left[\boldsymbol{r}_{\boldsymbol{c}}\right]^{n^{\mathrm{g}}}+\check{\boldsymbol{G}} \cdot\left[\boldsymbol{x}_{\boldsymbol{c}}\right]^{n^{\mathrm{g}}} & I^{n^{\mathrm{g}}}+\check{\boldsymbol{G}} \cdot\left[\boldsymbol{r}_{\boldsymbol{c}}\right]^{n^{\mathrm{g}}}-\check{\boldsymbol{B}} \cdot\left[\boldsymbol{x}_{\boldsymbol{c}}\right]^{n^{\mathrm{g}}}\end{array}\right] }
\end{equation}
\end{figure*}

\vspace{-5mm}
\bibliographystyle{IEEEtran} 
\bibliography{biblio}
\end{document}